\documentclass[11pt,letterpaper]{article}
\usepackage[margin=1in]{geometry}
\usepackage[utf8]{inputenc}
\usepackage{amsmath,amssymb,amsthm,mathrsfs}
\usepackage{physics}
\usepackage{complexity}
\usepackage{caption,subcaption}
\usepackage{algorithm}
\usepackage{algpseudocode}
\usepackage{setspace}
\usepackage{graphicx}
\usepackage[colorlinks=true]{hyperref}
\usepackage[dvipsnames]{xcolor}
\definecolor{darkred}  {rgb}{0.5,0,0}
\definecolor{darkblue} {rgb}{0,0,0.5}
\definecolor{darkgreen}{rgb}{0,0.5,0}
\hypersetup{
urlcolor      = blue,         
linkcolor     = darkblue,     
citecolor     = darkgreen,    
filecolor     = darkred,
linkcolor     = darkblue
}
\usepackage[nameinlink, capitalize]{cleveref}
\usepackage{tcolorbox}
\tcbuselibrary{breakable}
\usepackage{caption,subcaption}
\usepackage{tikz-cd}
\usetikzlibrary{decorations.pathreplacing, positioning,calc,arrows.meta,shapes.geometric}
\pgfmathsetmacro\MathAxis{height("$\vcenter{}$")}

\newcommand{\mc}{\mathcal}

\renewcommand{\polylog}{\mathrm{polylog}}
\renewcommand{\poly}{\mathrm{poly}}
\renewcommand{\epsilon}{\varepsilon}
\renewcommand{\ln}{\log}
\newcommand{\sA}{\mathsf{A}}
\newcommand{\sB}{\mathsf{B}}
\newcommand{\sC}{\mathsf{C}}

\usepackage[backend=biber,style=alphabetic,url=false,isbn=false,maxnames=10,minnames=3,maxalphanames=4,minalphanames=3]{biblatex}
\usepackage[toc,page]{appendix}

\newtheorem{theorem}{Theorem}[section]
\newtheorem*{theorem*}{Theorem}
\newtheorem*{question*}{Question}

\newtheorem{lemma}[theorem]{Lemma}
\newtheorem{claim}[theorem]{Claim}
\newtheorem{fact}[theorem]{Fact}
\newtheorem{corollary}[theorem]{Corollary}
\theoremstyle{definition}
\newtheorem{definition}[theorem]{Definition}
\newtheorem{remark}[theorem]{Remark}
\newtheorem{properties}[theorem]{Properties}
\definecolor{DarkRed}{RGB}{170,0,0}
\newtheorem{example}[theorem]{Example}

\numberwithin{equation}{section}

\usepackage{enumerate}
\newcommand{\under}[2]{\underbrace{#1}_{\substack{#2}}}

\begin{document}

\title{Rapid mixing for Gibbs states within a logical sector:\\ a dynamical view of self-correcting quantum memories}

\author{Thiago Bergamaschi\thanks{UC Berkeley. \href{mailto:thiagob@berkeley.edu}{thiagob@berkeley.edu}}
\and
Reza Gheissari\thanks{Northwestern University. \href{mailto:gheissari@northwestern.edu}{gheissari@northwestern.edu}}
\and
Yunchao Liu\thanks{Harvard University. \href{mailto:yunchaoliu@berkeley.edu}{yunchaoliu@berkeley.edu}}
}

\date{}
\maketitle

\begin{abstract}
Self-correcting quantum memories store logical quantum information for exponential time in thermal equilibrium at low temperatures. By definition, these systems are slow mixing. This raises the question of how the memory state, which we refer to as the Gibbs state within a logical sector, is created in the first place. 

In this paper, we show that for a broad class of self-correcting quantum memories on lattices with parity check redundancies, a quasi-local quantum Gibbs sampler rapidly converges to the corresponding low-temperature Gibbs state within a logical sector when initialized from a ground state. This illustrates a dynamical view of self-correcting quantum memories, where the ``syndrome sector'' rapidly converges to thermal equilibrium, while the ``logical sector'' remains metastable. As a key application, when initialized from a random ground state, this gives a rapid Gibbs state preparation algorithm for the 4D toric code in $\mathrm{polylog}(n)$ depth.
The main technical ingredients behind our approach are new, low-temperature decay-of-correlation properties for these metastable states.
\end{abstract}

\setcounter{tocdepth}{1}
\tableofcontents
\newpage

\section{Introduction}
A self-correcting quantum memory provides a way of storing quantum information in a quantum state that remains stable under thermal noise, even without active error correction. For example, the well-known 4D toric code~\cite{Dennis2001TopologicalQM,Alicki2008OnTS} gives such an encoding when the temperature is below a constant threshold. The stability of the memory state comes from the fact that it reaches some form of thermal equilibrium: if the Gibbs state is understood as the thermal perturbation of a uniformly random ground state (which encodes no logical information), then the memory state can be understood as the same thermal perturbation of a specific ground state, which we refer to as the \emph{Gibbs state within a logical sector}. The defining feature of these states is that they are \textit{metastable} at low temperatures. That is, the system is approximately stationary under thermal noise (modeled as quantum Glauber dynamics \cite{Davies1974MarkovianME}) for an exponentially long time.

This naturally raises the question: \emph{how are these memory states prepared in the first place?} From an algorithmic perspective, the stability of the memory implies that the quantum Glauber dynamics corresponds to a slow-mixing Markov chain, which could be viewed as a barrier to the efficient preparation of memory states.  On the other hand, physics intuition suggests that there is a natural means of state preparation: first encode the logical information into a ground state (a code-state), then simply run Glauber dynamics starting from that ground state. The underlying intuition is that if the only obstacle to mixing lies in the bottlenecks \emph{between} the energy wells surrounding the ground states, then the system should rapidly reach some form of thermal equilibrium \textit{within} the wells (see \cref{fig:memorydynamics}). That is, the system is expected to quickly converge to the corresponding Gibbs state within a logical sector, which then remains metastable.

This approach can be viewed as circumventing the slow-mixing barrier by starting from a good initial state, which is challenging to establish rigorously due to the breakdown of worst-case mixing time techniques. Recently, there has been significant effort to prove rapid mixing in this paradigm in the classical Markov chain literature~\cite{Chen2019FastAA,GS23Ising,Bresler2024Metastable,galanis2024plantingmcmcsamplingpotts,huang2024weakpoincareinequalitiessimulated,GSS25}. 
Closest to the setting of this paper,~\cite{GS23Ising} formalized this approach in the 2D Ising model, where they proved that when initialized from the all-zero configuration, the system rapidly converges to the restriction of the Gibbs state with majority $0$\footnote{For the Ising model, the Gibbs state within a logical sector can be simply understood as the true Gibbs state conditioned on the majority spin being zero. The quantum version is more delicate and is discussed in \cref{sec:gibbslogical}.}. A natural question is whether one can generalize these results to the quantum setting, showing analogous rapid mixing within a logical sector for self-correcting quantum memories such as the 4D toric code.

However, such an attempt immediately runs into a fundamental barrier: unlike the Ising Glauber dynamics, quantum error correcting codes lack a crucial property called \emph{monotonicity}, which says that the process initialized in a ground state always has stochastically fewer errors (bit flips) than that initialized in the Gibbs state. 
A crucial consequence of monotonicity is that the \textit{memory time} -- the time until the logical information is lost -- starting from the ground state, is provably larger than that of the metastable memory state.  
In other words, the evolution starting from the ground state must go through the memory state first, before converging to the Gibbs state.
In the absence of monotonicity, even proving the \emph{ground state stability} for the 4D toric code is still an open problem: current techniques are insufficient to rigorously rule out the possibility that running Glauber dynamics from the ground state quickly loses the logical information (see \cref{section:discussion}).

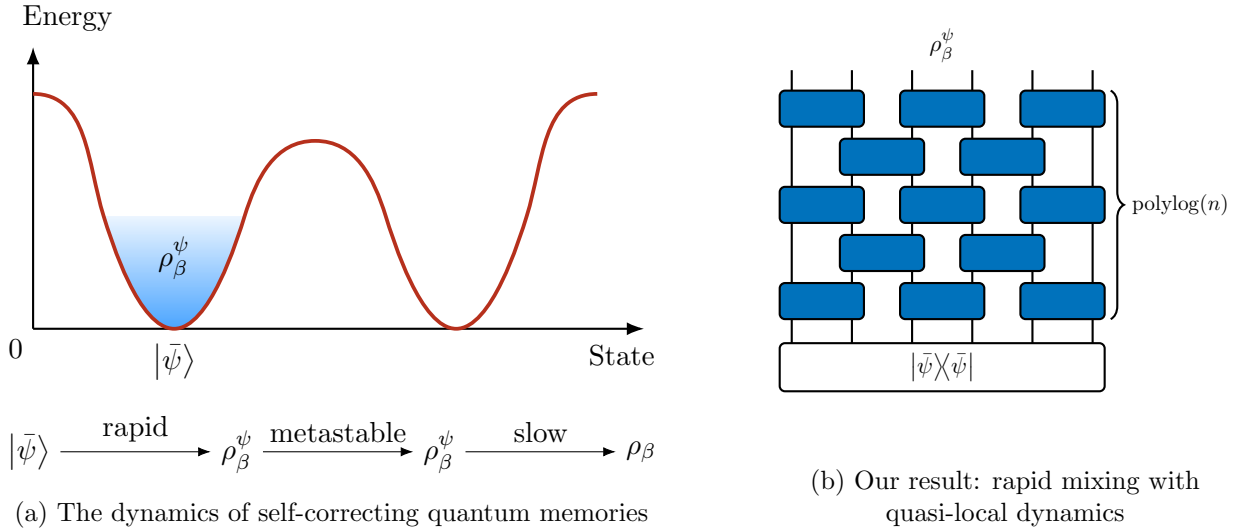
\begin{figure}[t]
\centering

\begin{subfigure}[b]{0.55\textwidth}
\centering

\begin{tikzpicture}[
    scale=1.25, %
    axis/.style={->, >=Latex, thick}
]

\draw[axis] (0,0) -- (6.5,0) node[below left=2pt and -8pt] {State};
\draw[axis] (0,0) -- (0,3) node[above right=2pt and -8pt] {Energy};

\def\wellbottom{0} %
\def\barrierheight{2.0} %
\def\entryexitheight{2.5} %
\def\slopecontrolheight{1.2} %

\def\wellonecenter{1.5}
\def\welltwocenter{4.5}
\def\barriercenter{3}
\def\wellwidth{0.75} %
\def\curvetension{0.9} %

\pgfmathsetmacro{\parabolaSlope}{2*(\slopecontrolheight-\wellbottom)/\wellwidth}
\pgfmathsetmacro{\parabolaAngle}{atan(\parabolaSlope)}

\begin{scope}
    \path[clip]
        (\wellonecenter - \wellwidth, \slopecontrolheight)  %
        parabola bend (\wellonecenter, \wellbottom)         %
        (\wellonecenter + \wellwidth, \slopecontrolheight)    %
        -- cycle; %

    \shade[
        bottom color=RoyalBlue, %
        top color=RoyalBlue!10,          %
        opacity=0.7
    ]
        (\wellonecenter - \wellwidth, \wellbottom) %
        rectangle
        (\wellonecenter + \wellwidth, \slopecontrolheight); %
\end{scope}

\draw[draw=BrickRed, line width=0.5mm]
    (0, \entryexitheight)
    to[out=0, in=180-\parabolaAngle, tension=\curvetension] 
    (\wellonecenter - \wellwidth, \slopecontrolheight)
    
    parabola bend (\wellonecenter, \wellbottom) (\wellonecenter + \wellwidth, \slopecontrolheight)
    
    to[out=\parabolaAngle, in=180, tension=\curvetension] 
    (\barriercenter, \barrierheight)
    
    to[out=0, in=180-\parabolaAngle, tension=\curvetension] 
    (\welltwocenter - \wellwidth, \slopecontrolheight)
    
    parabola bend (\welltwocenter, \wellbottom) (\welltwocenter + \wellwidth, \slopecontrolheight)
    
    to[out=\parabolaAngle, in=180, tension=\curvetension] 
    (6, \entryexitheight);

\node[below=2pt, text=black] at (\wellonecenter, \wellbottom) {$\ket{\bar \psi}$};
\node[above=15pt, text=black] at (\wellonecenter, \wellbottom) {$\rho_\beta^\psi$};

\node[below left] at (0,0) {$0$};

\end{tikzpicture}\\[2mm]

\begin{tikzpicture}[>=Latex, font=\normalsize, align=center]

\node (start) {$\ket{\bar\psi}$};
\node[right=2cm of start] (mid1) {$\rho_\beta^\psi$};
\node[right=2cm of mid1] (mid2) {$\rho_\beta^\psi$};
\node[right=2cm of mid2] (end) {$\rho_\beta$};

\draw[->] (start) -- (mid1) node[midway, above] {rapid};
\draw[->] (mid1) -- (mid2) node[midway, above] {metastable};
\draw[->] (mid2) -- (end) node[midway, above] {slow};

\end{tikzpicture}

\caption{The dynamics of self-correcting quantum memories}
\label{fig:memorydynamics}
\end{subfigure}
\hfill
\begin{subfigure}[b]{0.365\textwidth}
\centering
\begin{tikzpicture}[
    scale=0.8,
    transform shape,
    gate/.style={
        draw,
        thick,
        fill=RoyalBlue!100,
        minimum width=1.4cm,
        minimum height=0.6cm,
        rounded corners=2pt
    }
]

\def\qubitsep{1.0}   %
\def\layersep{0.8}   %

\foreach \q in {0,...,5} {
    \draw[thick] (\q*\qubitsep, 0) -- (\q*\qubitsep, 5.8*\layersep);
}

\node[
    draw, thick, fill=white, 
    minimum width=5.4cm,  %
    minimum height=0.8cm,       %
    rounded corners=2pt
] at (2.5*\qubitsep, -0.3cm) {$\ketbra{\bar\psi}$}; %

\node[above] at (2.5*\qubitsep, 5.8*\layersep) {$\rho_\beta^\psi$};

\foreach \q [evaluate=\q as \qnext using int(\q+1)] in {0,2,4} {
    \node[gate] at ($(\q*\qubitsep, 1*\layersep) !0.5! (\qnext*\qubitsep, 1*\layersep)$) {};
}

\foreach \q [evaluate=\q as \qnext using int(\q+1)] in {1,3} {
    \node[gate] at ($(\q*\qubitsep, 2*\layersep) !0.5! (\qnext*\qubitsep, 2*\layersep)$) {};
}

\foreach \q [evaluate=\q as \qnext using int(\q+1)] in {0,2,4} {
    \node[gate] at ($(\q*\qubitsep, 3*\layersep) !0.5! (\qnext*\qubitsep, 3*\layersep)$) {};
}

\foreach \q [evaluate=\q as \qnext using int(\q+1)] in {1,3} {
    \node[gate] at ($(\q*\qubitsep, 4*\layersep) !0.5! (\qnext*\qubitsep, 4*\layersep)$) {};
}

\foreach \q [evaluate=\q as \qnext using int(\q+1)] in {0,2,4} {
    \node[gate] at ($(\q*\qubitsep, 5*\layersep) !0.5! (\qnext*\qubitsep, 5*\layersep)$) {};
}

\draw [decorate,decoration={brace,amplitude=5pt,mirror},thick]
  (5.3,\layersep - 0.3) -- (5.3,5*\layersep + 0.3) node[midway,right=0.1]{\small{ $\polylog(n)$}};
\end{tikzpicture}
\vspace{0.3cm}

\captionsetup{justification=centering}
\caption{Our result: rapid mixing with quasi-local dynamics}
\label{fig:rapidmixing}
\end{subfigure}
\caption{Rapid mixing for Gibbs states within a logical sector. (a) Quantum Glauber dynamics starting from a ground state $\ket{\bar\psi}$ (bottom of well) is expected to converge rapidly to the corresponding Gibbs state within a logical sector $\rho_\beta^\psi$ (blue). This state then remains metastable for exponential time, until eventually converging to the true Gibbs state $\rho_\beta$. (b) We prove rapid mixing within a logical sector for a broad class of self-correcting quantum memories on lattices, where each step of the Markov chain is a quasi-local quantum channel acting on $\polylog(n)$ qubits (blue box). %
}
\label{fig:energywell}
\end{figure}

In this paper, we make progress on these questions. We prove that at sufficiently low (but constant) temperatures, a broad class of self-correcting quantum memories on lattices -- namely, those under the \textit{connectedness criterion} for self-correction of \cite{Bombin2009SelfcorrectingQC} -- admit a \textit{quasi-local} quantum Markov chain which rapidly converges to the Gibbs state within a logical sector, when appropriately initialized from a ground state (see \cref{fig:rapidmixing}). This provides a theoretical justification for the intuitive understanding of the memory dynamics shown in \cref{fig:memorydynamics}. Our main result, when specialized to the 4D toric code, is the following theorem. 

\begin{theorem}
    [Rapid Mixing within a Logical Sector, Informal]\label{theorem:maininformal}
    There exists a constant threshold temperature, below which the Gibbs state within a logical sector of the 4D toric code on $n$ qubits can be prepared in $\polylog (n)$ time, using a continuous time and geometrically local quantum Markov chain with jumps of locality $\polylog (n)$, when initialized from a ground state. Moreover, this process can be simulated by a $\polylog (n)$ depth quantum circuit of 2-qubit gates. 
\end{theorem}

We add that our result also implies efficient Gibbs state preparation for the 4D toric code in $\polylog(n)$ depth, by running our algorithm starting from a \emph{random} ground state. To the best of our knowledge, this is the fastest state-preparation algorithm for both the 4D toric code memory state and its low temperature Gibbs state, as measured by the quantum circuit depth.\footnote{Recall that ground states (the code-states) can be prepared in $O(\log n)$ depth starting from a product state with all-to-all connectivity; though polynomial depth is necessary on a 4D geometry, e.g.~\cite{Konig2014Generating}. Our algorithm (which starts from a ground state) is geometrically local in 4D.} 

The more general version of our main result is stated in  \cref{theorem:cbd_intro}, which shows rapid mixing under a quasi-local quantum Markov chain for error-correcting codes with certain properties adapted from~\cite{Bombin2009SelfcorrectingQC}.\footnote{As stated in \cref{theorem:maininformal}, for the 4D toric code our quasi-local updates can be additionally realized using efficient quantum circuits. This may not be true in general, but one can at least brute-force simulate the quasi-local updates, resulting in a quasi-polynomial time state-preparation algorithm.} 
In fact, our result also applies to \emph{classical} codes with these properties. Roughly speaking, we require the code to have parity check redundancies (i.e. metachecks) such that local clusters of syndromes can be corrected locally. These properties are satisfied by many known examples of self-correcting classical or quantum memories on finite-dimensional lattices, including high-dimensional toric and color codes, as well as the Ising and Potts models.

Our quasi-local quantum Markov chain is largely inspired by local decoders for the 4D toric code \cite{hastings2013decodinghyperbolicspacesldpc,Zhou2025FinitetemperatureQT}, and is designed in a way such that the dynamics is guaranteed to maintain the logical information (\cref{fig:di_dynamics}). It remains an interesting open question to improve our result to fully local (single-site) quantum Glauber dynamics. As discussed earlier, due to the lack of monotonicity, even the more basic question of ground state stability for the 4D toric code under Glauber dynamics is open.

\subsection{Related Work}
\label{section:related}

\paragraph{Quantum Gibbs sampling in the low-temperature regime.} Bounding the mixing time -- i.e., the time to converge to the Gibbs state, from a worst-case initial state -- is a central challenge in efficient quantum Gibbs sampling. Most existing bounds work in the \emph{high-temperature} regime, e.g.~\cite{KB16, capel2021modified, RFA24}. However, high-temperature quantum Gibbs states are essentially classical~\cite{yin2023polynomialtime,BLMT24}, which shifts the primary algorithmic challenge to the low-temperature regime where uniquely quantum features dominate. %

In the \emph{low-temperature} regime, existing fast-mixing guarantees are only for highly structured systems, e.g.~\cite{Bardet2023Rapid,BCL24, DLLZ24}, where the bounds in fact hold at all temperatures. The more fundamental challenge in this regime is the \emph{slow-mixing barrier}: many quantum systems, such as self-correcting memories, are known to have exponentially long mixing times (see also~\cite{Gamarnik2024SlowMO,rakovszky2024bottlenecksquantumchannelsfinite}). Consequently, developing efficient algorithms requires more refined methods beyond standard worst-case mixing-time techniques. Here we generalize the classical approach of ``running Glauber dynamics from a good initial state'' to the quantum setting, and derive new rapid mixing results in low-temperature quantum systems which mix slowly in general.

\paragraph{Self-correcting quantum memories.} The development of quantum many-body systems which are passively robust to thermal noise has been a central line of work in quantum error correction and condensed matter physics (see \cite{Brown_2016} for a review). Here we focus on Calderbank-Shor-Steane (CSS) codes on finite-dimensional lattices with sub-exponentially long memory times \cite{Alicki2008OnTS, Bombin2009SelfcorrectingQC}. Our results give a new theoretical basis for the understanding of the dynamical behavior of these systems (\cref{fig:memorydynamics}), and offer new insights into quantum phases of matter (see \cref{section:discussion}).

There are a number of models not captured by our framework. For instance, quantum memories based on quantum expander codes (c.f.~\cite{Hong2024QuantumMA, placke2024topologicalquantumspinglass}) are geometrically nonlocal and have no metachecks (see \cref{section:discussion}). Haah's 3D cubic code is shown to have non-trivial memory time up to a finite (\textit{temperature-dependent}) system size, and to have ground state stability within the memory time~\cite{haah11, Bravyi_2013, CC2012}. Finally, there are other models of stable quantum memory, including non-stabilizer codes~\cite{Chesi_2010, hsin2024nonabelianselfcorrectingquantummemory} and circuit models~\cite{Aharonov-Ben-Or, BDL24, GB25}.

\subsection{Our Approach}
\label{section:contributions}
Broadly speaking, our strategy is based on the philosophy that ``decay-of-correlation implies fast mixing'', a standard paradigm in the classical \cite{MO1, MO2, Dyer2002MixingIT} and quantum literature \cite{KB16,Brando2016FiniteCL} that decay of correlations in the Gibbs state (i.e. spatial mixing) implies rapid convergence of Glauber dynamics. Of course, at low temperatures the Gibbs state of the 4D toric code (much like the 2D Ising model) does not exhibit any standard notion of decay of correlations. Nevertheless, our plan is to first identify some limited form of correlation decay at low temperatures for the Gibbs state within a logical sector (\cref{section:decay_corr_logical}), and then show that this form of correlation decay is sufficient to imply rapid mixing for a modified heat-bath block dynamics, inspired by parallel decoders for the 4D toric code (\cref{section:algorithm}).

\subsubsection{Decay of Correlations within a Logical Sector}
\label{section:decay_corr_logical}

What is the right notion of decay of correlations for the Gibbs state within a logical sector? Our starting point is to draw inspiration from the classical literature, where it is known that some more limited form of correlation decay does hold for the Ising model when restricted to a single logical sector \cite{DuminilCopin2018ExponentialDO}. Recently \cite{GS23Ising} introduced a specific form of this correlation-decay which they called \textit{weak spatial mixing} (\textsf{WSM}) \textit{within a phase}\footnote{The term ``phase'' has the same meaning as ``logical sector'' used in this paper.}  -- which quantifies the decay of point-to-set, conditional correlations under ``ground-state-like'' (or monotone) boundary conditions (\cref{fig:wsm}) -- and showed that this is sufficient for rapid mixing within a logical sector in the Ising model.

Unfortunately, there are two issues when trying to do something similar for quantum Gibbs states: first, unlike classical models where one can study conditional correlations by pinning bits on the boundary, one cannot pin qubits in this manner; second, \textsf{WSM} (within a logical sector) alone is insufficient to establish mixing even in non-monotone classical models,\footnote{This limitation is also known to be a major challenge in the classical literature, where the analogous result of fast mixing in a logical sector from a ground state is not known for the 2D Potts model.} thus we have to formulate and prove a stronger spatial mixing condition. We show how to address these issues below.

\paragraph{The Gibbs distribution over syndromes.} The first key insight is to separate the logical sector from the syndrome sector of the Hilbert space, which allows us to formulate decay-of-correlation properties using just the distribution over syndromes, which is purely classical. 
At any inverse temperature $\beta$, the Gibbs state $\rho_\beta$ of a quantum CSS code is isometric to a factorization over the logical and syndrome degrees of freedom:
\begin{equation}
    \rho_\beta \cong \mathbb{I} \otimes \pi_\beta\,,
\end{equation}
where $\mathbb{I}$ is the maximally mixed logical state, and $\pi_\beta$ is the classical Gibbs distribution over the code syndromes, which captures the energy excitations. In this notation, if we were to fix a logical state $\psi$, the Gibbs state \textit{within the logical sector} $\psi$ can be written as
\begin{equation}\label{eq:gibbswithinlogicalsector}
    \rho_\beta^\psi\cong \psi\otimes \pi_\beta\,.
\end{equation}
For a classical or quantum code with syndrome space $\Omega\subseteq \{0,1\}^m$, we have $\pi_\beta(s)\propto e^{-\beta |s|}$, $\forall s\in\Omega$.

\begin{example}
    In the 2D Ising model, $\pi_\beta$ is the distribution over domain walls (loops separating regions of $0/1$ spins) and $\psi$ is the spin assignment of the majority. %
\end{example}

\begin{example}
    In the 4D toric code, $\pi_\beta$ is a distribution over even subgraphs (together with some topological constraints) of $(\mathbb Z/w\mathbb Z)^4$, see \cref{section:toric_background}.
\end{example}

The above decomposition leads to the following perspective: suppose we initialize from a ground state of a classical or quantum code, which determines the logical sector and has an all-zero syndrome configuration. Proving rapid convergence to the Gibbs state within the logical sector can be achieved using the following two steps:
\begin{enumerate}
    \item Prove that the induced dynamics on the syndromes thermalizes quickly to $\pi_\beta$, starting from the all-zero syndrome configuration.
    \item Show that the logical information is stable throughout the dynamics.
\end{enumerate}
To realize this plan, we first formulate and prove new decay-of-correlation properties for $\pi_\beta$ (see below), and then use these properties to prove rapid convergence for a carefully designed quasi-local Gibbs sampler which maintains the stability of the logical sector (\cref{section:algorithm}).

\newcommand{\tikzboxscale}{0.8}
\begin{figure}[t]
\centering
\begin{subfigure}[b]{0.48\textwidth}
\centering

    \begin{equation*}
        \scalebox{\tikzboxscale}{\begin{tikzpicture}[baseline=0cm]

    \node (mybox) [draw=RoyalBlue, line width=1mm, rectangle, minimum width=4cm, minimum height=4cm] at (0,0) {};

    \node[below=1cm] at (mybox.center) {$\mathsf{B}$};

    \node (mycircle) [draw, circle, minimum size=1cm, fill=Gray!20] at (mybox.center) {$\mathsf{A}$};

\end{tikzpicture}}
   \quad\approx\quad
    \scalebox{\tikzboxscale}{\begin{tikzpicture}[baseline=0cm]

    \node (mybox) [draw, dashed, rectangle, minimum width=4cm, minimum height=4cm] at (0,0) {};

    \node[below=1cm] at (mybox.center) {$\mathsf{B}$};

    \node (mycircle) [draw, circle, minimum size=1cm, fill=Gray!20] at (mybox.center) {$\mathsf{A}$};
\end{tikzpicture}}
\end{equation*}
\caption{Weak spatial mixing within a logical sector}
\label{fig:wsm}
\end{subfigure}
\hfill
\begin{subfigure}[b]{0.48\textwidth}
\centering

    \begin{equation*}
        \scalebox{\tikzboxscale}{\begin{tikzpicture}[baseline=0cm]

    \node (mybox) [draw=RoyalBlue, line width=1mm, rectangle, minimum width=4cm, minimum height=4cm] at (0,0) {};

    \node[below=1cm] at (mybox.center) {$\mathsf{B}$};

    \node (mycircle) [draw, circle, minimum size=1cm, fill=Gray!20] at (mybox.center) {$\mathsf{A}$};
\end{tikzpicture}}
\quad\approx\quad
\scalebox{\tikzboxscale}{\begin{tikzpicture}[baseline=0cm]

    \draw[RoyalBlue, line width=1mm] (2,0.5) -- (2,2) -- (-2,2) -- (-2,-2) -- (2,-2) -- (2,-0.5);
    \draw[dashed] (2,0.5) -- (2,-0.5);
    \draw[RoyalBlue, line width=1mm] (2,-0.55) arc [start angle=270, end angle=90, radius=0.55];

    \node[left] at (2,0) {$\mathsf{C}$};
    \node[below=1cm] at (0,0) {$\mathsf{B}$};
    \node (mycircle) [draw, circle, minimum size=1cm, fill=Gray!20] at (mybox.center) {$\mathsf{A}$};
\end{tikzpicture}}
\end{equation*}
\caption{Domain strong spatial mixing}
\label{fig:ssm}
\end{subfigure}
\caption{Decay-of-correlations for the syndrome Gibbs distribution. (a) \textsf{WSM} \textit{within a logical sector} compares two marginal distributions on region \textsf{A}: the conditional syndrome Gibbs distribution where the boundary syndromes are pinned to $0$ (blue), and the true syndrome Gibbs distribution. (b) \textsf{Domain SSM} compares two distinct conditional distributions on \textsf{A} where different boundaries are pinned to $0$: either the original box $\partial \mathsf{B}$ is pinned to $0$, or a  locally deformed boundary $\partial(\mathsf{B}\setminus \mathsf{C})$.}
\label{fig:decayofcorrelation}
\end{figure}
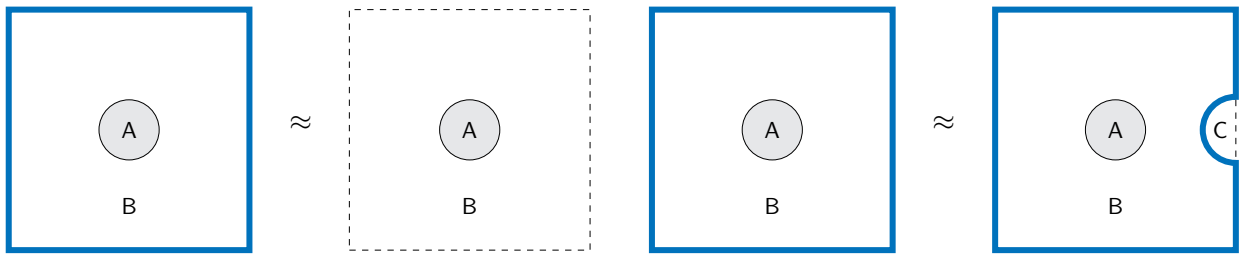

\paragraph{Domain strong-spatial-mixing for the syndrome Gibbs distribution.}
We introduce a new low-temperature correlation-decay property for the Gibbs distribution over syndromes $\pi_\beta$ of self-correcting classical or quantum memories, which we refer to as \textit{domain strong-spatial-mixing} (\textsf{Domain SSM}, \cref{fig:ssm}) and could be of independent interest. 

In the context of classical spin systems, the definition of strong-spatial-mixing (\textsf{SSM}, e.g.~\cite[Def. 2.2]{Dyer2002MixingIT}) quantifies the sensitivity of sets of spins to (perturbations of) their boundary conditions. In more detail, a distribution $\rho$ over spins $x\in \{0, 1\}^n$ is said to satisfy \textsf{SSM} if the marginal distribution on a region $\mathsf{A}\subset [m]$ (within a larger region \textsf{B} with boundary $\partial \mathsf{B}$) is approximately the same under two (similar) boundary conditions $\tau, \sigma\in \{0, 1\}^{\partial \mathsf{B}}$, 
\begin{equation}
 \max_{\tau, \sigma}   \big\| \rho\big(x_{\mathsf{A}}\big|  x_{\partial \mathsf{B}}=\tau\big) - \rho\big(x_{\mathsf{A}}\big|  x_{\partial \mathsf{B}}=\sigma\big)\big\|_1 \leq e^{-\Omega(d(\mathsf{A}, \mathsf{C}))},
\end{equation}
where $\mathsf{C}(\tau, \sigma):=\{u\in \partial\mathsf{B}:\tau_u\neq \sigma_u\}$ is the region that $\tau, \sigma$ are allowed to differ, and $d(\cdot,\cdot)$ is the distance metric on the underlying geometry. 

Naturally, the quantification over \textit{all} boundary conditions $\tau, \sigma$ is too strong to hold at low temperatures, even for a suitable definition of \textsf{SSM} just for the \textit{syndrome} distribution.\footnote{As an example, consider a 2D Ising model in an $L \times L$ box with two violated syndromes across from one another; this results in $1/\sqrt{L}$ rather than exponential decay of correlations in the syndrome space.} Inspired by the aforementioned notion of $\mathsf{WSM}$ within a logical sector, here we introduce a relaxation of $\mathsf{SSM}$ to specific, ``ground-state-like'' boundary conditions, where all the boundary \textit{syndromes} are pinned to 0. We refer to this notion as ``domain'' $\mathsf{SSM}$, as it quantifies the influence of deformations not in the boundary conditions, but in the domain itself (with ground state boundary conditions, see \cref{fig:ssm}): 

\begin{definition}
    [$\mathsf{Domain}$ $\mathsf{SSM}$, Informal]\label{definition:domain_ssm_intro} Fix an inverse-temperature $\beta$, and let $\pi_\beta$ denote the syndrome Gibbs distribution of a (classical or quantum) error-correcting code. Then $\pi_\beta$ is said to satisfy $\mathsf{Domain}$ $\mathsf{SSM}$, if for any region $\mathsf{B}\subset [m]$ and disjoint $\mathsf{A}, \mathsf{C}\subseteq \mathsf{B}$,
    \begin{equation}
    \big\|\pi_\beta(s_\mathsf{A}|s_{\partial \mathsf{B}}=0) - \pi_\beta(s_\mathsf{A}|s_{\partial (\mathsf{B}\setminus \mathsf{C})}=0)\big\|_1 \leq e^{-\Omega(d(\mathsf{A}, \mathsf{C}))} .
    \end{equation}
\end{definition}

The intuition behind \cref{definition:domain_ssm_intro} is that at low temperatures, the energy excitations (syndrome configurations) of the code are sparse and localized. The syndromes $\mathsf{A}$ at the center have essentially “forgotten” about the syndrome configurations far away, and they might as well be near the ground state.

\begin{remark}
\textsf{Domain SSM} (\cref{fig:ssm}) implies \textsf{WSM} within a logical sector (\cref{fig:wsm}). To see this, consider the special case where $\mathsf{B} = [m]$ is the full graph, and $\mathsf{C}$ is the boundary of some ball containing $\mathsf{A}$. 
\end{remark}

In \cref{lemma:ssm}, we prove that the family of codes considered in this paper admits $\mathsf{Domain}$ $\mathsf{SSM}$, below some constant threshold temperature. 

\begin{theorem}
    [\textsf{Domain SSM}, informal version of \cref{lemma:ssm}]\label{theorem:ssm_intro} Every self-correcting (classical or quantum) memory $\mathcal{M}$ satisfying the connectedness-criterion of \cite{Bombin2009SelfcorrectingQC} (\cref{definition:connectedness}) exhibits a (sufficiently small) constant temperature $\beta_\mathcal{M}^{-1}$ below which it satisfies \textsf{Domain} \textsf{SSM}.
\end{theorem}

We do not expect the encountered temperature threshold to be sharp; in fact, to the extent of our knowledge, there is no rigorous analysis of the precise critical point of the 4D toric code. A concise proof-sketch is placed in the next \cref{section:techniques}. Next, we seek to boost the \textsf{Domain SSM} property on the syndrome space, into a rapid mixing guarantee.

\subsubsection{Rapid Mixing under a Decoding-Inspired, Quasi-Local Dynamics}
\label{section:algorithm}

In the absence of monotonicity, we appeal to a standard relaxation of single-site Glauber dynamics in non-monotone models: block-dynamics, where blocks $\mathcal{B}^L$ of size $L=\polylog(n)$ are resampled according to the conditional measure of their current boundary.  In addition to being a relevant tool for the analysis of single-site dynamics, block dynamics are well-studied in their own right \cite{Martinelli_2004, Brando2016FiniteCL, chen2025quantumgibbsstateslocally}. Our starting point is an elegant paper of \cite{Dyer2002MixingIT}, who used combinatorial arguments to show (among other results) that \textsf{SSM} implies block-dynamics mixes quickly, even in non-monotone systems. While we broadly follow this strategy, it is key in their argument that \textsf{SSM} holds for \textit{all} spin boundary conditions. Since we only have guarantees under ``ground-state'' syndrome boundaries, we require several modifications to the algorithm and analysis.

The new insight in our approach is to design a modified heat-bath block dynamics, inspired by parallel decoders for the 4D toric code \cite{hastings2013decodinghyperbolicspacesldpc,Zhou2025FinitetemperatureQT}. We define a detailed-balanced dynamics -- which we refer to as the \emph{conditional block dynamics} (see \cref{fig:di_dynamics}, or \cref{definition:cbd} for a formal definition) -- that updates local blocks of qubits, while freezing all the connected components of violated syndromes that intersect the boundary of the block. As we explain in detail in \cref{section:mixing}, this constraint plays a crucial role: it allows the spontaneous creation/annihilation/propagation of clusters of syndromes within the bulk of the block, but explicitly avoids situations in which the clusters on the boundary grow or merge together. As a result, we are able to maintain well-controlled boundary conditions throughout the evolution.

The main result of this work is that this decoding-inspired block dynamics is rapid mixing.

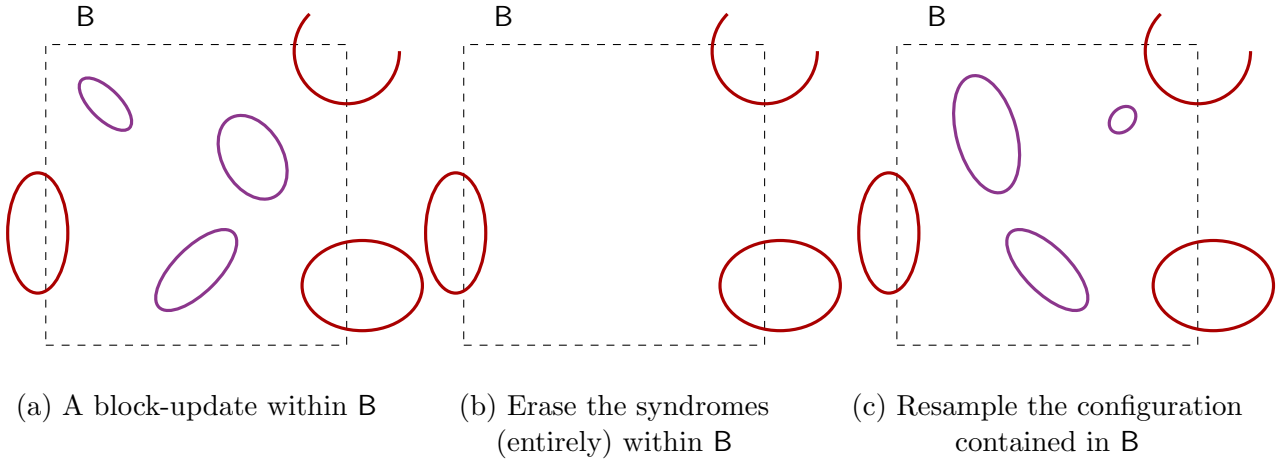
\begin{figure}[t]
    \begin{equation*}
    \begin{tikzpicture}[baseline=0cm]
        
            \node (mybox) [draw, dashed, rectangle, minimum width=4cm, minimum height=4cm] at (0,0) {};

            \node[anchor=south east, above left=0.1cm and -0.8cm] at (mybox.north west) {$\mathsf{B}$};

            \draw (mybox.north west) ++(0.8cm,-0.8cm) coordinate (topleft);
            \draw [Fuchsia, line width=1.2pt, rotate=-45] (topleft) ellipse (0.45cm and 0.2cm);

            \draw (mycircle.center) ++(.75cm,0.5cm) coordinate (innerellipsecenter);
            \draw [Fuchsia, line width=1.2pt, rotate=30] (innerellipsecenter) ellipse (0.4cm and 0.6cm);

            \draw (mycircle.center) ++(0.0cm,-1.0cm) coordinate (lowerellipsecenter);
            \draw [Fuchsia, line width=1.2pt, rotate=-45] (lowerellipsecenter) ellipse (0.3cm and 0.7cm);

            \draw (mybox.south east) ++(0.2cm,0.8cm) coordinate (bottomrightinner);
            \draw [DarkRed, line width=1.2pt, rotate=0] (bottomrightinner) ellipse (0.8cm and 0.6cm);

            \draw (mybox.south west) ++(-0.1cm,1.5cm) coordinate (bottomleftcenter);
            \draw [DarkRed, line width=1.2pt, rotate=0] (bottomleftcenter) ellipse (.4cm and .8cm);

            \draw[DarkRed, line width=1.2pt] (mybox.north east) ++(-0.5cm,0.4cm) arc (135:360:0.7cm);

            \node[below=0.5cm, align=center] at (mybox.south) {(a) A block-update within $\mathsf{B}$};
        \end{tikzpicture}
        \begin{tikzpicture}[baseline=0cm]
        
            \node (mybox) [draw, dashed, rectangle, minimum width=4cm, minimum height=4cm] at (0,0) {};

            \node[anchor=south east, above left=0.1cm and -0.8cm] at (mybox.north west) {$\mathsf{B}$};

            \draw (mybox.south east) ++(0.2cm,0.8cm) coordinate (bottomrightinner);
            \draw [DarkRed, line width=1.2pt, rotate=0] (bottomrightinner) ellipse (0.8cm and 0.6cm);

            \draw (mybox.south west) ++(-0.1cm,1.5cm) coordinate (bottomleftcenter);
            \draw [DarkRed, line width=1.2pt, rotate=0] (bottomleftcenter) ellipse (.4cm and .8cm);

            \draw[DarkRed, line width=1.2pt] (mybox.north east) ++(-0.5cm,0.4cm) arc (135:360:0.7cm);

            \node[below=0.5cm, align=center] at (mybox.south) {(b) Erase the syndromes \\ (entirely) within $\mathsf{B}$};
        \end{tikzpicture}
        \begin{tikzpicture}[baseline=0cm]
        
            \node (mybox) [draw, dashed, rectangle, minimum width=4cm, minimum height=4cm] at (0,0) {};

            \node[anchor=south east, above left=0.1cm and -0.8cm] at (mybox.north west) {$\mathsf{B}$};

            \draw (mybox.north west) ++(1.2cm,-1.2cm) coordinate (topleft);
            \draw [Fuchsia, line width=1.2pt, rotate=15] (topleft) ellipse (0.4cm and 0.8cm);
            
            \draw (mycircle.center) ++(-0cm,-1.0cm) coordinate (innerellipsecenter);
            \draw [Fuchsia, line width=1.2pt, rotate=45] (innerellipsecenter) ellipse (0.3cm and 0.7cm);

            \draw (mycircle.center) ++(1.0cm,1.0cm) coordinate (innerellipsecenter2);
            \draw [Fuchsia, line width=1.2pt, rotate=45] (innerellipsecenter2) ellipse (0.2cm and 0.15cm);

            \draw (mybox.south east) ++(0.2cm,0.8cm) coordinate (bottomrightinner);
            \draw [DarkRed, line width=1.2pt, rotate=0] (bottomrightinner) ellipse (0.8cm and 0.6cm);

            \draw (mybox.south west) ++(-0.1cm,1.5cm) coordinate (bottomleftcenter);
            \draw [DarkRed, line width=1.2pt, rotate=0] (bottomleftcenter) ellipse (.4cm and .8cm);

            \draw[DarkRed, line width=1.2pt] (mybox.north east) ++(-0.5cm,0.4cm) arc (135:360:0.7cm);

            \node[below=0.5cm, align=center] at (mybox.south) {(c) Resample the configuration \\ contained in $\mathsf{B}$};
        \end{tikzpicture}
        \quad
    \end{equation*}
    \caption{The conditional block dynamics. Given an initial configuration (a), the connected components of syndrome violations within $\mathsf{B}$ (purple) are decoded and erased, while those incident on the boundary $\partial\mathsf{B}$ (red) are fixed (b). Subsequently, the syndrome configuration in the ``bulk'' of $\mathsf{B}$ is resampled from the syndrome Gibbs distribution conditioned on the syndrome patterns incident on the boundary (c). }
    \label{fig:di_dynamics}
\end{figure}

\begin{theorem}
    [Rapid Mixing within a Logical Sector, informal version of \cref{theorem:cbd_rapid_mixing}] \label{theorem:cbd_intro} Every self-correcting (classical or quantum) memory $\mathcal{M}$ of $n$ qubits on a finite-dimensional lattice satisfying the connectedness criterion of \cite{Bombin2009SelfcorrectingQC} (\cref{definition:connectedness}), exhibits a constant inverse temperature $\beta_\mathcal{M}$, and constants $c_1, c_2>0, c_3> 1$ with the following guarantee. 

    At any inverse temperature $\beta>\beta_\mathcal{M}$, there exists a continuous-time $\polylog (n)$-local Lindbladian evolution $\mathcal{L}$ which, when initialized from a ground state $\psi$, rapidly converges to the associated Gibbs state within the logical sector $\psi$. That is, $\forall t\geq 0:$
    \begin{equation}\label{eq:maintheoremintro}
       \under{\big\|e^{t\mathcal{L}}[\psi] - \rho_\beta^\psi\big\|_1}{\textsf{Distance to Metastability}} \leq \under{c_1\cdot n \cdot e^{-c_2 t}}{\textsf{Rapid Mixing in a Sector}} + \under{t\cdot (n\cdot e^{-(\beta-\beta_\mathcal{M})\cdot \log^{c_3} n})}{\textsf{Leakage Rate}}\,.
    \end{equation}
    In particular, if $\log n\ll t\ll n^{O(1)}$, then the distance above can be made polynomially small.
\end{theorem}

 There are two different time-scales appearing in \cref{eq:maintheoremintro}, which distinguish two contributions to the convergence behavior: the rapid mixing within a logical sector, and the quasi-polynomially slow leakage out of that ``sector''. Our proof of \cref{theorem:cbd_intro} is based on a path coupling argument within a high-measure subset of the configuration space, consisting of syndrome configurations that are geometrically well-behaved. Within this subset, we prove contraction under block updates by leveraging \textsf{Domain SSM}. Combined with a ``memory'' statement which shows that the system remains confined to the well-behaved region for a sufficiently long time, we are able to establish rapid mixing from the ground state.

\cref{theorem:cbd_intro} also implies the existence of a quantum circuit of $\polylog(n)$ depth, comprised of $\polylog(n)$-local channels, which prepares the Gibbs state within a logical sector starting from the ground state. However, when implemented using 2-qubit gates, each quasi-local channel could have exponential depth in the size of the block in the worst case. Therefore the total quantum circuit depth to prepare the Gibbs state within a logical sector is quasi-polynomial in $n$ without assuming further structure of the code.\footnote{We remark this is a drawback often encountered by quantum Gibbs samplers based exclusively on decay-of-correlations, such as that of~\cite{Brando2016FiniteCL}.} Nonetheless, in the special case of the 4D toric code, we show via a connection to the even subgraph model and the \textit{worm process} \cite{PS01,CGTT16}, that the block updates can be implemented efficiently.

\begin{theorem}[Efficient Block Updates for the 4D toric code, informal version of \cref{theorem:block_updates_toric}]
\label{thm:efficiency}
For the 4D toric code, each conditional block update on a ball of side-length $L$ can be implemented up to error $\epsilon$ using $\poly(L)\cdot \log(1/\epsilon)$ local gates and ancillas.  
\end{theorem}

\noindent Consequently, the entire state preparation algorithm for the 4D toric code can be simulated using a quantum circuit with 2-qubit gates of circuit depth $\polylog(n)$, up to inverse-polynomial error.

\subsection{Discussion}
\label{section:discussion}
We end this section with a discussion on applications and related questions.

\paragraph{Classification of quantum phases of matter.} Our results also lead to an interesting characterization of quantum phases of matter in self-correcting quantum memories. According to the definitions of \cite{Coser2018ClassificationOP,Sang2023MixedStateQP}, a pair of mixed states $\rho, \sigma$ is said to be in the same phase $\rho \leftrightarrow \sigma$, if $\rho$ can be mapped to $\sigma$ via a low-depth sequence of (quasi-)local channels and vice versa.

In this context, our result thus rigorously establishes that the ground states $\ket{\psi}$ of a broad class of self-correcting quantum memories lie in the same phase as the associated Gibbs states within a logical sector $\rho_{\beta}^\psi$, which has been an open question (see e.g.~\cite{Zhou2025FinitetemperatureQT}).
\begin{equation}
    \begin{tikzcd}[column sep=large]
\ketbra{\psi}  \arrow[r, shift left=1.2ex, "{\mathcal{L}}"] 
  \arrow[r, shift right=1.2ex, swap, leftarrow, "{\mathsf{Dec}}"]
& \rho_{\beta}^\psi
\end{tikzcd}
\end{equation}
While our Lindbladian evolution is not strictly local, its quasi-local dynamics can be simulated by a $\polylog(n)$ depth circuit of $\polylog(n)$-local channels. Therefore, it has a quasi-local lightcone, which is often sufficient for applications such as relating long-range correlations in both states via Lieb-Robinson bounds. We remark that the converse channel is guaranteed by the existence of a quasi-local (parallel) decoder for the low-temperature excitations $\ketbra{\psi} =\mathsf{Dec}(\rho_{\beta}^\psi)$, which can be extracted from the original results of \cite{Alicki2008OnTS, Bombin2009SelfcorrectingQC} or from \cite{Zhou2025FinitetemperatureQT}.

\paragraph{Generalized Lifshitz conjecture.} A famous open question in the theory of Glauber dynamics is the Lifshitz law, which posits that the mixing time for the 2D Ising model in an $L\times L$ box, under $0$ boundary conditions, should scale as $L^2$ (see~\cite{M94,MT10,LMST13} for partial progress). Heuristically, this reflects the timescale required for curvature-driven domain shrinkage, in analog to surface tension. One could posit a generalization of this conjecture to other self-correcting memories, such as high dimensional toric or color codes, where one instead pins the syndromes at the boundary to $0$. Polynomial mixing of subgraphs of $\mathbb{Z}^d$ with 0 boundary, could yield, for instance, a fast simulation of each block update in our quasi-local block dynamics, similar to \cref{thm:efficiency}.

\paragraph{The quantum ``reverse Mpemba'' effect.} As previously discussed, a necessary step towards proving rapid mixing within a logical sector under single-site Glauber dynamics is to establish \emph{gound state stability}. Namely, one must show that the memory time of the 4D toric code, under single-site Glauber dynamics when initialized from a ground state, is at least comparable to that when initialized from the Gibbs state within a logical sector. In some sense, this is asking if ``lower temperature states can heat up quicker'', curiously related to the reverse of the \textit{Mpemba effect} \cite{LR17}, which says that hot water can freeze faster than cold water. 

While we do not expect self-correcting memories to exhibit this behavior, in the absence of monotonicity, it would seem that rigorously ruling it out requires new techniques. Indeed, we remark that the only non-monotone models of self-correcting memories where ground state stability has rigorously been proven are based on (a strong form of) expansion, with linear confinement up to linear distance such as classical or quantum Tanner codes~\cite{Leverrier2022quantum} (this is folklore and follows similarly to e.g.~\cite[Theorem~A.12]{Hong2024QuantumMA}).

\paragraph{Other models of self-correcting memories.} As discussed earlier in \cref{section:related}, our framework focuses on codes on finite-dimensional lattices and does not apply to classical and quantum expander codes. Here we elaborate further on their key structural differences: metachecks and the presence of local minima.
First, unlike lattice models, many expander codes have no metachecks at all; thus the Peierls arguments used in our framework are not applicable. The proofs of thermal stability in these models are instead based on linear confinement \cite{Hong2024QuantumMA, placke2024topologicalquantumspinglass}. 

Second, unlike lattice models, many expander codes have local minima: there exist low-energy eigenstates that are far from any ground state; in other words, constant-weight syndromes associated to high-weight Pauli errors. For these models, a precise definition of Gibbs states within a logical sector must take into consideration these local minimas. We leave it as an interesting open question to propose the correct definition and prove rapid mixing within a logical sector for quantum expander codes. 

\subsection{Organization}
In \cref{section:techniques}, we give a proof-sketch of our main results. \cref{section:preliminaries} introduces our notation and basic concepts, and \cref{section:self-correcting} describes the model of Glauber dynamics on error-correcting codes and the basic definitions of self-correction.

The core of our technical argument begins in \cref{section:conditions}, where we state the assumptions on the family of codes that we study. We then show in \cref{section:ssm} that these codes admit \textsf{Domain SSM} (\cref{theorem:ssm_intro}). We leverage this spatial mixing property in \cref{section:mixing} to prove the rapid mixing of our conditional block dynamics (\cref{theorem:cbd_intro}). Finally, in \cref{section:locality}, we show how to efficiently implement the conditional block updates for the 4D toric code.

The appendices provide background on the 4D toric code in \cref{section:toric_background} and on open quantum system dynamics in \cref{section:open_quantum_systems}.

\subsection*{Acknowledgements}
We thank Anurag Anshu, João Basso, Nikolas Breuckmann, Chi-Fang Chen, Tyler D. Ellison, Daniel Stilck França, Jeongwan Haah, Vedika Khemani, Lin Lin,  Sidhanth Mohanty, Nick O’Dea, Amit Rajaraman, Shengqi Sang, Umesh Vazirani, David Wu, and Alexander Zlokapa for helpful discussions. This work was done in part while the authors were visiting the Simons Institute for the Theory of Computing, supported by DOE QSA grant number FP00010905. T.B.~acknowledges support by the National Science Foundation under Grant No.~DGE~2146752. R.G.\ is supported in part by NSF CAREER grant 2440509 and NSF grant DMS-2246780. Y.L.~acknowledges support through the Harvard Quantum Initiative postdoctoral fellowship and Anurag Anshu's startup funds from Harvard University.

\section{Technical Overview}
\label{section:techniques}

We dedicate this section to a more detailed proof sketch, in establishing \textsf{Domain} \textsf{SSM} for the syndrome Gibbs distribution, and our mixing-time argument for the conditional block-dynamics. We tailor the discussion to the 4D toric code for simplicity and concreteness, and defer generalizations to the bulk of the paper.  

A dedicated background on the  4D toric code is placed in \cref{section:toric_background}. Here, we simply recall that the valid syndrome configurations of the 4D toric code, much like the 2D Ising model, form cycles (or more generally even subgraphs) on the edges of the torus.

\subsection{Self-Correction via Peierls Arguments}
\label{subsection:self_correction_summary}

We begin with an introduction to the key tool we use to establish spatial mixing properties: a Peierls argument. To illustrate this, we sketch \cite{Alicki2008OnTS}'s proof that the 4D toric code is thermally stable when initialized from the Gibbs state within a logical sector. Roughly speaking, the encoded logical information is stable until a topologically nontrivial syndrome-cycle arises. The crux of the argument is to reduce this dynamical stability question, \textit{when initialized from the Gibbs state within a logical sector}, to a static estimate. Since the syndrome configurations are already in thermal equilibrium (recall \cref{eq:gibbswithinlogicalsector}), it suffices to compute whether $\pi_\beta$ gives significant probability to configurations with long cycles. This is precisely the role of the Peierls argument.

Suppose we fix any even subgraph $\mathsf{C}\subset [m]$ on the $m$ edges of the lattice, of length $|\mathsf{C}|\geq \ell$ and let $\mathsf{K}_\mathsf{C}\subset \{0, 1\}^m$ denote the set of syndrome configurations which contain syndrome violations on $\mathsf{C}$. The main idea in the argument is to consider the set of configurations with $\mathsf{C}$ removed:
\begin{equation}
    \mathsf{K}_\mathsf{C}^* = \left\{ s\oplus \mathsf{C}: s\in \mathsf{K}_\mathsf{C}\right\}\,.
\end{equation}

\noindent The set of configurations $\mathsf{K}_\mathsf{C}^* $ is in bijection with $\mathsf{K}_\mathsf{C}$, and so long as $\mathsf{C}$ is topologically trivial, consists only of valid syndrome configurations.\footnote{For simplicity, we refer the reader to \cref{section:peierls} for the case that $\mathsf{C}$ is topologically non-trivial.} We can then analyze the probability that the specific collection of cycles $\mathsf{C}$ appears:
\begin{equation}
\mathbb{P}_{X\leftarrow\pi_\beta}[X\in \mathsf{K}_\mathsf{C}] \leq \frac{\sum_{s\in \mathsf{K}_\mathsf{C}} e^{-\beta|s|}}{\sum_{s\in \mathsf{K}_\mathsf{C}\cup \mathsf{K}_\mathsf{C}^*} e^{-\beta|s|}} \leq \frac{1}{1+e^{\beta |\mathsf{C}|}}\leq e^{-\beta \ell}.
\end{equation}

\noindent A union bound is then sufficient to show that the probability that a sample from $\pi_\beta$ contains any large cycle of length $\ell\gg\log m$ also decays exponentially in $\ell$, as long as the inverse temperature $\beta>\ln d$ is sufficiently large relative to the degree $d$ of the graph. If the torus is of side-length $w\sim m^{1/4}$, then the leakage rate of the memory is bounded by the probability that cycles of size $\Omega(w)$ occur, which implies a memory time of $e^{\Omega(w)}$. Next, we show how to derive decay-of-correlations properties for the syndrome Gibbs distribution via a delicate extension of this argument.

\subsection{Correlation-Decay via Separating Surfaces}
Here, we show how to prove $\mathsf{WSM}$ within a logical sector (\cref{fig:wsm}) for the syndrome Gibbs distribution of the 4D toric code; the extension to \textsf{Domain SSM} (\cref{fig:ssm}) is geometrically more involved, but follows the same strategy. That is, we specialize \cref{definition:domain_ssm_intro} to the case where $\mathsf{B}$ is the full graph, $\mathsf{A} = \{u\}$ is a single syndrome, and we let $\mathsf{C}=\partial\mathcal{B}^L_u$ denote the boundary of a box of radius $L$ around $u$, which reduces to \cref{fig:wsm}. 

Our goal now is to argue that the marginals at $u$ of the syndrome Gibbs distribution $\pi_\beta$ and $\pi_\beta(\cdot | s_{\partial\mathcal{B}_u^L} = 0)$ without and with boundary pinning, are close. Suppose we draw a sample independently from each:
\begin{equation}\label{equation:samples_from_conditional_gibbs_intro}
    X\leftarrow \pi_\beta\big(\cdot \big) \,,\qquad Y\leftarrow \pi_\beta\big(\cdot \big| s_{\partial\mathcal{B}_u^L} = 0\big)\,.
\end{equation}
The key idea is to reason that both $X, Y$ agree on a \textit{separating surface}, a notion originating from percolation theoretic ideas, found in e.g. \cite{grimmett1989percolation} (\cref{definition:sep_surface_intro}), with high probability. %

\begin{definition}
    [Separating Surface, Informal]\label{definition:sep_surface_intro} For a syndrome configuration $s\in \{0, 1\}^m$ and a center $u\in [m]$, a collection $\mathsf{S}\subset \mathcal{B}^L_u$ of syndromes is said to be a separating surface for $u$ if 
    \begin{enumerate}
        \item There are no violations in $\mathsf{S}$. That is, $\forall v\in \mathsf{S}:s_v = 0$.
        \item $\mathsf{S}$ disconnects (screens) $u$ from the boundary of $\mathcal{B}^L_u$ (see \cref{fig:separating_surface}).
    \end{enumerate}
\end{definition}

\begin{figure}[t]
    \centering
    \begin{equation}
        \begin{tikzpicture}[baseline=0cm]
        
            \node (mybox) [draw=RoyalBlue, line width=1mm, rectangle, minimum width=4cm, minimum height=4cm] at (0,0) {};

            \node[below right =1.2cm and 0.35cm] at (mybox.center) {$\mathsf{B}$};

            \node (mycircle) [draw, circle, minimum size=1cm, fill=Gray!20] at (mybox.center) {$\mathsf{A}$};

            \draw[RoyalBlue, dashed, line width=1.2pt, rotate=45] (mycircle.center) ellipse (1.5cm and 1.0cm);

            \draw (mybox.north west) ++(0.8cm,-0.8cm) coordinate (topleft);
            \draw [DarkRed, line width=1.0pt, rotate=-45] (topleft) ellipse (0.45cm and 0.2cm);

            \draw (mycircle.center) ++(.75cm,0.5cm) coordinate (innerellipsecenter);
            \draw [DarkRed, line width=1.0pt, rotate=30] (innerellipsecenter) ellipse (0.3cm and 0.15cm);

            \draw (mybox.south east) ++(-0.4cm,0.8cm) coordinate (bottomrightinner);
            \draw [DarkRed, line width=1.0pt, rotate=0] (bottomrightinner) ellipse (0.25cm and 0.6cm);

            \node[below=0.5cm, align=center] at (mybox.south) {(a) A sample $Y\leftarrow \pi_\beta\big(\cdot \big|s_{\partial\mathsf{B}} = 0\big)$};

        \end{tikzpicture}
        \quad \quad \quad
        \begin{tikzpicture}[baseline=0cm]
        
            \node (mybox) [draw, dashed, rectangle, minimum width=4cm, minimum height=4cm] at (0,0) {};

            \node[below right =1.2cm and 0.35cm] at (mybox.center) {$\mathsf{B}$};

            \node (mycircle) [draw, circle, minimum size=1cm, fill=Gray!20] at (mybox.center) {$\mathsf{A}$};

             \draw[RoyalBlue, dashed, line width=1.2pt, rotate=45] (mycircle.center) ellipse (1.5cm and 1.0cm);

            \draw (mybox.north west) ++(0.8cm,-0.8cm) coordinate (topleft);
            \draw [DarkRed, line width=1.0pt, rotate=-45] (topleft) ellipse (0.45cm and 0.2cm);

            \draw (mycircle.center) ++(.75cm,0.5cm) coordinate (innerellipsecenter);
            \draw [DarkRed, line width=1.0pt, rotate=30] (innerellipsecenter) ellipse (0.3cm and 0.15cm);

            \draw (mybox.south east) ++(0.2cm,0.8cm) coordinate (bottomrightinner);
            \draw [DarkRed, line width=1.0pt, rotate=0] (bottomrightinner) ellipse (0.8cm and 0.6cm);

            \node[below=0.5cm, align=center] at (mybox.south) {(a) A sample $X\leftarrow \pi_\beta\big(\cdot \big)$};
        \end{tikzpicture}
    \end{equation}

    \caption{The samples $X, Y$ from the conditional distributions above agree on a shared \textit{separating surface} of syndromes $\mathsf{S}$ (\cref{definition:sep_surface_intro}, dashed blue) within $\mathsf{B}$ with high probability. $\mathsf{S}$ does not intersect any syndrome violation (solid red cycles), and disconnects $\mathsf{A}$ from the rest of $\mathsf{B}$.}
    \label{fig:separating_surface}
\end{figure}
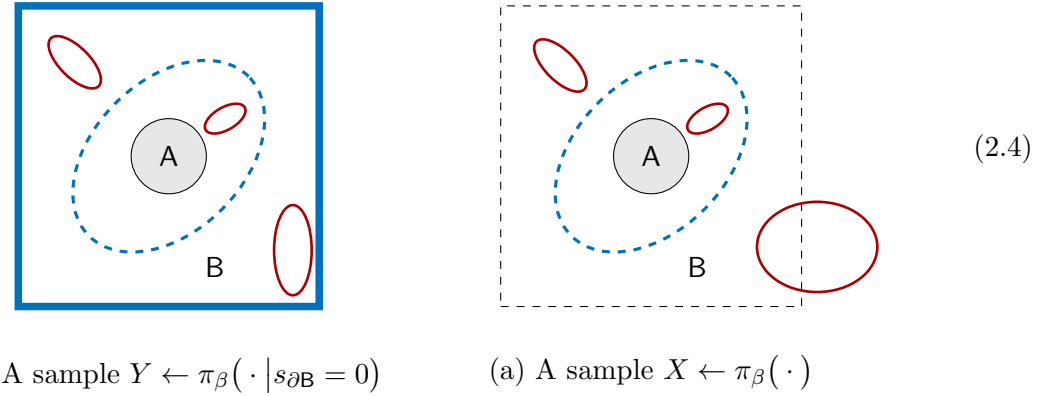

In the case of the 4D toric code, this (essentially) just stipulates that no collection of violated syndromes in $X\vee Y$ spans from outside $\partial \mathcal{B}^L_u$ to the center $u$. We claim (and justify shortly) that both $X$ and $Y$ agree on some separating surface $\mathsf{S}$ with high probability. For now, let us briefly discuss why the existence of a shared separating surface could be useful in developing coupling arguments. The crux of the proof is to leverage the Markov property.\footnote{The syndrome Gibbs distribution $\pi_\beta$ does not exactly satisfy a Markov property, but it does satisfy an approximate version which suffices for our purposes; see \cref{section:markov}.} The marginal distribution of the region in the ``interior" of the surface $\mathsf{S}$ is (approximately) the same for $X$ and $Y$, once one conditions on the existence of a separating surface at $\mathsf{S}$. One can then couple the samples in such a manner that $X_u=Y_u$ with high probability. 

To generalize this approach to \textsf{Domain SSM}, given the two boundaries $\partial \mathsf{B}, \partial(\mathsf{B\setminus C})$, we similarly exhibit a separating surface within $\mathsf{B}$ which screens $\mathsf{A}$ from the region around the region $\mathsf{C}$ where the boundary conditions differ, see \cref{section:ssm} for details. 

\paragraph{The shared separating surfaces argument.} Here we sketch the proof that the samples $X, Y$ from \cref{equation:samples_from_conditional_gibbs_intro} share a separating surface with high probability. 

\begin{claim}
    [Informal version of \cref{claim:separating_surfaces_toric}] There exists a constant $\beta_0$ such that for all $\beta>\beta_0$, the samples $X, Y\in \{0, 1\}^m$ from the conditional syndrome Gibbs distributions of \cref{equation:samples_from_conditional_gibbs_intro} agree on a separating surface $\mathsf{S}\subset \mathcal{B}^r_u$ around $u$ except with probability $\exp[-\Omega(r\cdot (\beta-\beta_0))]$.
\end{claim}

\noindent The argument proceeds via the contrapositive: If $X, Y$ don't share a separating surface within a radius $r\leq L$ around the center $u$, then there exists a path $\mathsf{p}\subset \mathcal{B}^r_u$ of length at least $r$ from $u$ to $\partial \mathcal{B}^r_u$, such that each edge in $\mathsf{p}$ is violated in either $X$ or $Y$. The intuition that this path exists is that one can stitch the path together iteratively, from the outside in, traversing edges (violated syndromes) in either $X$ or $Y$. Since the path $\mathsf{p}$ is of length $\geq r$, there is a collection of subsegments (or subpaths) $\mathsf{sp}\subseteq \mathsf{p}$ of total length $\geq \frac{r}{2}$ \textit{either} entirely violated in $X$ or entirely violated in $Y$; WLOG let it be $X$. It then suffices to compute the probability that such a long collection of violated subpaths arises in a sample from the syndrome Gibbs distribution $\pi_\beta$. Intuitively, this is about controlling the probability of a large set of connected components, which follows from a careful Peierls argument similar to \cref{subsection:self_correction_summary}, with a more delicate combinatorial counting.

\subsection{The Rapid Mixing Argument}

With \textsf{Domain SSM} in hand, the crux of our mixing time argument is to identify a ``good'' fraction of the state space that captures most of the weight in $\pi_\beta$, and show that our quasi-local dynamics mixes quickly within this restricted space as a result of \textsf{Domain SSM}. 

\begin{definition}
    [The Small-Loop Configurations] Let $\Omega\subseteq\{0,1\}^m$ be the set of all valid syndrome configurations. We consider the subset of the state space of valid syndromes $\Omega_R\subset \Omega$, where all the connected components of syndromes -- i.e. the connected components of cycles -- are of length $\leq R$. 
\end{definition}

Let $L$ be the side-length of the block in our conditional block dynamics (\cref{section:algorithm}). We consider the induced dynamics on syndromes and choose the parameters such that $L\gg R \gg \log m$. Note that when $R \gg \log m$, the Peierls argument of \cref{subsection:self_correction_summary} implies that the syndrome Gibbs distribution is supported mostly in $\Omega_R$, so it suffices to converge to the Gibbs state restricted to $\Omega_R$. The conceptual structure of our proof rests on establishing the following dichotomy:

\begin{enumerate}
    \item \textbf{The memory lemma} (\cref{lemma:Pisstable}). If the Markov chain is initialized from the ground state $0^m\in \Omega_R$, it takes time exponential in $R$ to leave the set $\Omega_R$. 

    \item \textbf{Rapid mixing of the restricted chain} (\cref{lemma:Q_fast_mixing}). The restricted chain $\mathsf{Q}$, given by running the conditional block-dynamics and rejecting transitions out of $\Omega_R$, is rapid mixing -- more precisely, in $O(\log(m/\epsilon))$ continuous time, assuming $R\gg \log m$. 
\end{enumerate}

The memory lemma (\cref{lemma:Pisstable}) is a statement on the stability of the small-loop configurations; and follows from similar arguments to the original proofs of self-correction of the 4D toric code \cite{Alicki2008OnTS} also discussed in \cref{subsection:self_correction_summary}. In some sense, \cref{lemma:Pisstable} implies that the conditional block-dynamics and its restriction to $\Omega_R$ are indistinguishable, at least over the time scale required for the dynamics to converge to the Gibbs distribution on $\Omega_R$.

In \cref{subsection:fast_mixing}, we elaborate on why the convergence in the syndrome distribution suggested above, implies the convergence of the original Lindbladian dynamics as in  \cref{theorem:cbd_intro}. Next, we explain the key argument in our proof: to show the restricted chain mixes quickly.

\paragraph{The restricted chain is rapid-mixing.}  

We follow a path-coupling approach (see \cref{fact:path_coupling} or \cite{guruswami2016rapidlymixingmarkovchains}), akin to the high-temperature argument of~\cite{Dyer2002MixingIT}. Given two configurations of syndromes $X, Y\in \Omega_R\subset \{0, 1\}^m$ which are ``adjacent", we establish a coupling $\omega:\Omega_R\times \Omega_R\rightarrow \Omega_R\times \Omega_R$ such that after one step of the restricted chain $\mathsf{Q}$, the distance between $X, Y$ has decreased.

In order to quantify the contraction between $X, Y$, we need to be cautious about the notion of adjacency. Indeed, recall that not all strings in $\{0, 1\}^m$ are valid syndromes; thus we cannot use the usual Hamming distance. 

\begin{definition}
    [The Cluster Distance] We define a graph $\mathcal{G} = (\Omega_R, E)$ over the state space $\Omega_R$, by connecting $x, y\in \Omega_R$ if they differ by the addition or removal of exactly one connected component of syndromes. The distance\footnote{For the 4D toric code, the distance is measured by the number of differing connected components of cycles.} $d_{\mathsf{Cluster}}(x, y)$ is the shortest path length from $x$ to $y$ in $\mathcal{G}$.
\end{definition}

In \cref{claim:cycle_distance_contracts}, we claim that at sufficiently low temperatures $\beta\geq \beta_0$ (the same threshold for the onset of \textsf{Domain} \textsf{SSM}), there exists a coupling $\omega$ such that the distance between adjacent pairs $X, Y$ s.t. $d_{\mathsf{Cluster}}(X, Y)=1$ contracts after each step of the restricted chain $\mathsf{Q}$:
\begin{equation}\label{equation:cd_contracts}
   \forall X, Y \text{ s.t. } d_{\mathsf{Cluster}}(X, Y)=1:\quad   \mathbb{E}_\omega\big[d_{\mathsf{Cluster}}(X', Y')\big] \leq 1-\alpha\equiv 1-\frac{1}{m}\,,
\end{equation}
\noindent where $X',Y'$ are one step of $\mathsf{Q}$ starting from $X,Y$ respectively. Following standard path-coupling arguments, applying \cref{equation:cd_contracts} iteratively in $t$ discrete times steps is sufficient to establish rapid mixing of $\mathsf{Q}$, starting from the ground state $0^m$:
\begin{equation}\label{equation:cd_contracts_iteratively}
    \big\|\mathsf{Q}^t(0^m) - \pi_\beta|_{\Omega_R}\big\|_1\leq  \mathbb{E}\big[d_{\mathsf{Cluster}}(X_t, Y_t)|X_0=0^m, Y_0\leftarrow \pi_\beta|_{\Omega_R}\big] \leq O(m)\cdot (1-\alpha)^t,
\end{equation}
or in words, an $O(m\log \frac{m}{\epsilon})$ discrete-time mixing time, which can be converted into the claimed result of $O(\log \frac{m}{\epsilon})$ in continuous time (see \cref{fact:cont_disc}). To conclude, we give a sketch of our proof of \cref{equation:cd_contracts}, that the cluster distance contracts. 

\paragraph{The cluster distance contracts, via $\mathsf{Domain}$ $\mathsf{SSM}$.}  
Note that $X, Y\in \Omega_R$ differ only on a single small connected component  $\mathsf{V}$ of violated syndromes of length $\leq R$. Thus, under a (coupled) conditional block update on some block $\mathcal{B}^L$, there are just three possibilities: 
\begin{enumerate}
    \item The block $\mathcal{B}^L$ completely misses the differing cluster $\mathsf{V}$ -- in which case the distance is unchanged. 
    \item The block completely contains the differing cluster, $\mathsf{V}\subset \mathcal{B}^L$ -- since $X, Y$ agree everywhere else (and in particular on $\partial\mathcal{B}^L$), after the update the distance is now $d_{\mathsf{Cluster}}(X', Y') = 0$.
    \item The boundary of the block-update $\partial \mathcal{B}^L$ hits the differing cluster $\mathsf{V}$. That is, $\mathsf{V}\cap \partial\mathcal{B}^L\neq \emptyset$. 
\end{enumerate}

\noindent This third case is the only possibility the distance between $X', Y'$ could have increased.

The challenge here is that, at low-temperatures, the update to the interior of $\mathcal{B}^L$ could be highly sensitive to its boundary configurations; in particular, the distance can increase by the entire size of the box in the worst case. This is why we utilize the \textit{conditional block dynamics}, introduced in \cref{section:algorithm}. By freezing all the connected components incident on $\partial \mathcal{B}^L$, the update to the interior of the box can always be viewed as conditioning on an ``inner boundary'' consisting exclusively of \emph{unviolated syndromes}. The intuition stemming from the $\mathsf{Domain}$ $\mathsf{SSM}$ property, is then that the updates in $X$ and $Y$ should only differ locally, in a radius $r\sim \log m$  around the differing cluster $\mathsf{V}$.

To be more precise, let us introduce some notation. Suppose $X, Y$ agree on connected components of cycles $\mathsf{V}_1, \mathsf{V}_2, \cdots$ which are incident on $\partial \mathcal{B}^L$, and disagree only on $\mathsf{V}$ (which, say, is in $Y$ but not $X$). By stitching their neighborhoods together, we can define an ``interior" to the box $\mathcal{B}^L$:
\begin{equation}
    \mathsf{B} \equiv \mathcal{B}^L\setminus \bigcup_i \big(N(\mathsf{V}_i)\cup \mathsf{V}_i\big)\,,
\end{equation}
where $N(\cdot)$ denotes neighborhood. This removes from the box the vertices in, and adjacent to, some $\mathsf{V}_i$. Crucially, since each $\mathsf{V}_i$ is a connected component of violated syndromes, the syndromes on its outer boundary $N(\mathsf{V}_i)\setminus \mathsf{V}_i$ must be unviolated (in both $X$ and $Y$). Consequently, $X_{\partial \mathsf{B}} = 0$. Analogously, if we remove $\mathsf{C} = \mathsf{V}\cup N(\mathsf{V})$ from $\mathsf{B}$, we construct a surface $\partial (\mathsf{B}\setminus \mathsf{C})$ such that $Y_{\partial (\mathsf{B}\setminus \mathsf{C})}=0$. The conditional block updates can then be cast as sampling from the conditional Gibbs distributions
\begin{equation}
    X|_{\mathsf{B}} \leftarrow \pi_\beta(s_\mathsf{B} | s_{\partial \mathsf{B}} = 0),\quad  Y|_{\mathsf{B}\setminus \mathsf{C}} \leftarrow \pi_\beta(s_{\mathsf{B}\setminus \mathsf{C}} | s_{\partial (\mathsf{B}\setminus \mathsf{C})} = 0).
\end{equation}

\noindent We can now apply the \textsf{Domain} \textsf{SSM} property to the region $\mathsf{B}$ and the boundary deformation $\mathsf{C}$, to conclude the marginals of $X', Y'$, on a distant region $\mathsf{B}\setminus \mathcal{B}_\mathsf{V}^r$ which is distance $r$ away from the differing cycle $\mathsf{V}$, are close. Accordingly, one can couple the updates such that $X', Y'$ agree on the region $\mathsf{B}\setminus \mathcal{B}_\mathsf{V}^r$, which in turn gives us a bound on the cluster-distance increase under Case 3 above. That is, $\forall X, Y$ s.t. $d_{\mathsf{Cluster}}(X, Y)=1$,
\begin{equation}
\mathbb{E}_\omega\big[d_{\mathsf{Cluster}}(X', Y')\big| \under{\big(\mathsf{V}\cap \partial\mathcal{B}^L\neq \emptyset\big)}{\text{Conditioned on Case 3}} \big] \leq |\mathcal{B}_\mathsf{V}^r| \bigg(1 + \under{|\mathcal{B}^L| \times e^{-\Omega(r(\beta-\beta_0))}}{\text{If the coupling fails}}\,\bigg)\,.
\end{equation}

\noindent The finite-dimensional lattice geometry (more generally, the amenability of the underlying graph) ensures that the probability the disagreeing set $\mathsf{V}$ is interior to the block (Case 2) is much more likely than on its boundary (Case 3); geometrical arguments with a careful choice of parameters then gives the desired claim in \cref{equation:cd_contracts}. Under an appropriate choice of parameters $R\sim \log m$ and $L\sim \polylog (m)$, we are able to establish \cref{equation:cd_contracts}, which in turn gives the desired rapid mixing bound in \cref{equation:cd_contracts_iteratively}. See \cref{section:mixing} for more details.

\section{Preliminaries and Notation}
\label{section:preliminaries}

In this section, we present basic definitions and facts of the underlying geometry of the codes we study (\cref{section:geometry}), classical Markov chains and Glauber dynamics (\cref{section:markov_chain}), and classical and quantum error-correction (\cref{section:error_correction}).

\subsection{Underlying Geometry}
\label{section:geometry}

Let $G = ([n], E)$ be an undirected graph. We refer to the ball or box $\mathcal{B}^R_\mathsf{U}\subseteq [n]$ of radius $R$ around a subset $\mathsf{U}\subseteq [n]$ as 
\begin{equation}
    \mathcal{B}^R_\mathsf{U} = \big\{v\in [n]: d(\mathsf{U}, v)\leq R\big\},
\end{equation}

\noindent where $d(\mathsf{U}, v)$ is the length of the shortest path in $G$ between $v$ and any $u\in\mathsf{U}$. We also associate the boundary of width $w$ of the box, 
\begin{equation}
    \partial_w\mathcal{B}^R_\mathsf{U} = \big\{v\in [n]: R< d(\mathsf{U}, v)\leq R+w\big\}   = \mathcal{B}^{R+w}_\mathsf{U}\setminus \mathcal{B}^R_\mathsf{U}
\end{equation}

\noindent and for simplicity simply refer to $\partial\mathcal{B}^R_\mathsf{U}=\partial_1\mathcal{B}^R_\mathsf{U}  $ as ``the boundary". On an amenable graph (such as finite dimensional lattices), the boundary (area) to bulk (volume) ratio of any box is $o(1)$ in $R$. The families of codes we will consider in this paper are defined on (a slight strengthening of) an amenable graph. 

\begin{definition}
    [Uniformly Amenable Graphs]\label{definition:amenable} We refer to a graphs $G = ([n], E)$ as $c(G)$-\textsf{uniformly-amenable} if there exist constants $c_1, c_2$ such that for every $w<R$, boxes in $G$ satisfy
    \begin{equation}
      \max_{u\in [n]}|\mathcal{B}^w_u|\cdot  \frac{ \max_{u\in [n]}|\partial_w\mathcal{B}^R_u| }{\min_{u\in [n]}|\mathcal{B}^R_u|} \leq c_2 \cdot \bigg(\frac{w^{c(G)}}{R}\bigg)^{c_1}.
    \end{equation}
\end{definition}

We comment that this is implied by standard amenability in conjunction with polynomial growth. As a canonical example,

\begin{example}
    The $d$-dimensional torus $\mathbb{T}_d=(\mathbb{Z}/w\mathbb{Z})^d$ is $c(\mathbb{T}_d) = d+1$ \textsf{uniformly-amenable}. 
\end{example}

We also require a basic fact on the number of connected components in bounded degree graphs.

\begin{fact}[Cluster counting, Lemma 5 of~\cite{Aliferis2008Accuracy}] \label{fact:clustering} Let $G = ([n], E)$ be a graph of degree $\leq z$, and fix $T\subset [n], |T|=t$. The number of subsets $S\subset [n], |S|=s$ such that $T\subseteq S$ and $S$ is a union of connected components in $G$, each of which contains at least one node in $T$, is at most $ e^{s} z^{s-t}$.
\end{fact}

\subsection{Markov Chains and Glauber Dynamics}
\label{section:markov_chain}

Consider a (discrete-time) ergodic Markov chain with transition matrix $P$, on a finite state space $\Omega$, reversible with respect to a stationary distribution $\pi$. To quantify its convergence given a generic starting point $x_0\in \Omega$, we compute the statistical distance
\begin{equation}
    d^{\mathsf{disc}}(x_0, t)  = \|\mathbb{P}[X_t^{x_0}] - \pi\|_1,
\end{equation}

\noindent where $X_0^{x_0}=x_0, X_1^{x_0}, \cdots, X_t^{x_0}$ describe the evolution of the chain $P$ starting from $x_0$.

We often interchange between continuous time and discrete time dynamics in our analysis. The continuous time Glauber dynamics on a graph $G$ assigns every vertex
an i.i.d. rate-1 Poisson clock; when the clock at $v\in [n]$ rings, replaces the spin $\sigma_v$ according to the conditional distribution given the neighboring spins.
\begin{fact}
    [\text{\cite[Theorem~20.3]{Levin2017MarkovCA}}]\label{fact:cont_disc}
    There exists an absolute constant $C$ such that $\forall x_0\in \Omega$,
    \begin{equation}
        d^{\mathsf{cont}}(x_0, C\cdot t)\leq d^{\mathsf{disc}}(x_0, n\cdot t) \leq  d^{\mathsf{cont}}(x_0, C^{-1}\cdot t).
    \end{equation}
\end{fact}

\noindent As such, our analysis focuses on the discrete dynamics.  The quantum analogs of these statements are immediate, and spelled out for completeness \cref{section:open_quantum_systems}.

A standard tool is the following path coupling lemma; we refer the reader to the survey \cite{guruswami2016rapidlymixingmarkovchains} for a review. 

\begin{fact}
    [Path Coupling \text{\cite[Lemma~6.3]{guruswami2016rapidlymixingmarkovchains}}]\label{fact:path_coupling} Let $\Gamma:\Omega\times\Omega\rightarrow [D]$ be an integer valued metric, defined by the shortest path length in some graph $H = (\Omega, E)$ over the state space. Suppose a coupling $\omega:\Omega\times\Omega\rightarrow \Omega\times\Omega$ exists for all adjacent pairs $(x_0, y_0)$ in $H$, satisfying
    \begin{equation}
        \mathbb{E}_\omega[\Gamma(X_1, Y_1)|X_0=x_0, Y_0=y_0] \leq (1-\alpha)\cdot \Gamma(x_0, y_0), 
    \end{equation}
    \noindent for some parameter $\alpha>0$. Then $\forall x\in \Omega,\,\,d^{\mathsf{disc}}(x, t) \leq D\cdot (1-\alpha)^t$.
\end{fact}

\subsection{Classical and Quantum Error-Correction}
\label{section:error_correction}

We dedicate this subsection to the basics of classical and quantum error correction, c.f. the book \cite{Lidar_Brun_2013}.

\begin{definition}
    An $[n, k, d]$ classical linear code $C$ of length $n$, dimension $k$ and alphabet size $2$ is a $k$-dimensional linear subspace $C\subseteq \mathbb{F}_2^n$. The distance of $C$ is $d = \min_{c\in C\setminus \{0\}} |c|$. $C$ can be specified by a parity check matrix $H\in \mathbb{F}_2^{m\times n}$, where $k = n - \mathsf{rank}_{\mathbb{F}_2} H$.
\end{definition}

\noindent Given any bit-flip error $e\in \{0, 1\}^n$, the \textit{syndrome} of $e$ is denoted as $\sigma(e)=He$. In this paper we will focus on a class of quantum error-correcting codes known as CSS \cite{Calderbank1996GoodQE, Steane1996SimpleQE} codes. 

\begin{definition}
    [CSS Codes] Given $H_X\in \mathbb{F}_{2}^{m_x\times n}$ and $H_Z\in \mathbb{F}_{2}^{m_z\times n}$ satisfying $H_XH_Z^T = 0$, consider the set of commuting $n$-qubit Pauli operators (referred to as stabilizers)

    \begin{equation}
        \mathcal{S}_X = \big\{X^{a} = X^{a_1}\otimes X^{a_2}\otimes \cdots \otimes X^{a_n} : a\in H_X\big\},\quad  \mathcal{S}_Z = \big\{Z^{b} : b\in H_Z\big\}
    \end{equation}
    \noindent where $a, b$ are rows of $H_X, H_Z$ respectively. The quantum error-correcting code $\mathsf{CSS}(H_X, H_Z)$ is the joint $+1$ eigenspace of $\mathcal{S}_X, \mathcal{S}_Z$.
\end{definition}

\noindent Given any bitstring $e\in \{0, 1\}^n$ associated to an $X$-error $X^e$, we denote as $\sigma_Z(e) = H_Ze$ its $Z$-syndrome. The support $\mathsf{supp}(a)\subset [n]$ of a parity check $a\in H$ (a row of $H$) are the symbols (bits or qubits) of the code that it acts on. 

\begin{definition}
    The $X$-distance $d_X$ of $\mathsf{CSS}(H_X, H_Z)$ is the minimum weight of any string $e$ such that $\sigma_Z(e) = H_Ze = 0$ and $e\notin \mathsf{rowspan}(H_X)$.
\end{definition}

In other words, the Pauli operator $X^e$ is undetected by all the $Z$ checks and isn't generated by the $X$ checks. The $Z$ distance is analogous. 

\begin{definition}
    The number of logical qubits $k$ encoded into $\mathsf{CSS}(H_X, H_Z)$ is $n - \mathsf{rank}_{\mathbb{F}_2}H_X-\mathsf{rank}_{\mathbb{F}_2}H_Z$.
\end{definition}

\noindent Put together, $\mathsf{CSS}(H_X, H_Z)$ is said to be a $[[n, k, d=\min(d_X, d_Z)]]$ quantum error-correcting code.

\subsection{Syndromes and their Redundancies}

Given a parity check matrix $H$ and a vector $s\in \{0, 1\}^m$, the set of solutions $x$ to $s=Hx$ is either empty or defines an affine subspace of dimension $2^k$. This is because the parity check matrix may not be full rank, and therefore not all vectors $s\in \{0, 1\}^m$ correspond to valid syndromes.

\begin{definition}
    [Valid Syndromes] Fix a parity check matrix $H\in \mathbb{F}_2^{m\times n}$. We refer to the image $\Omega = \mathsf{Im}(H)$ as the set of \textit{valid syndromes}. A \emph{meta-check matrix} $M\in \mathbb{F}_2^{t\times m}$ is a parity check matrix for valid syndromes, which satisfies $MH=0$ and thus $\mathsf{Im}(H)\subseteq \mathsf{Ker}(M)$.
\end{definition}

\begin{example}
    For instance, in both the 2D Ising model and 4D toric code, the metachecks enforce that the violated parity checks (supported on the edges of the lattice) define an even subgraph.  
\end{example}

The redundancies between syndrome vectors are captured by the \textit{meta-check} matrix $M$, which encodes the linear dependencies among the rows of $H$. Note that not all elements of $\mathsf{Ker}(M)$ are necessarily valid syndromes, which would entail $\mathsf{Ker}(M) = \mathsf{Im}(H)$.

We depict these dependencies via a graph $G_{M}$, referred to as the \textit{syndrome network}. 

\begin{definition}
    [The Syndrome Network]\label{definition:syndrome_network} Fix a parity check matrix $H\in \mathbb{F}_2^{m\times n}$ and a choice of meta-check matrix $M\in \mathbb{F}_2^{t\times m}$ s.t. $\mathsf{Im}(H)\subseteq \mathsf{Ker}(M)$. The syndrome network is a graph $G_M = ([m], E)$ where the vertices correspond to the parity checks of $H$, and two vertices $a, b\in [m]$ are connected in $E$ if and only if there exists a row in $M$ whose $a$-th and $b$-th entry are both 1.
\end{definition}

\begin{remark}
    The syndrome network defined above is only a Markov network of the Gibbs distribution over syndromes if $\mathsf{Ker}(M) = \mathsf{Im}(H)$; otherwise it is just a graphical description of some redundancies between the checks. We return to this point in \cref{section:conditions}.  
\end{remark}

An analogous description occurs in quantum CSS codes, to the $X$ and $Z$ parity check matrices independently. Each joint eigenspace of the stabilizers $\mathcal{S}_X, \mathcal{S}_Z$ is a subspace of dimension $2^k$. We index each of these joint eigenspaces using its $X$ and $Z$ syndromes. 

\begin{definition}
    [The Syndrome Vector] We associate a string $s = (s_X, s_Z)\in \mathsf{Im}(H_X)\times \mathsf{Im}(H_Z)$ to index each joint eigenspace of the collection $(\mathcal{S}_X, \mathcal{S}_Z)$ of $m_x+m_z$ commuting operators. We denote $\Pi_s$ as the orthogonal projection onto said eigenspace, 
    \begin{equation}
        \Pi_s = \prod_{i=1}^{m_x}\left(\frac{I+(-1)^{s_{X,i}}X^{a_i}}{2}\right) \cdot \prod_{j=1}^{m_z}\left(\frac{I+(-1)^{s_{Z,j}}Z^{b_j}}{2}\right),
    \end{equation}
where $a_i, b_j$ are rows of $H_X, H_Z$ respectively.
\end{definition}

\section{Dynamics of Self-Correcting Memories}
\label{section:self-correcting}

Here we describe our model of Glauber dynamics on error-correcting codes, and present definitions of self-correcting (classical or quantum) memories based on the stability of the dynamics.  We begin in \cref{section:code_hamiltonians} by describing how to associate a local Hamiltonian to any Low-Density-Parity-Check (LDPC) code. In \cref{section:glauber_codes} we describe the model of Glauber dynamics, on both the code-space and the syndrome-space. In \cref{section:memory_definition}, we describe the definitions of self-correction we use.

\subsection{The Gibbs States of Code Hamiltonians}
\label{section:code_hamiltonians}

Given a classical linear code specified by a parity check matrix $H\in \mathbb{F}_2^{m\times n}$, there is a canonical way to define a Hamiltonian: penalize (i.e. increase the energy of) each violated parity check.

\begin{definition}
    [Code Hamiltonian]\label{definition:code_hamiltonian} One can associate a (classical, diagonal) Hamiltonian $ \mathcal{H}_{\text{cl}}:(\mathbb{C}^2)^{\otimes n}\rightarrow (\mathbb{C}^2)^{\otimes n}$ to any linear code $H\in \mathbb{F}_2^{m\times n}$,
    \begin{equation}
        \mathcal{H}_{\text{cl}} =  \sum_{a\in H}\bigg(\frac{\mathbb{I} - Z^a}{2}\bigg),
    \end{equation}
    where $a\in H$ means indexing over row vectors of $H$. Analogously, we define a code Hamiltonian $ \mathcal{H}_{q}:(\mathbb{C}^2)^{\otimes n}\rightarrow (\mathbb{C}^2)^{\otimes n}$ of a quantum CSS code  $\mathsf{CSS}(H_X, H_Z)$:
\begin{equation}
    \mathcal{H}_q = \under{\sum_{a\in H_X}\bigg(\frac{\mathbb{I}-X^{a}}{2}\bigg)}{X \text{ checks}} + \under{\sum_{b\in H_Z} \bigg(\frac{\mathbb{I}-Z^{b}}{2}\bigg)}{Z \text{ checks}}.
\end{equation}
\end{definition}

\noindent These are positive, gapped, and commuting Hamiltonians, which are local in the sense that the support of the interaction terms is limited to the support of the parity checks.

\begin{fact}
    The eigenspaces of a given code Hamiltonian $\mathcal{H}$ are indexed by the valid syndromes of its parity check matrices. That is, for each valid syndrome $s\in \Omega$ one can define an eigenspace $\Pi_s$ of $\mathcal{H}$ of energy $|s|$, the Hamming weight of the syndrome vector. In other words,
    \begin{equation}
        \mathcal{H} = \sum_{s\in\Omega}\left|s\right|\cdot \Pi_s.
    \end{equation}
\end{fact}

This description of the eigenspaces of $\mathcal{H}$ gives us an explicit characterization of its Gibbs state.

\begin{fact}\label{fact:gibbs_fragmentation}
    Fix an inverse temperature $\beta$. The Gibbs states of code Hamiltonians are mixtures (convex combinations) of syndrome subspaces:
\begin{equation}
    \rho_\beta = \frac{e^{-\beta \mathcal{H}}}{ \Tr[e^{-\beta \mathcal{H}}]} =  \sum_{s\in \Omega} \bigg(\frac{\Pi_s}{2^{k}}\bigg) \times \pi_\beta(s)\,, \text{ where }\pi_\beta(s)=\frac{e^{-\beta |s|} }{\sum_{s'\in\Omega}e^{-\beta |s'|} }\,.
\end{equation}
\noindent We refer to the probability distribution $\pi_\beta(s)$ as the \emph{syndrome Gibbs distribution}.
\end{fact}

As we discuss shortly, the static properties of the syndrome Gibbs distribution $\pi_\beta(s)$ will play a crucial role in understanding the dynamical properties of the classical/quantum memory.

\begin{remark}
    The syndrome Gibbs distribution of a CSS code factorizes, i.e.,
    \begin{equation}
        \pi_\beta(s_X, s_Z) = \pi_\beta^X(s_X)\cdot \pi_\beta^Z(s_Z).
    \end{equation}
\end{remark}

\cref{fact:gibbs_fragmentation} captures a tensorization of the Gibbs state into two components. The first component is the classical distribution over syndromes. The second component is the logical component (where quantum information may be stored), and in some sense corresponds to encoding a random message into the $s$-syndrome space. 

\begin{example}
    In the 2D Ising model, $\pi_\beta(s)$ is the distribution over domain wall configurations (interfaces between 0's and 1's). $\frac{1}{2}\Pi_s$ captures the uniform distribution over two strings where the bits are filled in consistent with $s$ being the domain walls, with a fixed root vertex being $0$ or $1$.\footnote{Once one vertex is fixed, if the graph is connected the remainder can be filled in just using the domain wall knowledge.}
\end{example}

\subsection{Glauber Dynamics on Error Correcting Codes}
\label{section:glauber_codes}

\begin{figure}[t]
\centering
\begin{tcolorbox}
\textbf{Glauber dynamics.} Fix a parity check matrix $H\in \mathbb{F}_2^{m\times n}$, and let $x\in \{0, 1\}^n$ be the current configuration.

\begin{enumerate}
    \item Pick a bit $i\in [n]$ uniformly at random. 

    \item Compute the change in syndrome weight in applying a bit-flip error at location $i$,  
    \begin{equation}
        \Delta E = |H(x\oplus e_i)| - |Hx|.
    \end{equation}
    
    \item Apply the bit-flip $x'\leftarrow x\oplus e_i$ with probability given by a Boltzmann factor, 
    \begin{equation}
        p_{\mathsf{Accept}} = \frac{1}{1+e^{\beta \cdot \Delta E}}.
    \end{equation}
\end{enumerate}
\end{tcolorbox}
\begin{tcolorbox}
\textbf{Syndrome Markov chain.} Fix a parity check matrix $H\in \mathbb{F}_2^{m\times n}$ with valid syndromes $\Omega = \mathsf{Im}(H)$, and let $s\in\Omega$ be the current configuration.

\begin{enumerate}
    \item Pick a bit $i\in [n]$ uniformly at random. 

    \item Compute the change in syndrome weight in applying a bit-flip error at location $i$,  
    \begin{equation}
        \Delta E = |s\oplus (H e_i)| - |s|.
    \end{equation}
    
    \item Apply the update $s'\leftarrow s\oplus (H e_i)$ with probability given by a Boltzmann factor, 
    \begin{equation}
        p_{\mathsf{Accept}} = \frac{1}{1+e^{\beta \cdot \Delta E}}.
    \end{equation}
\end{enumerate}
\end{tcolorbox}
\caption{Single-site Glauber dynamics on a classical linear code, and its induced syndrome Markov chain.}
\label{fig:classical-glauber}
\end{figure}

We now elaborate on the particular dynamics which define classical and quantum memories. The simplest model of thermal stability \cite{Alicki2008OnTS, Bombin2009SelfcorrectingQC} is to study the evolution of the system under Glauber dynamics, at a sufficiently small constant temperature. Let us begin with the definition for classical linear codes. 

\begin{definition}[Glauber dynamics]\label{definition:classical_glauber}
    Fix a parity check matrix $H\in \mathbb{F}_2^{m\times n}$. The state space is the set of error configurations $\{0, 1\}^n$; the model of discrete-time Glauber dynamics $\mathsf{P}_{\mathsf{G}}:\mathbb{R}^{\{0, 1\}^n}\rightarrow \mathbb{R}^{\{0, 1\}^n}$ under single-site bit flips on the $n$ symbols of the code is depicted in \cref{fig:classical-glauber}. We denote the corresponding generator of the continuous time Markov semi-group as $\mathsf{L}_{\mathsf{G}}:\mathbb{R}^{\{0, 1\}^n}\rightarrow \mathbb{R}^{\{0, 1\}^n}$.
\end{definition}

By linearity, we note that the evolution (the transition probabilities) in \cref{fig:classical-glauber} doesn't depend directly on the current configuration $x$, instead only on its current syndrome $Hx$. It is therefore instructive to instead consider the induced Markov chain on the syndromes of the code. 

\begin{definition}[Syndrome Markov chain]
    Fix a parity check matrix $H\in \mathbb{F}_2^{m\times n}$. The discrete-time syndrome Markov chain $\mathsf{P}_{\text{syn}}:\mathbb{R}^\Omega\rightarrow \mathbb{R}^\Omega$ (and continuous, $\mathsf{L}_{\text{syn}}$) is defined by the induced Glauber dynamics on the set of valid syndromes $\Omega = \mathsf{Im}(H)$, shown in \cref{fig:classical-glauber}. 
\end{definition}

Definitions of thermal stability in the quantum setting depend directly on the noise model being studied. In this paper, we adopt the ``$X$ or $Z$" process $\mathcal{L} = \mathcal{L}_X+\mathcal{L}_Z$ \cite{Alicki2008OnTS}, corresponding to the Davies generator with single-qubit Pauli $X$ and $Z$ jump operators. For clarity and conciseness, we defer a detailed review of the dynamics of open quantum systems to \cref{section:open_quantum_systems}. Here we simply note that this model cleanly reduces to two independent Glauber dynamics, being run on the $X$ and $Z$ parity check matrices. 

\begin{remark}
    As a matter of notation, generators of quantum Markov semi-groups (italic) $\mathcal{L}:\mathcal{B}(\mathbb{C}^d)\rightarrow \mathcal{B}(\mathbb{C}^d)$ act on the space of density matrices, with $d$ the Hilbert space dimension. The classical analogs $\mathsf{L}:\mathbb{R}^\Omega\rightarrow \mathbb{R}^\Omega$ generate dynamics over classical probability distributions over $\Omega$.
\end{remark}

\begin{definition}
    Let $H_X\in \mathbb{F}_2^{m_X\times n}$, $H_Z\in \mathbb{F}_2^{m_Z\times n}$ define the code $\mathsf{CSS}(H_X, H_Z)$. Individually, we associate generators $\mathsf{L}_X, \mathsf{L}_Z$ to the continuous-time single-site Glauber dynamics over $H_X$ (with $Z$ errors), and $H_Z$ (with $X$ errors) as in \cref{fig:classical-glauber}. Further, we denote $\mathsf{L}_{\text{syn}}^X, \mathsf{L}_{\text{syn}}^Z$ as the Markov chain generators over $X$ and $Z$ syndromes. 
\end{definition}

\subsection{Definitions of Self-Correcting Memories}
\label{section:memory_definition}

A self-correcting memory, at a high level, provides an encoding of classical or quantum information that is stable against thermal noise. This robustness is quantified in the probability a logical failure occurs during the noise process.

\begin{definition}\label{definition:cl-self-correcting}
    A classical error correcting code associated with encoding and decoding channels $(\mathsf{Enc}, \mathsf{Dec})$ is said to be $(t, \epsilon)$ self-correcting against a noise model specified by a generator $\mathsf{L}$ if, 
    \begin{equation}
       \forall x\in \{0, 1\}^k:\quad  x'\leftarrow \mathsf{Dec}\circ e^{t\mathsf{L}}\circ \mathsf{Enc}(x), \quad \mathbb{P}[x'\neq x] \leq \epsilon.
    \end{equation}
\end{definition}

In words, after evolving the system under Glauber dynamics for time $t$, the probability of imparting a logical error after decoding is bounded by $\epsilon$. The analogous quantum definition  below stipulates that after encoding any $\ket{\psi}$, the logical information in the resulting state survives for time $t$, in the sense that $\ket{\psi}$ is recoverable by a decoding channel. 

\begin{definition}\label{definition:self-correcting}
    A quantum error correcting code associated with encoding and decoding channels $(\mathsf{Enc}, \mathsf{Dec})$ is said to be $(t, \epsilon)$ self-correcting against a noise model specified by a Lindbladian $\mathcal{L}$ if
    \begin{equation}
       \forall \ket{\psi} \in (\mathbb{C}^2)^{\otimes k}:\quad  \big\| \mathsf{Dec}\circ e^{t\mathcal{L}} \circ \mathsf{Enc}(\ketbra{\psi})- \ketbra{\psi}\big\|_1 \leq \epsilon.
    \end{equation}
\end{definition}

There are two \textit{in-equivalent} ways to interpret the encoding maps in \cref{definition:self-correcting}, which hinge on precisely which state we are encoding the message into. A standard interpretation is that $\mathsf{Enc}$ is the encoding map of the error correcting code, which encodes the logical information into a \emph{ground state} of the code Hamiltonian. Unfortunately, as discussed the stability of the ground state under the dynamics is extremely delicate to prove in general.

Instead, most provable results in the literature are in the setting where $\mathsf{Enc}$ encodes logical information directly into a \emph{metastable state}, whose syndrome is in thermal equilibrium with the noise model~\cite{Alicki2008OnTS, Bombin2009SelfcorrectingQC}. We refer to this state as the \textit{Gibbs state within a logical sector}. %

\subsection{Gibbs States within a Logical Sector}
\label{sec:gibbslogical}

The notion of a Gibbs state within a logical sector should capture a state (approximately) in thermal equilibrium, restricted to the local energy landscape around a specific ground state (codeword). There are, however, subtleties to a general definition for all CSS/classical linear codes. Here, we give a simple definition (based on \cite{Alicki2008OnTS}), which uses the syndrome Gibbs distribution.\footnote{As mentioned in \cref{section:discussion} and later in \cref{section:conditions}, it is limited to codes with no ``local minima" of very low energy.}
To begin, we define the decoding channel $\mathsf{Dec}$. Here, there already is some flexibility to the definition: we need to be to able to associate a ``canonical" correction/error, to any given syndrome. 

\begin{definition}
    [Canonical decoder]\label{def:canonical_decoder} Fix a parity check matrix $H\in\mathbb{F}_2^{m\times n}$. A canonical decoder is a mapping $\mathsf{corr}:\mathsf{Im}(H)\rightarrow \mathbb{F}_2^{n}$ from a valid syndrome to an error that is consistent with the syndrome, i.e. $s=H\cdot \mathsf{corr}(s)$.
    
    The definition extends analogously to syndromes $s \in \mathsf{Im}(H_X)\times \mathsf{Im}(H_Z)$ of quantum $\mathsf{CSS}(H_X,H_Z)$ codes, where $\mathsf{corr}:\mathsf{Im}(H)\rightarrow \{\mathbb{I}, X, Y, Z\}^{\otimes n}$ is some canonical Pauli error consistent with $s$.  
\end{definition}

\begin{example}
    This canonical error could be chosen to be the minimum weight error of syndrome $s$, $\mathsf{corr}_{\mathsf{minw}}(s)=\arg\min\{|e|:He=s\}$; ties are broken arbitrarily.
\end{example}

\begin{example}\label{example:homological}
    In the 4D toric code the common choice is $\mathsf{corr}_{\mathsf{contract}}$, which first individually contracts all topologically trivial loops according to the minimum weight enclosing surface of errors \cite[Sec.~V]{Alicki2008OnTS}. This preserves the homology class of the error configuration. Topologically non-trivial loops are then corrected using the minimum weight error (see also \cref{example:topo_non_trivial}). 
\end{example}

As we later discuss, it is a generalization of \cref{example:homological} that we adopt in the paper. 

\begin{remark}
    The decoding channel $\mathsf{Dec}$ is a (hypothetical) noiseless process defined by first measuring the syndrome $s$, and subsequently applying $\mathsf{corr}(s)$ -- thus returning the string (state) to the code-space, of syndrome $0$. The \textit{encoded} string (state) can now be extracted by decoding the codeword (code-state) via Gaussian elimination (or reverting the encoding unitary \cite{got97}). 
\end{remark}
    
The Gibbs distribution associated with a parity check matrix $H$ can be understood as:
\begin{enumerate}
    \item Sample a uniformly random codeword $c$  (s.t. $Hc=0$).
    \item Sample a random syndrome $s\leftarrow\pi_\beta(\cdot)$.
    \item Output $c\oplus \mathsf{corr}(s)$.
\end{enumerate}
Note that the choice of the decoder $\mathsf{corr}$ does not matter due to the randomness of the codeword. The Gibbs states of quantum CSS codes are analogous, with the codewords replaced by random code-states $\ket{\psi}$. If instead we fix the choice of codeword/state, but still sample from the same error distribution as above, we arrive at the Gibbs state within a logical sector. Note that here the decoder $\mathsf{corr}$ does matter, since we are performing perturbations around a fixed codeword.

\begin{definition}[Gibbs state within a logical sector]\label{def:gibbs_within_logical} Fix a parity check matrix $H\in\mathbb{F}_2^{m\times n}$, and let $\pi_\beta$ denote its syndrome Gibbs distribution. Then, for any codeword $c\in \{0, 1\}^n$, we refer to $\rho^c_{\beta}$ as the \textit{Gibbs distribution within the logical sector} $c$, where a sample from $\rho^c_{\beta}(\cdot)$ is given by
\begin{equation}
    c\oplus \mathsf{corr}(s),\quad s\leftarrow \pi_\beta(\cdot).
\end{equation}
Analogously, consider a CSS code $\mathsf{CSS}(H_X,H_Z)$ with syndrome Gibbs distribution $\pi_\beta$ supported on valid syndromes $\Omega = \mathsf{Im}(H_X)\times \mathsf{Im}(H_Z)$. Then, for any code-state $\ket{\psi}\in (\mathbb{C}^2)^{\otimes n}$,  we refer to $\rho^\psi_{\beta}$ as the \textit{Gibbs state within the logical sector} $\ket{\psi}$, where
\begin{equation}
\rho^\psi_{\beta}=\sum_{s\in\Omega} \mathsf{corr}(s)\ketbra{\psi} \mathsf{corr}(s)\cdot  \pi_\beta(s).
\end{equation}
\end{definition}

\begin{example}
    In the context of the 2D Ising model with min-weight corrections ($\mathsf{corr}_{\mathsf{minw}}$), \cref{def:gibbs_within_logical} says that the Gibbs distribution within the logical sector $0^n$ is simply the Gibbs distribution conditioned on the majority of spins being $0$. 
\end{example}

 Suppose we now run Glauber dynamics starting from $\rho^\psi_{\beta}$. Since the syndrome distribution is already in equilibrium, only the logical information can change under the dynamics. That is, the state will leak to other logical sectors, eventually converging to the genuine Gibbs state. The goal in the design of a self-correcting quantum memory is to prove this leaking process is slow, i.e. $\rho^\psi_{\beta}$ is approximately stationary under Glauber dynamics.

\begin{remark}
      \cite{Alicki2008OnTS} proved the thermal stability of the 4D toric code (below a constant temperature threshold), when initialized from $\rho^\psi_{\beta}$. That is, 
      \begin{equation}
          \left\|e^{\mc Lt}[\rho^\psi_{\beta}]-\rho^\psi_{\beta}\right\|_1\leq t \cdot e^{-\delta L},
      \end{equation}
      where $\delta>0$ is a constant and $L$ is the side-length of the 4D torus. As previously noted, this result is not applicable when initialized from the ground state $\ket{\psi}$, as the ground state is not stationary under any (even restricted to a metastable sector) positive temperature dynamics.
\end{remark}

The thesis of this paper is that for many models of self-correcting memories, the two views on stability -- from the ground state, or from the Gibbs state within a sector -- are consistent. More precisely, the dynamics starting from a ground state $\ket{\psi}$ is expected to undergo two stages:
\begin{enumerate}
    \item Starting from the ground state $\ket{\psi}$, local dynamics \textit{quickly} converges to the corresponding Gibbs state within a logical sector $\rho_\beta^\psi$.
    \item The Gibbs state within a logical sector $\rho_\beta ^\psi$ is metastable under local dynamics, and \emph{slowly} converges to the true Gibbs state $\rho_\beta$ after exponential time.
\end{enumerate}

The main goal of this work is to provide a theoretical justification of the first stage, by proving (some form of) rapid mixing to Gibbs state within a logical sector starting from the ground-state.

\section{Sufficient Conditions for Self-Correction}
\label{section:conditions}

We dedicate this section to stating the central assumptions in this paper, which define the family of self-correcting classical and quantum memories our results are applicable to. We then prove two of their most basic properties. In more detail: 
\begin{itemize}
    \item In \cref{section:criterion}, we introduce the \textit{connectedness criterion} (\cref{definition:connectedness}), a sufficient condition for self-correction first introduced by \cite{Bombin2009SelfcorrectingQC}. 
\end{itemize}

\noindent Roughly speaking, this criterion captures codes whose parity-checks and meta-checks are both LDPC, and which further admit an ``energy barrier". Unfortunately, as we discuss, on its own the connectedness criterion is insufficient to establish rapid mixing from the ground state. We further require minor additional assumptions on the ability to locally correct low-weight syndromes.

\begin{itemize}
    \item In \cref{section:peierls}, we show how the connectedness criterion (\cref{definition:connectedness}) enables the application of a Peierls argument, which bounds the probability of topologically non-trivial arrangements of syndromes. 
\end{itemize}

\noindent This gives (essentially) a self-contained re-proof that  \cref{definition:connectedness} is sufficient for self-correction, when the quantum memory is initialized from the Gibbs state within a logical sector. 

\begin{itemize}
    \item In \cref{section:markov}, we leverage \cref{definition:connectedness} to prove a Markov property for the Gibbs state within a logical sector, via syndrome Gibbs distribution of these families of codes. 
\end{itemize}

The Hammersley-Clifford theorem stipulates that the Gibbs distributions of bounded classical Hamiltonians are Markov random fields, and thereby their conditional distributions factorize -- a property that extends to commuting quantum Hamiltonians \cite{Brown2012QuantumMN}. However, such a condition does not necessarily extend to the Gibbs state within a logical sector. Here we prove some form of an approximate local Markov property for these states at low temperatures.

\subsection{The Family of Codes}
\label{section:criterion}

We dedicate this subsection to clearly formalizing the assumptions on the family of codes we consider in this paper, which are adapted from the criterion for self-correction proposed by \cite{Bombin2009SelfcorrectingQC}. Broadly speaking, the assumptions we require make reference to:

\begin{itemize}
    \item The \emph{sparsity} of the parity checks and metachecks (\cref{definition:sparsity}).
    \item An \emph{energy barrier}, which quantifies the energy cost of imparting a logical error via a small connected component of syndromes (\cref{definition:syndrome_clustering}).
    \item The ability to \emph{locally correct} small components of syndromes (\cref{definition:locally_erasable}). 
\end{itemize}

We bundle these conditions together in \cref{definition:connectedness}. Let us begin with the simplest assumption, \cref{definition:sparsity}, on the sparsity of interactions. 

\begin{tcolorbox}
\begin{properties}
    [$(\ell, d)-\mathsf{Sparsity}$]\label{definition:sparsity} A parity check matrix $H\in \mathbb{F}_2^{m\times n}$ with meta-check matrix $M\in \mathbb{F}_2^{t\times n}$ and syndrome network $G = ([m], E)$ is said to satisfy $(\ell, d)-\mathsf{Sparsity}$, if
     \begin{enumerate}
        \item $H$ \textbf{is LDPC.} $H$ is $\leq \ell$ column and row sparse.
        \item $M$ \textbf{is LDPC.} $G$ has degree $\leq d$.
    \end{enumerate}
\end{properties}
\end{tcolorbox}

Any discussion on self-correction must include a mechanism for which the error-correcting process (its decoder) could fail. For this purpose, we begin with the concept of a \textit{critical syndrome}: a syndrome which is one single-site update away from imparting a logical error.

\begin{definition}
    [Critical Syndromes, \cite{Bombin2009SelfcorrectingQC}]\label{definition:critical} Fix $H\in \mathbb{F}_2^{m\times n}$, and a decoder $\mathsf{corr}:\mathsf{Im}(H)\rightarrow \mathbb{F}_2^{n}$. A syndrome $s\in \mathsf{Im}(H)$ is said to be $\mathsf{critical}$ with respect to $\mathsf{corr}$ if there exists a single bit-flip update $i\in [n]$, leading to a new syndrome $s'\leftarrow s\oplus (He^i)$, such that
    \begin{equation}
        \mathsf{corr}(s) \oplus e^i\oplus \mathsf{corr}(s') \neq 0^n.
    \end{equation}
    In other words, the decoder is not consistent between $s$ and $s'$. The extension to CSS codes $\mathsf{CSS}(H_X, H_Z)$ replaces the condition with checking if the effective resulting corruption $\mathsf{corr}(s) \oplus P_i\oplus \mathsf{corr}(s')$ (for some single qubit Pauli $P_i$) is not generated by the code stabilizers. 
\end{definition}

Whenever the decoder $\mathsf{corr}$ is implicit, we simply say $s$ is $\mathsf{critical}$.

\begin{example}
    In a classical code under the min-weight correction function $\mathsf{corr}_{\mathsf{minw}}$, a syndrome is \textsf{critical} if there exists a consistent error which is almost equally distant to two codewords.
\end{example}

During the dynamics of the memory, errors and syndromes will arise and propagate in far-apart regions of the code. To quantify the independence among different subsets of syndromes, we further need to introduce the concept of an \textit{erasable} subset of syndromes.

\begin{definition}
[Erasable Subsets of Syndromes] \label{definition:erasable}
Fix a parity check matrix $H\in\mathbb{F}_2^{m\times n}$. A subset of violated syndromes $\mathsf{E}\subset [m]$ is \textsf{erasable} if  $\mathbf{1}_\mathsf{E}\in \mathsf{Im}(H)$ is a valid syndrome ($\mathbf{1}_\mathsf{E}\in\{0,1\}^m$ is a vector whose $i$-th entry is 1 if $i\in\mathsf{E}$ and 0 otherwise).
\end{definition}

In words, there is an error which is consistent with that subset of syndromes. Suppose $\mathsf{E}$ is an \textsf{erasable} subset, and there is a valid syndrome $s\in \mathsf{Im}(H)$ that contains $\mathsf{E}$. Then, we can apply an error (e.g. $\mathsf{corr}(\mathbf{1}_\mathsf{E})$)) resulting in a new syndrome $\Tilde{s}$ where $\mathsf{E}$ is erased from $s$.\footnote{That is, $\Tilde{s}_{\mathsf{E}}=0$, $\Tilde{s}_{[m]\setminus \mathsf{E}}=s_{[m]\setminus \mathsf{E}}$.} However, note that \cref{definition:erasable} does not make reference to the \textit{locality} (support size) of the erasing process. 

\begin{example}
    In the 2D Ising model, a topologically trivial domain wall can be erased by flipping all the spins in its interior. Similarly, in the 4D toric code, topologically trivial loops of syndromes can be removed by applying a Pauli correction to the ``interior'' of the loop.
\end{example}

We can now introduce the next set of assumptions, \cref{definition:syndrome_clustering} , which roughly speaking will constrain how syndromes cluster together in the meta-check graph $G$.

\begin{tcolorbox}
\begin{properties}
    [$(\chi, \nu)-$\textsf{Syndrome Clustering}]\label{definition:syndrome_clustering} Fix $s\in \mathsf{Im}(H)$ and a decoder $\mathsf{corr}:\mathsf{Im}(H)\rightarrow \mathbb{F}_2^n$. The violated syndromes in $s$ can be decomposed $\mathsf{V}_1\cup \mathsf{V}_2\cup \cdots \subset [m]$ into disjoint, non-adjacent, $\mathsf{erasable}$ subsets in $G$, satisfying:
    \begin{enumerate}
        \item \textbf{Energy Barrier.} If each $\mathsf{V}_i\subset [m]$ is small,  $|\mathsf{V}_i|\leq m^\chi$ for some $\chi<1$, then $s$ is not \textsf{critical} with respect to decoder $\mathsf{corr}$ of \cref{definition:generalized_contract}.
        \item $\nu$ \textbf{Connected Components are Erasable.} Each $\mathsf{V}_i$ can further be decomposed into at most $\nu$ connected components of syndromes in $G$. 
    \item \textbf{Small Connected Components are Individually Erasable.} Moreover, if $|\mathsf{V}_i|\leq m^\chi$, then each connected component of $\mathsf{V}_i$ is \textsf{erasable}.
    \end{enumerate}
\end{properties}
\end{tcolorbox}

\noindent The decoding function adopted in this paper is a generalization of the ``loop contraction" decoder of \cref{example:homological}.

\begin{definition}
    [Generalized $\mathsf{corr}_{\mathsf{contract}}$]\label{definition:generalized_contract} Given a syndrome $s$ decomposed into a collection of \textsf{erasable} subsets $\{\mathsf{V}_i\}_i$, we define $\mathsf{corr}(s)$ to sequentially erase each $\mathsf{V}_i$ according to the min weight error consistent with $\mathbf{1}_{\mathsf{V}_i}$.
\end{definition}

Let us unpack the \textsf{Syndrome Clustering} conditions above, appealing to the 2D Ising model and the 4D toric code as examples. The first of which is a standard ``energy barrier" condition.

\begin{example}
    In the 2D Ising model, if all the connected components of syndromes (i.e. domain walls/loops) are of size $\ll n^{1/2}$, then the encoded bit defined by the homology class of the spin configuration can't change under a single bit-flip. Thus, this syndrome configuration is not $\mathsf{critical}$. 
\end{example}

Most relevant to us are \cref{definition:syndrome_clustering} items 2 and 3, which roughly speaking stipulate that components of syndromes can be individually corrected.\footnote{As we extensively discuss later, this is the condition which enables an application of a Peierls-style argument \cite{Alicki2008OnTS}, to understand the large deviation probabilities of critical syndromes. As pointed out by \cite{Bombin2009SelfcorrectingQC}, the following discussion on topologically trivial/non-trivial configurations was overlooked in \cite{Alicki2008OnTS}.} At first glance, the stipulation that $\mathsf{V}_i$ is a collection of at most $\nu$ connected components (instead of just 1) may seem a bit odd. This is in order to address the distinction between \textit{topologically-trivial}, and \textit{topologically-non-trivial} syndrome configurations, which cannot be individually ``contracted" using local bit flips. 

\begin{example}\label{example:topo_non_trivial}
    Consider the Ising model on a 2D torus/doughnut, with two topologically non-trivial loops marking domain walls between regions of 0's and 1's (half-doughnut). These loops can only be removed if they annihilate each other.   
\end{example}

\cite{Bombin2009SelfcorrectingQC} proved that $(\ell, d)-\mathsf{Sparsity}$ (\cref{definition:sparsity}) as well as the first two conditions of $(\chi, \nu)-$\textsf{Syndrome} \textsf{Clustering} (\cref{definition:syndrome_clustering}) are sufficient to establish \textit{self-correction}, up to time-scales exponential in $m^\chi$. Unfortunately, these abstract conditions fall just slightly short of being sufficient for \textit{rapid mixing within a logical sector}, which is the goal of this work. Intuitively, this is since the conditions above make reference to how syndromes organize and cluster into small connected components, but they do not quantify the locality of the bit-flips errors which correct those syndromes. Accordingly, we add the following assumption:

\begin{tcolorbox}
\begin{properties}
    [\textsf{Locally Erasable Syndromes}]\label{definition:locally_erasable} Let $\mathbf{1}_{\mathsf{V}}\in \mathsf{Im}(H)$ denote a syndrome, supported on some \textsf{erasable} connected component $\mathsf{V}\subset [m]$ contained in some ball $\mathsf{V}\subset \mathcal{B}^r$ in $G$. Then its canonical correction $\mathsf{corr}(\mathbf{1}_{\mathsf{V}})\in \{0, 1\}^n$ is also contained in some ball  $\textsf{supp}(\mathcal{B}^r)\subset [n]$.
\end{properties}
\end{tcolorbox}

\begin{example}
    If the syndromes of the 4D toric code are loops confined to some subcube of the torus, then one can erase those syndromes by acting only on the subcube.
\end{example}

We bundle these conditions together and refer to them as the \textit{connectedness criterion}.

\begin{tcolorbox}
\begin{properties}
    [The Connectedness Criterion] \label{definition:connectedness} A parity check matrix $H\in \mathbb{F}_2^{m\times n}$ with associated syndrome network $G=([m], E)$ is said to satisfy $(\ell, d, \nu, \chi)-\mathsf{CC}$, if it satisfies $(\ell, d)-\mathsf{Sparsity}$, $(\chi, \nu)-$\textsf{Syndrome Clustering}, and has \textsf{Locally Erasable Syndromes}.
\end{properties}
\end{tcolorbox}

We remark that generic self-correcting classical and quantum memories on finite dimensional lattices, are expected to satisfy \cref{definition:connectedness}. Here we highlight some examples.

\begin{fact}
     Examples of codes satisfying the connectedness criterion include:
    \begin{enumerate}
        \item The 2D Ising model satisfies $(4, 6, 2, \frac{1}{2})-\mathsf{CC}$.
        \item \cite{Zhou2025FinitetemperatureQT} The Z checks of the 3D fermionic toric code satisfy $(4, 12, 4, \frac{1}{3})-\mathsf{CC}$.
        \item \cite{Dennis2001TopologicalQM} The X or Z checks of the 4D toric code satisfy $(6, 14, 4, \frac{1}{4})-\mathsf{CC}$.
        \item \cite{Bombin2009SelfcorrectingQC} The X or Z checks of the 6D Color code satisfy $(\ell, d, \nu, \frac{1}{6})-\mathsf{CC}$, for explicit constants $\ell, d, \nu$ which depend on the choice of lattice tesselation. 
    \end{enumerate}
\end{fact}

See \cref{section:toric_background} for a discussion.

The conditions above are merely sufficient for self-correction. As commented in \cref{section:discussion}, quantum expander codes violate both \textsf{Syndrome-Clustering} and \textsf{Locally Erasable Syndromes}, and thus are not captured in this framework. 

\subsection{The Peierls Argument}
\label{section:peierls}

In this section, we present a standard Peierls argument, which gives us control over the probability of large connected components of syndrome configurations. In particular, this will allow us to claim that all connected components of syndromes are small, and thereby \textsf{erasable}, with high probability. 

In the next subsection, we reason that conditioning on said event allows us to prove a Markov property for the syndrome Gibbs distribution. 

\begin{lemma}
    [The Peierls Argument]\label{lemma:peierls_wsm} Fix $H\in \mathbb{F}_2^{m\times n}$ satisfying $(\ell, d, \nu, \chi)-\mathsf{CC}$ (\cref{definition:connectedness}) and an inverse temperature $\beta$. Consider a subset of violated syndromes $\mathsf{V}\subset [m]$ that is \textsf{erasable}. Denote the collection of all valid syndrome configurations which contain $\mathsf{V}$ as $\mathsf{Consist(\mathsf{V})}\subseteq \mathsf{Im}(H)$. The probability a sample from $\pi_\beta(\cdot)$ contains $\mathsf{V}$ is
    \begin{equation}
        \mathbb{P}_{X\leftarrow \pi_\beta}\big[X\in \mathsf{Consist(\mathsf{V})}\big] \leq e^{-\beta |\mathsf{V}|}.
    \end{equation}
\end{lemma}

\begin{example}
    In the context of the 2D Ising model, this is asking for the probability that a given sample from $\pi_\beta$ contains a specific collection of domain walls (including topologically non-trivial ones). 
\end{example}

\begin{proof}
    Since $\mathsf{V}$ is \textsf{erasable}, consider the set $\mathsf{Consist(\mathsf{V})}^*$ of syndrome configurations, given by removing $\mathsf{V}$ from each syndrome in $\mathsf{Consist(\mathsf{V})}$:
    \begin{equation}
        \mathsf{Consist(\mathsf{V})}^* = \left\{ K\setminus \mathsf{V}: K\in \mathsf{Consist(\mathsf{V})}\right\}.
    \end{equation}
By definition, we are guaranteed $\mathsf{Consist(\mathsf{V})}^*\subset \mathsf{Im}(H)$ is valid and one-to-one with $\mathsf{Consist(\mathsf{V})}$. Further, the energy of $K \setminus V$ differs from that of $K$ by $|\mathsf{V}|$. Therefore, 
    \begin{align}
        \mathbb{P}_{X\leftarrow \pi_\beta}\big[X\in \mathsf{Consist(\mathsf{V})}\big] & = \frac{\sum_{K\in  \mathsf{Consist(\mathsf{V})}} e^{-\beta |K|} }{\sum_{K\in \mathsf{Im}(H)}e^{-\beta|K|}} \\  &\leq  \frac{\sum_{K\in   \mathsf{Consist(\mathsf{V})}} e^{-\beta |K|}}{\sum_{K\in  \mathsf{Consist(\mathsf{V})}} e^{-\beta |K|}+\sum_{K\in \mathsf{Consist(\mathsf{V})}^*} e^{-\beta |K|}} \leq \frac{1}{1+e^{\beta |\mathsf{V}|}}\leq e^{-\beta |\mathsf{V}|}. \nonumber \qedhere 
    \end{align}
\end{proof}

We can now re-prove a lemma essentially present in \cite{Bombin2009SelfcorrectingQC} on the maximum size of connected components of syndromes. 

\begin{claim}
    [All Syndrome Configurations are Topologically Trivial]\label{claim:all-erasable} In the context of \cref{lemma:peierls_wsm}, let $|\mathsf{LCC}|$ denote the size of the largest connected component of syndromes in the sample $X$. Then, so long as $\beta > 1+\ln d $,
    \begin{equation}
        \mathbb{P}_{X\leftarrow \pi_\beta}[|\mathsf{LCC}| \geq m^\chi] \leq \frac{m^\nu (de^{1-\beta})^{m^\chi}}{1-de^{1-\beta}}.
    \end{equation}
\end{claim}

\begin{remark}
    This claim is sufficient to (re-)establish the stability of $H$ under Glauber dynamics, when the syndrome distribution is initialized from $\pi_\beta$. In other words, \cref{definition:connectedness} is indeed sufficient for self-correction. 
\end{remark}

As we discuss shortly, \cref{claim:all-erasable} and \cref{definition:syndrome_clustering} further implies that with high probability, all connected components of syndrome configurations are individually \textsf{erasable}; i.e. there are no topologically non-trivial syndrome configurations. 

\begin{corollary}\label{corollary:all-erasable}
    We denote as $\mathcal{E}_{\mathsf{eras}}$ the event that all the connected components of violated syndrome in a sample from $\pi_\beta$ are individually \textsf{erasable}. Then, so long as $\beta > 1+\ln d $, 
    \begin{equation}
        \mathbb{P}[\mathcal{E}_{\mathsf{eras}}] \geq 1-\mathbb{P}_{X\leftarrow \pi_\beta}[|\mathsf{LCC}| \geq m^\chi] \geq 1- \frac{m^\nu (de^{1-\beta})^{m^\chi}}{1-de^{1-\beta}}.
    \end{equation}
\end{corollary}

\begin{proof}[Proof of \cref{claim:all-erasable}] 
    For any sample $X$, there exists a set of violated syndromes $\mathsf{V}\subset [m]$ contained in $X$ such that: (1) $\mathsf{V}$ is \textsf{erasable}; (2) $\mathsf{V}$ contains the largest connected component in $X$;  (3) $\mathsf{V}$ can be decomposed into at most $\nu$ connected components (\cref{definition:syndrome_clustering}.2). From \cref{definition:syndrome_clustering}.3, it suffices to union bound over all such sets $\mathsf{V}$ of size at least $m^\chi$. 
    \begin{align}
        \mathbb{P}_{X\leftarrow \pi_\beta}[|\mathsf{LCC}| \geq m^\chi] &\leq \sum_{|\mathsf{V}|\geq m^\chi} \mathbb{P}[X\in \mathsf{Consist}(\mathsf{V})] \leq \sum_{|\mathsf{V}|\geq m^\chi} e^{-\beta |\mathsf{V}|} \\
        &\leq m^\nu\times \sum_{\ell\geq m^\chi} (ed)^\ell\cdot e^{-\beta \ell} \leq \frac{m^\nu (ede^{-\beta})^{m^\chi}}{1-ede^{-\beta}},
    \end{align}
    \noindent so long as $de^{1-\beta}<1$. In the above, in the second inequality we applied \cref{lemma:peierls_wsm}, and subsequently we count configurations of total size $\ell$, comprised of $\leq \nu$ connected components using \cref{fact:clustering}.
\end{proof}

\subsection{A Markov Property for the Syndrome Gibbs Distribution}
\label{section:markov}

As discussed, the Gibbs distribution/state of a commuting Hamiltonian obeys the Hammersley-Clifford theorem \cite{Brown2012QuantumMN}. Thereby, it admits the Markov property, in the sense that the state factorizes into conditional distributions across \textit{screening tripartitions}:

\begin{definition}[Screening Tripartitions] \label{definition:screens}
    Fix a graph $G = ([m], E)$, and a tripartition $\mathsf{A}\cup \sB\cup \sC = [m]$ of the vertices. We say $\sB$ screens $\sA$ from $\sC$ in $G$ if the distance $d(\sA, \sC)>1$; in other words the subgraph $G|_{\sA\cup \sC}$ with the vertices of $\sB$ removed disconnects $\sA$ and $\sC$.
\end{definition}

\cite{Brown2012QuantumMN} nevertheless does not imply that the Gibbs state \textit{within a sector} (\cref{def:gibbs_within_logical}) admits a Markov property; nor does it give a characterization of the syndrome Gibbs distribution. Indeed, although the syndrome Gibbs distribution can be understood as described by an effective Hamiltonian, the terms may have infinite weight, and thereby not captured by \cite{Brown2012QuantumMN}.

Nevertheless, here we show how to recover a Markov property, so long as one conditions on the syndromes in $\sA, \sC$ being sufficiently ``localized" and non-propagating through $\sB$.

\begin{lemma}
    [A Markov Property for the Syndrome Gibbs Distribution]\label{lemma:markov}
    Fix a parity check matrix $H\in \mathbb{F}_2^{m\times n}$ and an associated syndrome network $G = ([m], E)$. Let $\sA\cup \sB\cup \sC = [m]$ be a tri-partition such that $\sB$ screens $\sA$ from $\sC$ in $G$. Then the syndrome Gibbs distribution is Markov, conditioned on $\mathcal{E}_{\mathsf{eras}}$ and pinning the syndromes in $\sB$ to $0$. That is, $\forall \beta$ and $(s_\sA, s_\sB, s_\sC)\in \mathcal{E}_{\mathsf{eras}}$:

    \begin{equation}
   \pi_\beta(s_\sA, s_\sC|\mathcal{E}_{\mathsf{eras}}\wedge s_\sB = 0) = \pi_\beta(s_\sA|\mathcal{E}_{\mathsf{eras}}\wedge s_\sB = 0)\times \pi_\beta(s_\sC|\mathcal{E}_{\mathsf{eras}}\wedge s_\sB = 0).
\end{equation}
\end{lemma}

\cref{lemma:markov}, together with \cref{claim:all-erasable} on the event $\mathcal{E}_{\mathsf{eras}}$, imply that the syndrome Gibbs distribution is statistically close to a distribution $\pi_\beta(\cdot|\mathcal{E}_{\mathsf{eras}})$ which admits (some form of) a Markov property. Before the proof of \cref{lemma:markov}, we discuss the role of the two ``non-standard" conditions above: the conditioning on \textsf{erasable} syndromes $s\in \mathcal{E}_{\mathsf{eras}}$, and pinning the screening boundary $\sB$ to $0$ (as opposed to any $s_\sB\in \{0, 1\}^{|\sB|}$). We first reason the Markov property cannot hold \textit{exactly} in the syndrome Gibbs distribution.  

\begin{example}
    Consider a 4-partition of the 2D Ising model on a torus, into cylindrical regions $\sA\cup \sB_1\cup \sC\cup \sB_2$, and let us pin syndromes in $\sB=\sB_1\cup \sB_2$ to $0$. A single topologically non-trivial domain wall in $\sC$ implies an odd number of topologically non-trivial domain walls in $\sA$ (a long range correlation!). The Markov property therefore cannot hold exactly (w/o conditioning). 
\end{example}

See the remark below the proof for a comment on the generalization to other pinnings $s_\sB\neq 0$.

\begin{proof}[Proof of \cref{lemma:markov}] Consider any two syndrome configurations $(s_\sA, s_\sB=0, s_\sC), (s_\sA', s_\sB=0, s_\sC')\in \mathcal{E}_{\mathsf{eras}}$ which are (valid, and) consisting of only \textsf{erasable} connected components of syndrome violations. We claim (and prove shortly) that $s_\sA$ is also a valid assignment of syndromes in $\sA$ if ``completed" by $s'_\sC$, that is:
    \begin{equation}
        (s_\sA, s_\sB=0, s_\sC')\quad \text{ and }\quad  (s_\sA', s_\sB=0, s_\sC)\in \mathsf{Im}(H)
    \end{equation}
    are both valid and \textsf{erasable}. Then,
    \begin{equation}
        \frac{\pi_\beta(s_\sA, s_\sC|\mathcal{E}_{\mathsf{eras}}\wedge s_\sB = 0)}{\pi_\beta(s_\sA', s_\sC|\mathcal{E}_{\mathsf{eras}}\wedge s_\sB = 0)} = \frac{e^{-\beta(|s_\sA|+|s_\sC|)}}{e^{-\beta(|s_\sA'|+|s_\sC|)}} = e^{-\beta(|s_\sA|-|s_\sA'|)} =\frac{\pi_\beta(s_\sA|\mathcal{E}_{\mathsf{eras}}\wedge s_\sB = 0)}{\pi_\beta(s_\sA'|\mathcal{E}_{\mathsf{eras}}\wedge s_\sB = 0)}
    \end{equation}
    \noindent which gives the desired Markov property. The proof of the desired claim then follows from the sequence
    \begin{equation}
        (s_\sA, s_\sB=0, s_\sC)\in  \mathcal{E}_{\mathsf{eras}}\Rightarrow   (s_\sA, s_\sB=0, s_\sC=0)\in  \mathcal{E}_{\mathsf{eras}}\Rightarrow (s_\sA, s_\sB=0, s_\sC')\in  \mathcal{E}_{\mathsf{eras}}
    \end{equation}
    \noindent where we can erase $s_\sC$ and replace it by $s'_\sC$ without modifying $s_\sA, s_\sB=0$ since $\sB$ screens $\sA$ from $\sC$.
\end{proof}

\begin{remark}
    The condition $s_\sB=0$ can be generalized to some $s_\sB\in \{0, 1\}^{|\sB|}$ without much trouble, but generalizing to all valid $s_\sB$ requires an additional assumption:
    \begin{equation}
        \forall(s_\sA, s_\sB, s_\sC), (s_\sA', s_\sB, s_\sC')\in \mathcal{E}_{\mathsf{eras}}, \quad  (0, 0, s_\sC\oplus s_\sC')\in \mathsf{Im}(H),
    \end{equation}
    \noindent which is true in specific models, but is technically an additional assumption on top of \cref{definition:connectedness}.
\end{remark}

\section{Spatial Mixing within a Logical Sector}
\label{section:ssm}

Much like the 2D Ising model, the Gibbs state of the code Hamiltonian of a self-correcting quantum memory isn't expected to satisfy any uniform notion of decay of correlations at low temperatures. Nevertheless, next we prove decay-of-correlation results for the Gibbs state within a logical sector.

We organize this section as follows. In \cref{section:definitions_sm}, we define the low-temperature versions of decay of correlations considered in this work; namely weak and 
 (domain) strong spatial mixing \textit{within a logical sector} and present the main result of this section. In \cref{section:sep_surf} we introduce and prove basic properties of separating surfaces. Then, in \cref{section:ssmproof} we prove that self-correcting memories which satisfy the connectedness criterion (\cref{definition:connectedness}) satisfy \textsf{SSM} within a logical sector.

\subsection{Definitions of Decay-of-Correlations within a Logical Sector}
\label{section:definitions_sm}

The following definition of weak-spatial-mixing \textit{within a logical sector}, in some sense quantifies point-to-set correlations in the syndrome Gibbs distribution when the boundaries are pinned to their ground state configuration. 

\begin{definition}
    [WSM within a Logical Sector]\label{definition:wsm} Fix a parity check matrix $H\in \mathbb{F}_2^{m\times n}$, and an inverse temperature $\beta$. Then the associated syndrome Gibbs distribution $\pi_\beta$ is said to satisfy weak-spatial-mixing (\textsf{WSM}) within a logical sector if its marginals are approximately invariant under pinning to $0$ the syndromes at the boundary of a ball of radius $R$. That is, 
     \begin{equation}
       \forall u\in [m]:\quad \big\|\pi_\beta\big(s_{u} \big| s_{\partial\mathcal{B}_u^R} = 0\big) -  \pi_\beta\big(s_{u} \big)\big\|_1\leq \delta_{\beta, R}^{\mathsf{WSM}},
    \end{equation}
with error $\delta_{\beta, R}^{\mathsf{WSM}}$ that is exponentially-decaying in $R$. The ball $\mathcal{B}^{R}_u$ is defined with respect to the syndrome network of \cref{definition:syndrome_network}.
\end{definition}

At low temperatures, the syndrome configurations are sparse and localized. \cref{definition:wsm} informally says that the syndrome $u$ at the center has ``forgotten" about the syndrome configurations far away, and that they might as well be near the ground state. 

\begin{remark}  
    Applied to the Ising model, \cref{definition:wsm} is equivalent to that of \cite[Eq.~(1.3)]{GS23Ising}, who instead pin the spins lying at the boundary. 
\end{remark}

Next, we introduce the notion of strong-spatial-mixing (\textsf{SSM}). In general, \textsf{SSM} quantifies the sensitivity of a region of spins $\mathsf{A}$ (within a larger region $\mathsf{B}$) to perturbations of the boundary of $\mathsf{B}$. If \textsf{WSM} within a logical sector is a relaxation of \textsf{WSM} to specific, ground-state-like boundary conditions, then analogously, one can introduce a weaker version of \textsf{SSM} as a relaxation of \textsf{SSM} to specific, ground-state-like boundaries. We refer to the version of \textsf{SSM} we study as ``domain" \textsf{SSM} as it quantifies deformations of boundaries/domains:

\begin{definition}
    [Domain SSM]\label{definition:ssm} Fix a parity check matrix $H\in \mathbb{F}_2^{m\times n}$, and an inverse temperature $\beta$. The associated syndrome Gibbs distribution $\pi_\beta$ is said to satisfy \textit{domain strong-spatial-mixing} (\textsf{SSM}) \textit{within a logical sector,} if its marginals are approximately invariant under deformations of boundaries pinned to $0$. That is, if $\mathsf{B}\subset [m]$ is an arbitrary subset, then
    \begin{equation}
     \forall R\geq 1,   \forall \mathsf{A}, \mathsf{C}\subseteq \mathsf{B} \text{ s.t. } d(\mathsf{A}, \mathsf{C})\geq R:\quad \big\|\pi_\beta\big(s_{\mathsf{A}} \big| s_{\partial\mathsf{B}} = 0\big) -  \pi_\beta\big(s_{\mathsf{A}}|s_{\partial(\mathsf{B}\setminus \mathsf{C})}=0\big)\big\|_1\leq |\mathsf{C}|\cdot \delta_{\beta,R}^{\mathsf{SSM}},
    \end{equation}
with error $\delta_{\beta,R}^{\mathsf{SSM}}$ that is exponentially-decaying in $R$. Again, the distances and boundaries are defined with respect to the syndrome network of \cref{definition:syndrome_network}.
\end{definition}

\begin{remark}
    Our definition of \textsf{SSM} under ``boundary perturbations" differs from the standard setup, where specific bits of the boundary are flipped (see e.g. \cite[Definition 2.2]{Dyer2002MixingIT}). Since we are conditioning on syndrome configurations, flipping individual syndrome bits on the boundary may not be well defined. Instead, we opt for deformations of the boundary.
\end{remark}

The main result of this section is a proof that a general family of self-correcting (classical or quantum) memories, whose parity checks satisfy the ``connectedness criterion" in \cref{definition:connectedness}, admit decay-of-correlations within a logical sector.

\begin{lemma}
    [Domain \textsf{SSM} within Self-Correcting Memories]
    \label{lemma:ssm} Let $H\in \mathbb{F}_2^{m\times n}$ be a parity check matrix satisfying $(\ell, d, \nu, \chi)-\mathsf{CC}$ as in \cref{definition:connectedness} for some choice of constants; in particular, its syndrome network has bounded degree $d$.  Then, for all $\beta>\beta_0(d)\equiv 1+2\ln (4d)$ the associated syndrome Gibbs distribution satisfies domain $\mathsf{SSM}$ within a logical sector with error
    \begin{equation}
    \delta_{\beta, R}^{\mathsf{SSM}} \leq 20\cdot d\cdot  e^{-\min(m^\chi, R)\cdot (\beta-\beta_0(d))/2}.
    \end{equation}
\end{lemma}

\begin{remark}
    The constraint $R\leq m^\chi$ in some sense captures the threshold of a syndrome configuration around $\mathsf{C}$ that can impart an error above the energy barrier of the code. 
\end{remark}

\subsection{Separating Surfaces}
\label{section:sep_surf}

The main ingredient in the proof approach is the following notion of a separating surface (prevalent e.g., in the literature on percolation in $\mathbb Z^d$~\cite{grimmett1989percolation}) in the syndrome network. At a high level, for a center syndrome $u$, a separating surface imposes a boundary around $u$ of unviolated syndromes. 

\begin{definition}
    [$(\mathsf{B, C}, r)$-Separating-Surface]\label{definition:separating_surfaces} Fix $H\in \mathbb{F}_2^{m\times n}$, the associated syndrome network $G = ([m], E)$, and an integer $r>0$. Fix subsets $\mathsf{B}\subseteq [m]$ and $\mathsf{C}\subset \mathsf{B}$. Then, a collection of syndromes $\mathsf{S}\subset \mathcal{B}_{\mathsf{C}}^r\cap\mathsf{B}$ is said to be an $r-\mathsf{Surf}$ for $\mathsf{C}$ in $\mathsf{B}$ under $s$ if
    \begin{enumerate}
        \item There are no violations in $\mathsf{S}$. $\forall u\in \mathsf{S}:\quad s_u=0$.
        \item $\mathsf{S}$ disconnects $\mathsf{C}$ within $\mathsf{B}$. That is, let $\mathsf{Conn}_{\mathsf{B}}(\mathsf{C}|\mathsf{S})\subseteq \mathsf{B}$ denote the syndromes connected to $\mathsf{C}$ in $G$ without traversing $\mathsf{S}$ nor $[m]\setminus \mathsf{B}$. Then, 
        \begin{equation}
            \mathsf{Conn}_{\mathsf{B}}(\mathsf{C}|\mathsf{S}) \cap\partial \mathcal{B}^r_{\mathsf{C}} =\emptyset\,.
        \end{equation}
    \end{enumerate}
\end{definition}

\noindent Observe that $\mathsf{S}\cup\partial\mathsf{B}$ screens $\mathsf{C}$ from $(\partial \mathcal{B}^r_{\mathsf{C}}\setminus \mathsf{S})$ in the sense of \cref{definition:screens}.

Recall now that to prove $\mathsf{SSM}$, our goal is to compare the marginal distribution on a third subset $\mathsf{A}\subset \mathsf{B}$, of two conditional Gibbs distributions, under two different boundaries $\partial \mathsf{B}$ and $\partial (\mathsf{B}\setminus \mathsf{C})$. Let us first draw two independent samples, from each of these conditional distributions:
\begin{equation}
    X\leftarrow \pi(\cdot |s_{\partial\mathsf{B}} = 0), \quad    Y\leftarrow \pi(\cdot |s_{\partial(\mathsf{B}\setminus \mathsf{C})} = 0)\,.
\end{equation}

We refer the reader to \cref{section:techniques} for intuition on why the existence of a \textit{shared} $(\mathsf{B, C}, r)$ separating surface is useful in developing a coupling argument for the marginals on the region $\mathsf{A}$. We first establish a sufficient condition for the existence of such a shared surface in $X, Y$: the absence of a certain path of syndrome violations from the discrepancy $\mathsf{C}$, to the boundary a distance $r$ away.  

\begin{claim}
    [The Contrapositive to a Shared Separating Surface]\label{claim:contrapositive_boundary} Fix an integer $r>1$. If the samples $X, Y$ do not share an $r$-$\mathsf{Surf}$, then there exists a path $\mathsf{p}\subset \mathsf{B}$ in $G$ from $\mathcal{B}^{1}_\mathsf{C}$ to $\partial\mathcal{B}^{r-1}_\mathsf{C}$, where each vertex on the path corresponds to a syndrome violated in either $X$ or $Y$.
\end{claim}

\begin{proof}
    Consider the subset of vertices $\mathsf{V}\subset \mathsf{B}$ which are violated in $X$ or $Y$, and are contained in $\mathsf{B}$. Consider the subgraph $G|_{\mathsf{V}\cup\mathsf{C}}$ and denote as $\mathsf{U}$ the connected component of $\mathsf{V}\cup \mathsf{C}$ traversing only vertices in $\mathsf{V}$. Next we consider the surface $\mathsf{S} = (N_G(\mathsf{U})\setminus \mathsf{U})\cap \mathsf{B}$ given by the neighborhood of the connected component of violations $\mathsf{U}$ in $G|_{\mathsf{B}}$. Note that $\mathsf{S}$ satisfies:
    \begin{enumerate}
        \item $\mathsf{S}\subseteq \mathsf{B}$.
        \item $\mathsf{S}\cap \mathsf{V}=\emptyset$. That is, $\mathsf{S}$ consists of unviolated syndromes in both $X, Y$. Otherwise, said violation would extend the connected component $\mathsf{U}$.
        \item $\mathsf{U} =\mathsf{Conn}_\mathsf{B}(\mathsf{C}|\mathsf{S})$, since the neighbors in $G|_{\mathsf{B}}$ of $\mathsf{C}$ which are not in $\mathsf{V}$, are placed in $\mathsf{S}$.
    \end{enumerate}

   Now, suppose $\mathsf{S}$ is not an $r-\mathsf{Surf}$. Then either $\mathsf{S}\nsubseteq \mathcal{B}_{\mathsf{C}}^r\cap\mathsf{B}$, or $\mathsf{U}\cap\partial\mathcal{B}^r_\mathsf{C}\neq \emptyset$. Since $\mathsf{U}$ is connected, both possibilities entail a path of violations in $\mathsf{V}$ from $\mathcal{B}^1_\mathsf{C}$ to $\partial\mathcal{B}^{r-1}_\mathsf{C}$ in $\mathsf{V}$, as claimed. 
\end{proof}

\begin{claim}
    [Coupled Separating Surfaces] \label{claim:separating_surfaces_toric} Fix $H\in \mathbb{F}_2^{m\times n}$ satisfying \cref{definition:connectedness}, and subsets $\mathsf{C\subset\mathsf{B}}\subseteq[m]$. Independently sample
    \begin{equation}
    X\leftarrow \pi(\cdot|s_{\partial\mathsf{B}} = 0)\,, \quad    Y\leftarrow \pi(\cdot|s_{\partial(\mathsf{B}\setminus \mathsf{C})} = 0)\,.
\end{equation}
    Then, if $\beta>\beta_0(d) = 1+2\ln (4d)$, $X$ and $Y$ agree on an $(\mathsf{B, C}, r+1)$-$\mathsf{Surf}$ with probability all but
    \begin{equation}
     \mathbb{P}_{X, Y}[\nexists \text{ shared } (\mathsf{B, C}, r+1)-\mathsf{Surf}]   \leq 16\cdot e^{-\min(r, m^\chi)(\beta-\beta_0(d))/2}\,.
    \end{equation}
\end{claim}

\begin{proof}[Proof of \cref{claim:separating_surfaces_toric}] We argue the existence of the separating surface as in \cref{definition:separating_surfaces}, via the contrapositive as in \cref{claim:contrapositive_boundary}. If the path $\mathsf{p}$ is of length $k$, then either at least $k/2$ of the vertices in $\mathsf{p}$ are violated in $X$, or at least $k/2$ of the vertices in $\mathsf{p}$ are violated in $Y$. We assume the first case from now on; the second case is analogous. 

Let $\mathsf{sp}\subseteq \mathsf{p}$ be the violated syndromes in $X$ that are part of $\mathsf{p}$, where $|\mathsf{sp}|\geq |\mathsf{p}|/2$. Our goal now is to understand the probability of such a subpath $\mathsf{sp}\subseteq X$ via a Peierls argument. Here we require a careful understanding of the size of $\mathsf{sp}$, and in particular, whether or not it is \textsf{erasable}. Let $\mathsf{V}$ denote the violated syndromes in $X$ that are connected to $\mathsf{sp}$. If the various connected components of $\mathsf{V}$ are all of bounded size (and therefore are topologically trivial), then $\mathsf{V}$ is \textsf{erasable}; in particular, we can run a Peierls argument to compute the probability of such a sub-path:
    \begin{align}\label{equation:percolation_bound}
        \mathbb{P}_{X\leftarrow \pi_\beta}[\mathsf{sp}\subseteq X \wedge X\in \mathcal{E}_{\mathsf{eras}}]&\leq \mathbb{P}_{X\leftarrow \pi_\beta}[\exists \mathsf{V}: \mathsf{sp}\text{ connected to } \mathsf{V},\mathsf{V}\subseteq X, \mathsf{V}\textsf{ erasable}]\\
        &\leq \sum_{\substack{\mathsf{V} \text{ connected to }\mathsf{sp}\\\mathsf{V}\textsf{ erasable}}} e^{-\beta |\mathsf{V}|}\\& \leq  \sum_{\ell\geq |\mathsf{sp}|} \bigg(e^{\ell} \cdot d^{\ell-|\mathsf{sp}|}\bigg)\cdot e^{-\beta \ell} \\
        &= d^{-|\mathsf{sp}|}\cdot \sum_{\ell\geq |\mathsf{sp}|} (de^{1-\beta})^\ell = \frac{(e^{1-\beta})^{|\mathsf{sp}|}}{1-de^{1-\beta}}\,.
    \end{align}
    Here, the second line uses the union bound, the Peierls argument in \cref{lemma:peierls_wsm}, and the Markov property in \cref{lemma:markov} (to ensure the Peierls bound holds under conditioning on $0$ boundaries). The third line sums over all subsets that are connected to $\mathsf{sp}$ (relaxing the constraint of being \textsf{erasable}) and uses \cref{fact:clustering}; the fourth line assumes that $de^{1-\beta}<1$. We next run a union bound over all paths $\mathsf{p}$ starting from $\mathcal{B}_{\mathsf{C}}^1$ to the boundary $\partial\mathcal{B}_\mathsf{C}^r$, and then over all sub-paths $\mathsf{sp}$ of $\mathsf{p}$: 
    \begin{align}
        \mathbb{P}_{X\leftarrow \pi_\beta}[\exists \mathsf{p} :\partial\mathcal{B}_\mathsf{C}^1\rightarrow \partial\mathcal{B}_\mathsf{C}^r \wedge \mathcal{E}_{\mathsf{eras}}] &\leq \sum_{|\mathsf{p}|\geq r} \sum_{\substack{\mathsf{sp}\subseteq \mathsf{p} \\ |\mathsf{sp}|\geq \mathsf{p}/2  }} \mathbb{P}_{X\leftarrow\pi_\beta}\big[\mathsf{sp}\subseteq X\wedge \mathcal{E}_{\mathsf{eras}}\big] \\ &\leq \sum_{|\mathsf{p}|\geq r} d^{|\mathsf{p}|} \sum_{|\mathsf{sp}|\geq |\mathsf{p}|/2 } \binom{|\mathsf{p}|}{|\mathsf{sp}|} \times \mathbb{P}_{X\leftarrow\pi_\beta}\big[\mathsf{sp}\subseteq X\wedge \mathcal{E}_{\mathsf{eras}}\big] \\ &\leq |\mathcal{B}^1_\mathsf{C}|\cdot \sum_{\ell\geq r} d^\ell\times  2^\ell \times \frac{(e^{1-\beta})^{\ell/2}}{1-de^{1-\beta}} \\&\leq |\mathcal{B}^1_\mathsf{C}|\cdot (2de^{(1-\beta)/2})^{r}\times \frac{1}{1-de^{1-\beta}} \times \frac{1}{1-2d(e^{1-\beta})^{1/2}}\,,
    \end{align}
so long as $\beta> \beta_0(d)$ such that $2d(e^{1-\beta_0(d)})^{1/2}=1/2$. We can now use \cref{corollary:all-erasable} to ensure the absence of such a path in $X$ with high probability: 
\begin{align}
    \mathbb{P}_{X\leftarrow \pi_\beta}[\exists \mathsf{p} :\mathcal{B}_\mathsf{C}^1\rightarrow \partial\mathcal{B}_\mathsf{C}^r] & \leq 4 \cdot |\mathcal{B}^1_\mathsf{C}|\cdot e^{-r\cdot (\beta-\beta_0(d))/2} + \frac{m^\nu (de^{1-\beta})^{m^\chi}}{1-de^{1-\beta}} \\
    & \leq 8\cdot |\mathcal{B}^1_\mathsf{C}|\cdot e^{-\min(r, m^\chi )\cdot (\beta-\beta_0(d))/2}\,.
\end{align}
    To conclude, we union bound over the sample from $Y$ as well. We arrive at:
    \begin{equation*}
        \mathbb{P}[\exists \text{ shared separating surface}] \geq 1- 16\cdot |\mathcal{B}^1_\mathsf{C}|\cdot e^{-\min(r, m^\chi )(\beta-\beta_0(d))/2}\,. \qedhere
    \end{equation*}    
\end{proof}

\subsection{Proof of SSM within a Logical Sector}
\label{section:ssmproof}

We are now in a position to prove the main result of this section, on the Domain SSM property for the syndrome Gibbs distribution of parity check matrices satisfying \cref{definition:connectedness}. 

\begin{proof}[Proof of \cref{lemma:ssm}]
    The crux of the proof lies in the Markov property of the Gibbs distribution (\cref{lemma:markov}). Let us fix an inverse temperature $\beta$ and regions $\mathsf{A, B, C}$ such that $d(\mathsf{A, C})>r$. From \cref{claim:separating_surfaces_toric}, we have that samples from the distributions $\pi(\cdot|s_{\partial\mathsf{B}} = 0)$ and $\pi(\cdot|s_{\partial(\mathsf{B}\setminus \mathsf{C})} = 0)$, both share a separating surface (around the site $u$) with high probability. From \cref{corollary:all-erasable}, samples from the distributions conditioned on all syndrome components being erasable
    \begin{equation}\label{eq:independent_samples}
         \pi_\beta\big(\cdot \big|s_{\partial\mathsf{B}} = 0\wedge\mathcal{E}_\mathsf{eras})\,, \quad \text{ and } \quad \pi_\beta\big(\cdot \big| s_{\partial(\mathsf{B}\setminus \mathsf{C})} = 0\wedge \mathcal{E}_\mathsf{eras}\big)\,,
    \end{equation}

    \noindent also share such surfaces except with probability $18\cdot |\mathcal{B}^1_\mathsf{C}|\cdot  e^{-\min(r, m^\chi)(\beta-\beta_0(d))/2}$.

    We can thereby sample $X, Y$ independently (from the distributions in \cref{eq:independent_samples}), from the boundary discrepancy ``inwards", until the existence of the first such shared $(\mathsf{B, C}, r)$ surface. Given the arbitrariness of the geometry, we spell out this ``revealing algorithm" in \cref{alg:revealing}. We claim that if \cref{alg:revealing} does not reject, then it outputs a $(\mathsf{B, C}, r)$ separating surface. Indeed, consider the connected component of $\mathsf{A}$ among un-revealed vertices; by the stopping condition (Line 3) its neighborhood must contain only unviolated syndromes in $X$ or $Y$. Otherwise, there would be a revealed vertex in $G|_{\mathsf{B}}$ associated to a syndrome violation, adjacent to an unrevealed vertex in $\mathsf{B}$.

    Conditional on the existence of said $(\mathsf{B, C}, r)$ surface $\mathsf{S}$, then $\mathsf{S}$ in addition to the shared boundary $\partial\mathsf{B}\cap \partial(\mathsf{B}\setminus\mathsf{C})$ screens $\mathsf{C}$ from $\mathsf{A}$ with a boundary of $0$s. When additionally conditioned on $\mathcal{E}_{\mathsf{eras}}$, we can then apply the Markov property in \cref{lemma:markov} to ensure the marginal distributions within the surface match exactly. We can therefore couple the assignment interior to the surface to ensure they agree (with probability 1). That is, 
    \begin{equation}
        \|\pi_\beta\big(s_\mathsf{A} \big|s_{\partial\mathsf{B}} = 0\wedge\mathcal{E}_\mathsf{eras}) - \pi_\beta\big(s_\mathsf{A} \big| s_{\partial(\mathsf{B}\setminus \mathsf{C})} = 0\wedge \mathcal{E}_\mathsf{eras}\big)\|_1\leq 18\cdot |\mathcal{B}^1_\mathsf{C}|\cdot e^{-\min(R, m^\chi)(\beta-\beta_0(d))/2}\,.  
    \end{equation}
    Applying \cref{corollary:all-erasable} again, with the bound on the degree of the syndrome network $|\mathcal{B}^1_\mathsf{C}| \leq (d+1)\cdot |\mathsf{C}|$, then gives the desired bound:
    \begin{equation*}
        \delta_{\beta, r}^{\mathsf{SSM}} \leq 20 \cdot d\cdot e^{-\min(r, m^\chi)(\beta-\beta_0(d))/2}\,.  \qedhere
    \end{equation*}
\end{proof}

\begin{algorithm}
    \caption{The Revealing Algorithm $(\mathsf{A, B, C})$}\label{alg:revealing}
    \begin{algorithmic}[1]
        \State \textbf{Maintain} a set of already sampled (``revealed") syndromes $\mathsf{R}$; denote as $\partial_{\mathsf{int}}\mathsf{R}\subseteq \mathsf{R}\cap \mathsf{B}$ the interior boundary of $\mathsf{R}$ in the subgraph $G|_{\mathsf{B}}$.
        \State \textbf{Initialize} the set $\mathsf{R} =  \partial \mathsf{B}\cup \mathsf{B}^1_{\mathsf{C}}$,  and the samples:
        \begin{equation}
          X_{ \mathsf{R}}\leftarrow \pi_\beta\big(s_{\mathsf{R}} \big|s_{\partial\mathsf{B}} = 0\wedge\mathcal{E}_\mathsf{eras}), \quad Y_{\mathsf{R}} \leftarrow \pi_\beta\big(s_{\mathsf{R}} \big|s_{\partial(\mathsf{B\setminus \mathsf{C}})} = 0\wedge\mathcal{E}_\mathsf{eras})
        \end{equation}
        \While{there exists a violation in $\partial_{\mathsf{int}}\mathsf{R}$, in either $X_\mathsf{R}$ or $Y_\mathsf{R}$, and $\mathsf{A}\cap \mathsf{R}=\emptyset$.}
        \State Identify the violation $v\in \partial_{\mathsf{int}}\mathsf{R}$ closest to $\mathsf{C}$; ties broken arbitrarily. 
        \State \textit{Reveal} the neighborhood $N_v\cap \mathsf{B}$ of $v$ in $G$, by sampling the unrevealed neighbors in $\mathsf{B}$:
        \begin{gather}
            X_{(N_v\cap \mathsf{B})\setminus \mathsf{R}}\leftarrow \pi_\beta\big(s_{(N_v\cap \mathsf{B})\setminus \mathsf{R}} \big|X_\mathsf{R}\wedge \mathcal{E}_\mathsf{eras}),\\ Y_{(N_v\cap \mathsf{B})\setminus \mathsf{R}}\leftarrow \pi_\beta\big(s_{(N_v\cap \mathsf{B})\setminus \mathsf{R}} \big|Y_\mathsf{R}\wedge \mathcal{E}_\mathsf{eras})
        \end{gather}
        \State Update the subset 
        \begin{equation}
         \mathsf{R}\leftarrow \mathsf{R}\cup (N_v\cap \mathsf{B}),
        \end{equation}
        \EndWhile
        \State \textbf{Return} $\partial_{\mathsf{int}}\mathsf{R}$ if $\mathsf{A}\cap \mathsf{R}=\emptyset$; else output $\bot$.
    \end{algorithmic}

\end{algorithm}

\section{Rapid Mixing within a Logical Sector}
\label{section:mixing}

In this section, we present our algorithm to prepare the Gibbs states within a sector of self-correcting (quantum or classical) memories satisfying the criterion of \cref{definition:connectedness}, and its analysis. We begin in \cref{subsection:algorithm} by presenting the description of the modified block dynamics we consider. In \cref{subsection:results_cbd}, we present the main results of this section, which establish rapid mixing within a sector of the dynamics when initialized from a ground state. 

The remain subsections pertain to our analysis: In \cref{subsection:small_loops} we define the restriction of the state space to the small loop configurations, $\Omega_R$. In \cref{subsection:chain_restriction} we consider the Markov chain restricted to $\Omega_R$, and prove the restricted chain and the original chain are indistinguishable. In \cref{subsection:fast_mixing}, we prove the Markov chain restricted to $\Omega_R$ is rapid mixing; contingent only on a ``cluster-distance contraction" claim which we defer to  \cref{subsection:cycle_contraction}.

\subsection{The Algorithm}
\label{subsection:algorithm}

We begin by defining a Markov chain on the space of syndromes; and then reason on how to implement said chain using physical updates on the code.

\subsubsection{The Syndrome Markov Chain}

We define the following Markov chain $\mathsf{P}\equiv \mathsf{P}_{\mathsf{syn}}$ over the valid syndrome space, based on a modification to the standard block dynamics. Let us fix some box-length $L$, to be determined later.  Roughly speaking, in the modification we introduce, instead of resampling the entirety of the box $\mathcal{B}^L$ conditioned on the current boundary $\partial\mathcal{B}^L$ configuration, we additionally pin all the connected components of syndrome configurations (say, the cycles) which are incident on the boundary $\partial\mathcal{B}^L$.

\begin{tcolorbox}[breakable]
\begin{definition}
    [The Conditional Block Dynamics]\label{definition:cbd} Let $H\in \mathbb{F}_2^{m\times n}$ be a parity check matrix satisfying $(\ell, d, \nu, \chi)-\mathsf{CC}$ as in \cref{definition:connectedness}, and fix an integer $L$ s.t.  $|\mathcal{B}^L|\leq m^\chi$.
    
    We consider the following Markov Chain over $\Omega=\mathsf{Im}(H)\subset \{0, 1\}^m$, with transition matrix $\mathsf{P}\equiv \mathsf{P}_{\mathsf{syn}}:\mathbb{R}^\Omega\rightarrow \mathbb{R}^\Omega$, where a transition is sampled as follows. 
    \begin{enumerate}
        \item A syndrome $u\in [m]$ is sampled uniformly at random. Let $\mathcal{B}^L\subset [m]$ denote the box of radius $L$ around $u$.
        \item Measure the syndrome $s$ within and on the boundary of $\mathcal{B}^L$. Let $\mathcal{V} = \{\mathsf{V}_i:\: \mathsf{V}_i\subset [m]\}_i$ denote all the connected components of violated syndromes which are incident on $\partial\mathcal{B}^L$. We define the ``interior" of the box $\mathcal{B}^L$:
        \begin{equation}
            \mathsf{B}\equiv \mathcal{B}^L\setminus \bigcup_i \big( N(\mathsf{V}_i) \cup \mathsf{V}_i\big)\,.
        \end{equation}
        \noindent Note that the syndromes in $N(\mathsf{V}_i)\setminus \mathsf{V}_i$ must be 0, ensuring that if $\mathsf{B}\neq \emptyset,$ $s_{\partial \mathsf{B}} = 0$.
        
        \item Generate the next syndrome $s'\in \mathsf{Im}(H)$, by resampling the syndrome distribution within the interior $\mathsf{B}$, conditioned on the boundary.
        \begin{equation}
            s'_{\mathsf{B}}\leftarrow \pi_\beta\big(s_{\mathsf{B}} \big| s_{\partial\mathsf{B}} = 0\big)\,, \quad s'_{[m]\setminus \mathsf{B}} = s_{[m]\setminus \mathsf{B}}\,.
        \end{equation}
\end{enumerate}    
\end{definition}
\end{tcolorbox}

Here two remarks are in order. 

\begin{remark}
    The Markov chain $\mathsf{P}\equiv \mathsf{P}_{\mathsf{syn}}$ is on the space of syndromes, and doesn't detail how to implement the updates on the physical code. We return to this point in \cref{subsubsection:physical_updates}.
\end{remark}

Next, note that by pinning the components $\mathcal{V}$ which are incident on the boundary, we avoid configurations in which those connected components extend or merge together (or change at all!). Consequently, if we initialize the system from a ground state (of syndrome $0$), then the syndrome configurations encountered during the evolution are never larger than the size of a box $\mathcal{B}^L$. By virtue of \cref{definition:syndrome_clustering} this will entail the encountered syndromes are not $\mathsf{critical}$.

\begin{remark}
 Alternatively, if the initial configuration contains a component which is too large to fit in any box $\mathcal{B}^L$, then it is never modified by the dynamics. 
\end{remark}

\subsubsection{Implementing the Physical Updates}
\label{subsubsection:physical_updates}

To implement a given update $s\rightarrow s'$ of the syndrome chain $\mathsf{P}\equiv \mathsf{P}_{\mathsf{syn}}$ on the actual code, it suffices to (carefully) pick errors $e\in \mathbb{F}_2^n$ consistent with the syndrome change. Moreover, we would like these (quasi) local updates on the syndromes, to be implemented using (quasi) local updates on the symbols of the code:
\begin{tcolorbox}
    \begin{definition}
        [Physical Updates]\label{definition:phys_updates} Fix $s', s\in \mathsf{Im}(H)$ which differ within a region $\mathsf{B}\subset [m]$ as in \cref{definition:cbd}. To generate the update on the code, we apply the canonical corrections $$\mathsf{corr}(s_\mathsf{B})\oplus \mathsf{corr}(s'_\mathsf{B})\,.$$
    \end{definition}
\end{tcolorbox}

The fact that $\mathsf{B}$ has pinned boundaries $\partial\mathsf{B} = 0$ and is of size $|\mathcal{B}^L|\leq  m^\chi$ ensures $s_\mathsf{B}$ is erasable, and thereby \cref{definition:phys_updates} is well-defined. The change in syndrome associated to said update is

\begin{equation}
    H\big(\mathsf{corr}(s_\mathsf{B})\oplus \mathsf{corr}(s'_\mathsf{B})\big) = s_\mathsf{B}\oplus s'_\mathsf{B} = s\oplus s'\,,
\end{equation}

\noindent thereby giving the correct syndrome update. \cref{definition:phys_updates} precisely erases the syndrome within $\mathsf{B},$ and subsequently replaces them with $s'$. Further, the \textsf{Locally Erasable} property (\cref{definition:locally_erasable}) ensures the update $\mathsf{corr}(s_\mathsf{B})\oplus \mathsf{corr}(s'_\mathsf{B})$ is quasi-local, i.e. contained in a box of spins $\mathsf{supp}(\mathcal{B}^L)$.\footnote{While we do not yet discuss the efficiency of these quasi-local updates, we remark that an error which is equivalent to the min-weight correction can be found by Gaussian elimination on the $\mathsf{supp}(\mathcal{B}^L) \times \mathcal{B}^L$ submatrix of $H$ (again, by the \textsf{Locally Erasable} property \cref{definition:locally_erasable}). See also \cref{section:toric_is_even_subgraph}.}

\begin{tcolorbox}
    \begin{definition}
    \label{definition:code_chain}
        Fix $H\in \mathbb{F}_2^{m\times n}$ and the Markov chain $\mathsf{P}_\mathsf{syn}$ over $\mathsf{Im}(H)$ from \cref{definition:cbd}. We denote as $\mathsf{P}_{\mathsf{code}}$ the discrete time Markov chain over the state space $\{0, 1\}^n$ given by implementing the updates of $\mathsf{P}_{\mathsf{syn}}$ using the update rules in \cref{definition:phys_updates}.
    \end{definition}
\end{tcolorbox}

We now prove two simple lemmas on the relationship between the (discrete-time) Markov chains on the code-space and that on the syndrome-space. 

\begin{lemma}
    [(Classical) Code to Syndrome Chain Reduction]\label{lemma:cl_reduction} Fix $H\in \mathbb{F}^{m\times n}$ satisfying $(\ell, d, \nu, \chi)-\mathsf{CC}$ as in \cref{definition:connectedness} for any choice of constants, and let $\mathsf{P}_{\mathsf{syn}}, \mathsf{P}_{\mathsf{code}}$ be the Markov chains defined in \cref{definition:cbd}, \cref{definition:code_chain} respectively.  Then, so long as $|\mathcal{B}^{L}|\leq m^{\chi}$,
    \begin{equation}\label{equation:cl_reduction}
    \forall c \text{ s.t. }Hc=0^m, \quad  \forall t\in \mathbb{Z}^+:\quad   \| \mathsf{P}_{\mathsf{code}}^t(c) - \rho^c_\beta\|_1 =  \| \mathsf{P}_{\mathsf{syn}}^t(0^m) - \pi_\beta\|_1\,.
    \end{equation}
\end{lemma}

The continuous time case is analogous. 

\begin{proof}
    To establish \cref{equation:cl_reduction}, it suffices to prove that after time $t$, the cumulative error $e$ applied to the original code-word $c$ is $e=\mathsf{corr}(s)$, where $s$ is the current syndrome (recall \cref{def:gibbs_within_logical}). In implementing the block update resulting in the next syndrome $s'$, we apply the bit-flip $e=\mathsf{corr}(s_\mathsf{B})\oplus \mathsf{corr}(s'_\mathsf{B})$ of syndrome $s\oplus s'$, of support within $\mathcal{B}^L$. Since $|\mathcal{B}^L|\leq m^\chi$, all the connected components of violations in $s, s'$ are \textsf{erasable}, so the update is well defined. We conclude
    \begin{equation*}
        \mathsf{corr}(s)\oplus e = \mathsf{corr}(s)\oplus \mathsf{corr}(s_\mathsf{B})\oplus \mathsf{corr}(s'_\mathsf{B}) = \mathsf{corr}(s')\,. \qedhere
    \end{equation*}
\end{proof}

To establish the analog for quantum error-correcting codes, we ``time-share" the Markov chains associated to the X and Z syndromes. 

\begin{lemma}
    [(Quantum) Code to Syndrome Chain Reduction]\label{lemma:q_reduction} Fix a CSS code $\mathsf{CSS}(H_X, H_Z)$ satisfying $(\ell, d, \nu, \chi)-\mathsf{CC}$ as in \cref{definition:connectedness} for any choice of constants. Let $\mathsf{P}_{\mathsf{syn}}^X, \mathsf{P}_{\mathsf{code}}^X, \mathsf{P}_{\mathsf{syn}}^Z, \mathsf{P}_{\mathsf{code}}^Z$ be the Markov chains associated to $H_X, H_Z$ respectively, as defined in \cref{definition:cbd}, \cref{definition:code_chain}. Let $\mathcal{L}_X, \mathcal{L}_Z$ be the generators associated to $\mathsf{P}_{\mathsf{code}}^X, \mathsf{P}_{\mathsf{code}}^Z$ as defined in \cref{fact:cont_disc_quantum}.  
    
    Let $\psi$ be a codeword of $\mathsf{CSS}(H_X, H_Z)$. Then, so long as $|\mathcal{B}^{L}|\leq m^\chi$, the continuous time evolution under $\mathcal{L} = \mathcal{L}_X + \mathcal{L}_Z$ satisfies

    \begin{equation}
        \forall t\geq 0:\quad   \| e^{t\mathcal{L}}(\psi) - \rho^\psi_\beta\|_1 \leq  2\cdot \max_{b\in \{X, Z\}}\mathbb{E}_{k \sim \mathsf{Pois}(mt)} \| (\mathsf{P}^b_{\mathsf{syn}})^k(0^{m_b}) - \pi_\beta^b\|_1\,.
    \end{equation}
    
\end{lemma}

\begin{proof}
    When acting on an eigenstate of the code Hamiltonian, $\mathcal{L}_X, \mathcal{L}_Z$ commute. The continuous time formulation then translates to an application of $\mathsf{P}_{\mathsf{code}}^X, \mathsf{P}_{\mathsf{code}}^Z$ for $k$ steps, where $k$ is drawn from a sum of $m$ independent Poisson random variables of mean $t$.

    Correctness, i.e. the constraint that the logical information is unmodified, is analogous to that of \cref{lemma:cl_reduction}, where $\mathsf{corr}(s)\oplus e$ is equivalent to $\mathsf{corr}(s')$ up to a code stabilizer.
\end{proof}

\subsection{Results of this Section}
\label{subsection:results_cbd}

The main theorem of this section states that the conditional block dynamics, as a discrete-time Markov chain on the syndrome configuration, mixes quickly to the target syndrome Gibbs distribution, up to a small ``leakage" rate.

\begin{theorem}
    [The Conditional Block Dynamics is Rapid Mixing]\label{theorem:cbd_rapid_mixing}
    Let $H\in \mathbb{F}_2^{m\times n}$ be a parity check matrix satisfying $(\ell, d, \nu, \chi)-\mathsf{CC}$ as in \cref{definition:connectedness} for any choice of constants, and assume its syndrome network is $c$-\textsf{uniformly-amenable} as in \cref{definition:amenable}. Fix $L$ satisfying $|\mathcal{B}^{L}|\leq m^\chi$.
    
    Then, there exists a temperature $\beta_0(d)$, such that so long as $\beta>\beta_0(d)$, the conditional block dynamics converges quickly to the syndrome Gibbs distribution when initialized from the ground state. That is, 
    \begin{equation}
         \|\pi_\beta - \mathsf{P}_{\mathsf{syn}}^t (0^{m})\|_1\leq  m\cdot e^{-t/m}+ \max(m, t)\exp\bigg[-\Omega\bigg(\big(\beta-\beta_0(d)\big) L^{1/c}\bigg)\bigg]\,.
    \end{equation}
\end{theorem}

By virtue of the reduction in \cref{lemma:q_reduction} from the Lindbladian dynamics to the syndrome Markov chain, our main result in \cref{theorem:cbd_intro} then follows by an immediate corollary.

\begin{proof}[Proof of \cref{theorem:cbd_intro}] Applying the main result of this section on the evolution of the conditional block dynamics (\cref{theorem:cbd_rapid_mixing}) to the reduction in \cref{lemma:q_reduction}, we have
    \begin{align}
        \forall t\geq 0:\quad   \| e^{t\mathcal{L}}(\psi) - \rho^\psi_\beta\|_1 &\leq  2\cdot \max_{b\in \{X, Z\}}\mathbb{E}_{k \sim \mathsf{Pois}(mt)} \| (\mathsf{P}^b_{\mathsf{syn}})^k(0^{m_b}) - \pi_\beta^b\|_1\\  &\leq 2m\cdot \mathbb{E}_{k\sim \mathsf{Pois}(mt)}[e^{-k/m}] + m\cdot t\cdot \exp\bigg[-\Omega\bigg(\big(\beta-\beta_0(d)\big) L^{1/c}\bigg)\bigg]\,,
    \end{align}
    \noindent where we used that the bound becomes trivial unless $t\geq 1$. Finally, leveraging the moment generating function for the Poisson distribution, 
    \begin{equation}
     \forall a\in [0, \frac{1}{2}]:\quad    \mathbb{E}_{k\sim \mathsf{Pois}(\lambda)}[e^{-a\cdot k}]  \leq e^{-a\cdot \lambda/2}\,,
    \end{equation}
    \noindent which, after adequating the choice of constants, concludes the desired claim. 
\end{proof}

\begin{corollary} [Gibbs State Preparation]
    In the context of \cref{theorem:cbd_intro}, when initialized from the maximally mixed state over the degenerate ground space, the Lindbladian dynamics $\mathcal{L}$ rapidly converges to the global Gibbs state. That is, 

    \begin{equation}
         \forall t\geq 0:\quad   \| e^{t\mathcal{L}}\big(\bar{\mathbb{I}}\big) - \rho_\beta\|_1 \leq c_1\cdot n \cdot e^{-c_2 t} + t\cdot n\cdot e^{-(\beta-\beta_\mathcal{M})\cdot \log^{c_3} n}\,.
    \end{equation}
\end{corollary}

\subsection{The Small-Loop Configurations $\Omega_R$}
\label{subsection:small_loops}

\begin{definition}
    [The Reachable Set] We define $\Omega^{\mathsf{Reach}}\subset \mathsf{Im}(H)$ to be the collection of configurations which are in the same connected component in the valid syndrome space as the initial state $0^{ m}$ under transitions of the conditional block dynamics chain $\mathsf{P}_L$.
\end{definition}

\begin{remark}[All syndrome components are \textsf{erasable} in $\Omega^{\mathsf{Reach}}$]\label{remark:syndromes_in_box}
    In particular, the configurations in $\Omega^{\mathsf{Reach}}$ are comprised of connected components of syndromes, where each individual connected component is contained within some box $\mathcal{B}^L$. So long as $|\mathcal{B}^L|\leq m^\chi$, all connected components are individually \textsf{erasable} (cf. \cref{definition:connectedness}.5).
\end{remark}

\begin{claim}\label{claim:detailed_balance_reach}
    The conditional block dynamics chain $\mathsf{P}_L$ satisfies detailed balanced when restricted to $\Omega^{\mathsf{Reach}}$. In particular, its unique stationary distribution is the syndrome Gibbs distribution restricted to $\Omega^{\mathsf{Reach}}$.
\end{claim}

\begin{proof}[Proof of \cref{claim:detailed_balance_reach}] Consider any two configurations $x, y\in \Omega^{\mathsf{Reach}}$ which are adjacent under one step of the chain. In particular, $x, y$ must agree on all but the interior of the box $\mathsf{B}$, which has 0 boundary coniditions. We let $x|_\mathsf{B}, y|_\mathsf{B}$ be the violated syndromes in the bulk of $\mathsf{B}$. Note that 
\begin{equation}
    \big|x\big| + \big|y|_\mathsf{B}\big| =  \big|y\big| + \big|x|_\mathsf{B}\big|\,.
\end{equation}
From the Markov property \cref{lemma:markov}, the transition probability from $x$ to $y$ under $\mathsf{P}$ satisfies:
\begin{equation}
    \mathsf{P}(x\rightarrow y) \propto \exp\big(-\beta  \big|y|_\mathsf{B}\big|\big)\,.
\end{equation}
Further, the syndrome Gibbs distribution restricted to $\Omega^{\mathsf{Reach}}$ is proportional to
\begin{equation}
\forall z\in \Omega^{\mathsf{Reach}}: \pi_\beta(\Omega^{\mathsf{Reach}})(z) \propto \exp(-\beta\cdot |z|)\,.
\end{equation}
We have then that 
\begin{equation}
    \forall x, y\in \Omega^{\mathsf{Reach}}:\quad \pi_\beta(\Omega^{\mathsf{Reach}})(x) \times \mathsf{P}(x\rightarrow y) =  \pi_\beta(\Omega^{\mathsf{Reach}})(y) \times \mathsf{P}(y\rightarrow x)\,,
\end{equation}
\noindent which gives the desired claim.
\end{proof}

\begin{definition}
    [Small-Loop Configurations]\label{definition:small_loop} Fix $R<L$. We define $\Omega_R\subset \mathsf{Im}(H)$ to be the configurations wherein every connected component of syndromes in $x\in \Omega_R$ has size $\leq R$. 
    
    We refer to the syndrome Gibbs distribution restricted to $\Omega_R$ as
    \begin{equation}
        \pi_\beta\big|_{\Omega_R} (x) \propto \begin{cases}
            \pi_\beta(x) \quad x\in\Omega_R\\
            0 \quad\quad  \text{ otherwise}
        \end{cases}\,.
    \end{equation}
\end{definition}

\begin{claim}\label{claim:small_loops_near_gibbs}
Fix $R<L$, and assume $|\mathcal{B}^L|\leq m^\chi$. Then, so long as $\beta>\beta_0(d)$ as defined in \cref{lemma:ssm}, the syndrome Gibbs distribution restricted to $\Omega_R$ is close to the global syndrome Gibbs distribution:
\begin{equation}
   \big\| \pi_\beta\big|_{\Omega_R} -\pi_\beta\big\|_1\leq \delta_R \equiv 2m\times e^{-R(\beta-\beta_0(d))} \,.
\end{equation}
\end{claim}

\begin{proof}[Proof of \cref{claim:small_loops_near_gibbs}] The proof follows immediately from that of \cref{claim:all-erasable}. We first argue that with high probability, all the connected components of syndromes are individually erasable (\cref{claim:all-erasable}). The claim then follows from the Peierls argument of \cref{lemma:peierls_wsm} and the union bound. To conclude, we reason the conditions $|\mathcal{B}^L|\leq m^\chi$ and $R<L$ entail the latter error is larger than the former. 
\end{proof}

\begin{remark}
    We remark that if $R\leq L$, then $\Omega_R\subset \Omega^{\mathsf{Reach}}$. Therefore, the stationary distribution of $\mathsf{P}_L$ is also close to the syndrome Gibbs distribution:
    \begin{equation}\label{equation:reach_near_gibbs}
       \big\| \pi_\beta\big|_{\Omega^{\mathsf{Reach}}} - \pi_\beta\big\|_1\leq\delta_L\,.
    \end{equation}
\end{remark}

The following lemma proves that these small-loop configurations are metastable for time $\approx e^R$, regardless of if initialized from the ground state or from the syndrome Gibbs distribution. 

\begin{lemma}
    [The Small-Loop Configurations are Stable]\label{lemma:Pisstable} Fix an integer $t\geq 1$, $R<L$, and assume $|\mathcal{B}^L|\leq m^\chi$. Assume $\beta>\beta_0(d)$ as defined in \cref{lemma:ssm}. Let $X_0\in \Omega_R$ be some initial state, and $X_1, X_2\cdots X_t$ states generated by applying transitions of $\mathsf{P}_L$. Then, the probability $X_t\notin \Omega_R$ satisfies:
    \begin{equation}
    \forall x\in \Omega_R : \quad   \mathbb{P}\big[X_t\notin \Omega_R\big|X_0=x\big] \leq t\times \delta_{R}^{\mathsf{Escape}} \equiv t\times |\mathcal{B}^L| \times e^{-R(\beta-\beta_0(d))}\,.
    \end{equation}
\end{lemma}

\begin{proof}
    We claim that, given any configuration $x\in \Omega_R$, after one step of the chain $\mathsf{P}_L$, 
    \begin{equation}
        \forall x\in \Omega_R : \quad   \mathbb{P}\big[X_1\notin \Omega_R\big|X_0=x] \leq |\mathcal{B}^L| \times (de^{1-\beta})^R \times \frac{1}{1-de^{1-\beta}}\,,
    \end{equation}
    \noindent the proof then proceeds by a union bound over $t$ and the choice of $\beta>\beta_0(d)$.
    To prove this claim, we once again proceed via a Peierls argument. Following the definition of $\mathsf{P}_L$, suppose we fix a desired  box $\mathcal{B}^L$ and the syndrome configurations lying and traversing its boundaries. Then, we consider the subset of configurations $\mathcal{K}_{\mathsf{V}}\subset \Omega^{\mathsf{Reach}}$ reachable by $x$ under such an update, wherein there exists a new connected component $\mathsf{V}$ of size $\ell\geq R$. Note that \cref{remark:syndromes_in_box} implies the new component has size $|\mathsf{V}|\leq |\mathcal{B}^L|\leq m^\chi$, and thereby is $\mathsf{erasable}$ (\cref{definition:connectedness}).
    
    We can thereby apply the Peierls argument in \cref{lemma:peierls_wsm}, to conclude
    \begin{equation}
        \mathbb{P}[X_1\in \mathcal{K}_{\mathsf{V}}] \leq e^{-\beta \ell}\,.
    \end{equation}
    And following the now familiar union bound over $\ell$ and starting points $u$ we have the desired statement. 
\end{proof}

\subsection{Restricting the Chain to $\Omega_R$}
\label{subsection:chain_restriction}

Integral in our analysis will be to consider the mixing process when restricted to the small-loop fraction of the state space $\Omega_R$. In this subsection, we claim if this ``restricted chain" is properly initialized from within $\Omega_R$, then for time $\approx e^{R}$ it ``doesn't notice" it is running the restricted chain instead of the original chain $\mathsf{P}_L$.

\begin{definition}
    [Markov Chain restricted to $\Omega_R$] We define the $|\Omega_R|\times |\Omega_R|$ transition matrix $\mathsf{Q}$ to be the restriction of $\mathsf{P}_L$ to the associated submatrix. 
    \begin{equation}
        \forall x\neq y\in \Omega_R, \quad \mathsf{Q}_{x, y} = \mathsf{P}_{x, y}, \text{ and }\mathsf{Q}_{x, x} = 1-\sum_{y\neq x}\mathsf{Q}_{x, y}\,.
    \end{equation}
\end{definition}

In words, transitions out of $\Omega_R$ are rejected. The following lemma implies that due to the stability of the small-loop configurations $\Omega_R$, the original conditional chain $\mathsf{P}_L$ and the restricted chain $\mathsf{Q}$ are unlikely to differ: 

\begin{lemma}
    [$\mathsf{P}$ and $\mathsf{Q}$ are indistinguishable]\label{lemma:pq_indistinguishable} Let $x\in \Omega_R$ be an arbitrary starting point. In the context of \cref{lemma:Pisstable},  
    \begin{equation}
        \big\|\mathsf{P}^t_L(x) - \mathsf{Q}^t(x)\big\|_1\leq t\times  2\delta_R^{\mathsf{Escape}} = t\times 2|\mathcal{B}^L|e^{-R\cdot (\beta-\beta_0(d))}\,.
    \end{equation}
\end{lemma}

\begin{proof}
    Let $v_i = \mathsf{P}^{t-i}\mathsf{Q}^{i}(x)$. By the triangle inequality, we can reduce the distinguishability of the chains after $t$ steps, to a single step after $i$ steps of $\mathsf{Q}$.

    \begin{equation}
        \big\|\mathsf{P}^t(x) - \mathsf{Q}^t(x)\big\|_1 =  \big\|v_0-v_t\big\|_1 \leq \sum_i \big\|v_{i+1}-v_i\big\|_1\,.
    \end{equation}

    \noindent Let us unpack the $i$th term in the RHS above. 
    \begin{align}
        \big\|v_{i+1}-v_i\big\|_1 &= \big\|\mathsf{P}^{t-i-1}\big(\mathsf{Q}^{i+1}(x) - \mathsf{P}\mathsf{Q}^{i}(x)\big)\big\|_1  \\&\leq  \big\|\mathsf{Q}^{i+1}(x) - \mathsf{P}\mathsf{Q}^{i}(x)\big\|_1 \tag{Data-Processing Inequality}\\
        &\leq \max_{y\in \Omega_R} \|\mathsf{Q}(y) - \mathsf{P}(y)\|_1\,. \tag{\textsf{Q} is restricted to $\Omega_R$}
    \end{align}

    \noindent In this manner, it suffices to consider the distinguishability of $\mathsf{P}, \mathsf{Q}$ starting from any element of $\Omega_R$. Further, 
    \begin{align}
        \big\|\mathsf{P}(y) - \mathsf{Q}(y)\big\|_1 &= |\mathsf{P}_{y, y}-\mathsf{Q}_{y, y}| + \sum_{x\in \Omega^{\mathsf{Reach}} \setminus \Omega_R} \mathsf{P}_{y, x} \\&\leq 2\cdot \sum_{x\in \Omega^{\mathsf{Reach}} \setminus \Omega_R} \mathsf{P}_{y, x} \leq 2 \sum_{x\notin \Omega_R} \mathsf{P}_{y, x} \leq 2\cdot \delta_R^{\mathsf{Escape}}\,,
    \end{align}
    where in the last inequality we leveraged the stability \cref{lemma:Pisstable}. Plugging in above, we get
    \begin{equation*}
        \big\|\mathsf{P}^t(x) - \mathsf{Q}^t(x)\big\|_1\leq 2\cdot t\cdot \delta_R^{\mathsf{Escape}}\,. \qedhere
    \end{equation*}
\end{proof}

\subsection{The Restricted Chain is Fast-Mixing}
\label{subsection:fast_mixing}

In this section, we prove that the restricted chain $\mathsf{Q}$ is rapid-mixing even from a worst case initialization within $\Omega_R$. To do so, we proceed via a delicate path coupling argument. In the subsequent subsection, we put all the ingredients together to prove the original chain $\mathsf{P}_L$ is fast mixing. 

\begin{lemma}
    [The Restricted Chain $\mathsf{Q}$ is Fast Mixing] \label{lemma:Q_fast_mixing} Assume $\beta>\beta_0(d)$ as defined in \cref{lemma:ssm}, that $G$ is $c$-\textsf{uniformly-amenable} as in \cref{definition:amenable}, and that

    \begin{equation}
        R \geq \frac{8\ln m}{\beta-\beta_0(d)} \text{ and }L \geq c_1\cdot R^c \,,
    \end{equation}

    \noindent for an explicit, sufficiently large constant $c_1$. Then, there exists an $\alpha \ge 1/m$ s.t. for any two starting points $x, y\in \Omega_R$, 
    \begin{equation}
        \| \mathsf{Q}^t(x) - \mathsf{Q}^t (y)\|_1 \leq  4m\times e^{-\alpha\cdot t}\,.
    \end{equation}
\end{lemma}

Integral in our analysis will be the following metric of convergence (or coalescence) between states in $\Omega_R$, the cluster distance. 

\begin{definition}
    [The Cluster Distance] We define a graph $\mathcal{G} = (\Omega_R, E)$ over the state space $\Omega_R$, by connecting $x, y\in \Omega_R$ by edge if they differ by the addition or removal of exactly one connected component of syndromes. The distance $d_{\mathsf{Cluster}}(x, y)$ is the shortest path length from $x$ to $y$ in $\mathcal{G}$.
\end{definition}

\begin{remark}\label{remark:max_cycle_distance}
    For any two $x, y\in \Omega^{\mathsf{Reach}}$, $0\leq d_{\mathsf{Cluster}}(x, y)\leq 2\cdot m$.
\end{remark}

The central claim of this section \cref{claim:cycle_distance_contracts}, proves that the expected cluster distance of two small-loop configurations $x, y\in \Omega_R$ after one step of the chain, contracts by a constant factor. In fact, following the path-coupling method discussed below, it suffices to consider two adjacent configurations under $\mathcal{G}$. 

\begin{claim}
    [The Cluster Distance Contracts in $\Omega_R$]\label{claim:cycle_distance_contracts} Let $X, Y\in \Omega_R$ be such that  $d_{\mathsf{Cluster}}(X, Y)=1$. Assume $\beta>\beta_0(d)$, $G$ is $c$-\textsf{uniformly-amenable} (\cref{definition:amenable}), and that 

    \begin{equation}
        L \geq c_1 R^{c}\,,\quad  \text{ and }\quad  R\geq \frac{8  \ln m}{\beta-\beta_0(d)}\,,
    \end{equation}
    for an explicit (sufficiently large) constant $c_1$. Then, there exists a coupling $\omega$ such that after one step of $\mathsf{Q}$, $(X', Y')\leftarrow \omega$ satisfies
    \begin{equation}
       \mathbb{E}_\omega[d_{\mathsf{Cluster}}(X', Y') ] \leq 1-\alpha\equiv 1 - \frac{1}{m}
       \,.
    \end{equation}
\end{claim}

The proof of \cref{claim:cycle_distance_contracts} is quite technical, and requires a careful combination of the separating surfaces approach of the previous \cref{section:ssm} with the combinatorial approach to block-dynamics of \cite{Dyer2002MixingIT}. Before we present the proof, let us explain how to conclude \cref{lemma:Q_fast_mixing}, on the mixing time of the restricted chain $\mathsf{Q}$; and, subsequently, the mixing time of $\mathsf{P}$ in \cref{theorem:cbd_rapid_mixing}.

\begin{proof}[Proof of \cref{lemma:Q_fast_mixing}] Following standard arguments (cf.~\cite{guruswami2016rapidlymixingmarkovchains}), the coupling ensured by \cref{claim:cycle_distance_contracts} over adjacent $(A, B)\in \Omega_R\times \Omega_R$ can be lifted to arbitrary $(X_t, Y_t)\in  \Omega_R\times \Omega_R$, in such a manner that for any $X_t,Y_t$, 
    \begin{equation}
        \mathbb{E}_\omega\big[d_{\mathsf{Cluster}}(X_{t+1}, Y_{t+1})\big] \leq \big(1-\alpha\big)\times d_{\mathsf{Cluster}}(X_{t}, Y_{t})\,.
    \end{equation}
    \noindent Which iteratively implies
    \begin{equation}
         \mathbb{P}[X_t\neq Y_t] \leq\mathbb{E}_\omega\big[d_{\mathsf{Cluster}}(X_{t+1}, Y_{t+1})\big] \leq 2m\times \big(1-\alpha\big)^t \leq 2m\times e^{-\alpha\cdot t}\,.
    \end{equation}
    where we used \cref{remark:max_cycle_distance} to upper bound $d_{\mathsf{Cluster}}(X_0, Y_0)$. To conclude, we have 
    \begin{equation*}
        \|\mathsf{Q}^t(x)-\mathsf{Q}^t(y)\|_1 = 2\sup_{A\subset \Omega_R} \big|\mathsf{Q}^t(x)\big(A\big)-\mathsf{Q}^t(y)\big(A\big)\big| \leq 2\cdot \mathbb{P}[X_t\neq Y_t]\leq 4m\times e^{-\alpha\cdot t}\,. \qedhere
    \end{equation*}
\end{proof}

\noindent We are now in a position to conclude the proof of the main theorem of this section, \cref{theorem:cbd_rapid_mixing}.

\begin{proof}[Proof of \cref{theorem:cbd_rapid_mixing}] Our goal is to understand $\|\mathsf{P}^t(0^{\times m}) - \pi_\beta\|_1$. Using \cref{equation:reach_near_gibbs}, \cref{claim:detailed_balance_reach},  \cref{equation:reach_near_gibbs} again, and then \cref{claim:small_loops_near_gibbs}:\footnote{For conciseness, for two distributions $p, q$ we say $p\approx_\delta q$ if $\|p-q\|_1\leq \delta$.} 
    \begin{equation}
        \pi_\beta \approx_{\delta_L} \pi_\beta\big|_{\Omega^{\mathsf{Reach}}}  = \mathsf{P}^t \pi_\beta\big|_{\Omega^{\mathsf{Reach}}}  \approx_{\delta_L}\mathsf{P}^t \pi_\beta \approx_{\delta_R} \mathsf{P}^t \pi_\beta\big|_{\Omega_R}\,.
    \end{equation}

    \noindent The stability of $\Omega_R$ then entails (\cref{lemma:pq_indistinguishable})
    \begin{equation}
        \big\| \mathsf{Q}^t\pi_\beta\big|_{\Omega_R} - \mathsf{P}^t \pi_\beta\big|_{\Omega_R}\big\|_1 \leq 2t\cdot  \delta_R^{\mathsf{Escape}}\,, \quad  \text{ and } \quad  \big\| \mathsf{Q}^t(0^{\times m}) - \mathsf{P}^t (0^{\times m})\big\|_1\leq 2t\cdot  \delta_R^{\mathsf{Escape}}\,.
    \end{equation}

    \noindent We can next leverage that $\mathsf{Q}$ is fast mixing (\cref{lemma:Q_fast_mixing}), to obtain
    \begin{equation}
       \big \|\mathsf{P}^t \pi_\beta\big|_{\Omega_R} - \mathsf{P}^t (0^{\times m})\big\|_1\leq 4t\cdot  \delta_R^{\mathsf{Escape}}+4me^{-\alpha t}\,.
    \end{equation}
    We can now conclude by combining all the above
    \begin{equation}
        \big\|\pi_\beta - \mathsf{P}^t (0^{\times m})\|_1 \leq \delta_R+2\delta_L + 4t\cdot  \delta_R^{\mathsf{Escape}} + 4me^{-\alpha t}\,.
    \end{equation}
    The appropriate choice of parameters $R = \Theta(L^{1/c})$ then gives the desired claim. 
\end{proof}

\subsection{The Cluster Distance Contracts in $\Omega_R$}
\label{subsection:cycle_contraction}

We dedicate this section to the deferred proof of \cref{claim:cycle_distance_contracts}, that the cluster distance contracts in expectation after one step of $\mathsf{Q}$.

\begin{proof}[Proof of \cref{claim:cycle_distance_contracts}] Suppose $X, Y\in\Omega_R$ are any two adjacent configurations, in that $X, Y$ differ in the addition of a specific (\textsf{erasable}, cf. \cref{remark:syndromes_in_box}) connected component of syndromes $\mathsf{V}$ of size $\leq R$. After one block update from $\mathsf{Q}$, $X, Y\rightarrow X', Y'$ (say, in some block $\mathcal{B}_u^L$), there are three possibilities:\\

\noindent \textbf{Case 1.} The update misses the component, $\mathsf{V}\cap \mathcal{B}_u^L = \emptyset$. In which case, since $x, y$ agree on $\mathcal{B}^L_u$, we have 
        \begin{equation}
            d_{\mathsf{Cluster}}(X', Y')=d_{\mathsf{Cluster}}(X, Y)=1\,.
        \end{equation}

\noindent \textbf{Case 2.} The update entirely contains the component, $\mathsf{V}\subset \mathcal{B}_u^L$. In which case, since $X, Y$ agree on the boundary $\partial \mathsf{B}_u^L$ and on all components incident on the boundary, after the update they match:
        \begin{equation}
            X'=Y'\Rightarrow d_{\mathsf{Cluster}}(X', Y')=0\,.
        \end{equation}

\noindent \textbf{Case 3.} The component lands on the boundary of $\mathcal{B}_u^L$, $\mathsf{V}\cap\partial \mathsf{B}_u^L \neq \emptyset$. We claim (and prove shortly, in \cref{claim:boundary_coupling}) that in said case we can engineer a coupling $\omega$ such that the cluster-distance increases by a limited amount in expectation: 
        \begin{equation}\label{equation:boundary_coupling}
           \forall r\geq 1:\quad  \mathbb{E}_\omega[d_{\mathsf{Cluster}}(X', Y')\big| \mathsf{V}\cap\partial \mathcal{B}_u^L \neq \emptyset] \leq |\mathcal{B}^r_\mathsf{V}| (1+ z\cdot |\mathcal{B}^L_u|\times
          e^{-(\beta-\beta_0(d))r/2})\,,
        \end{equation}
for an explicit constant $z$. Before proving \cref{equation:boundary_coupling}, let us show how the cases above conclude the proof of \cref{claim:cycle_distance_contracts}.  Using elementary geometry, let us bound the probability of each of the events above. Recall $|\mathsf{V}|\leq R$. 
    \begin{equation}
       \text{Case 1: } \mathbb{P}[\mathsf{V}\cap \mathcal{B}_u^L = \emptyset] \leq 1 - \frac{|\mathcal{B}^{L}_\mathsf{V}|}{m}\,, \quad  \text{and Case 3: } \mathbb{P}[\mathsf{V}\cap\partial \mathcal{B}_u^L \neq \emptyset] \leq \frac{\max_u|\partial_R\mathcal{B}^{L}_u|}{m}\equiv \frac{|\partial_R\mathcal{B}^{L}|}{m}\,.
    \end{equation}

    \noindent where case 3 captures the number of sites $u\in [m]$ which are at distance $\geq L-R$ but $\leq L$ from $\mathsf{V}$. To proceed, we pick $R, r\leq L$ satisfying 
    \begin{equation}
    R\geq r \geq \frac{8}{\beta-\beta_0(d)} \ln m \geq \frac{8}{\beta-\beta_0(d)} \ln z|\mathcal{B}^L|\,.
    \end{equation}

    \noindent We can now proceed by case analysis, by considering the expected change in the cluster-distance,
    \begin{align}
        \mathbb{E}_\omega[d_{\mathsf{Cluster}}(X', Y')] &= \mathbb{P}[\mathsf{V}\cap \mathcal{B}_u^L = \emptyset] + \mathbb{E}_\omega[d_{\mathsf{Cluster}}(X', Y')\big| \mathsf{V}\cap\partial \mathcal{B}_u^L \neq \emptyset]\cdot \mathbb{P}[\mathsf{V}\cap\partial \mathcal{B}_u^L \neq \emptyset] \\
        &\leq 1 - \frac{\min_u|\mathcal{B}_u^{L}|}{m}\times \bigg(1 - \frac{|\partial_R\mathcal{B}^{L}|\cdot |\mathcal{B}^{R+r}|}{\min_u|\mathcal{B}^L_u|}\bigg)\\& \leq 1 - \frac{\min_u|\mathcal{B}_u^{L}|}{2m}\leq 1 - \frac{1}{m}\,.
    \end{align}

\noindent So long as $G$ is $c(G)$-\textsf{uniformly-amenable} as in \cref{definition:amenable}, and $L \geq c_1\cdot R^{c(G)}$ for some explicit constant $c_1$.
\end{proof}

Let us now return to the claim above in \cref{equation:boundary_coupling}, restated below:

\begin{claim}
    [Coupling two updates which differ on a block boundary]
    \label{claim:boundary_coupling} In the context of the proof of \cref{claim:cycle_distance_contracts}, $\forall r\geq 1$,
    \begin{equation}
        \mathbb{E}_\omega[d_{\mathsf{Cluster}}(X', Y')\big| \mathsf{V}\cap\partial \mathcal{B}_u^L \neq \emptyset] \leq |\mathcal{B}^r_\mathsf{V}| (1+ z\cdot |\mathcal{B}^L|\times
          e^{-(\beta-\beta_0(d))r/2}) \,.
    \end{equation}
    \noindent where $z=20\cdot d$ is an explicit constant. 
\end{claim}

The crux of this proof is to leverage the domain $\mathsf{SSM}$ result of \cref{lemma:ssm}. Roughly speaking, since the two configurations $X, Y$ only differ on a single small component incident on the boundary $\partial\mathsf{B}_u^L$, at a distance $r$ away the conditional Gibbs measure should be approximately independent of the presence of that component. We can then couple the configurations $X', Y'$ at distance $\geq r$ away, resulting in roughly a $\leq |\mathcal{B}^r|$ increase in the hamming/cluster distance. 

\begin{proof}[Proof of \cref{claim:boundary_coupling}] Recall $X, Y$ differ only on the differing component $\mathsf{V}=\mathsf{V}_0$, and agree on all other connected components $\mathsf{V}_1,\cdots \neq \mathsf{V}$ traversing $\partial\mathcal{B}^L$. By stitching their neighborhoods together (without $\mathsf{V}$), we can define an ``interior" region of $\mathcal{B}^L$:

    \begin{equation}\label{eq:contractionproof}
        \mathsf{B}\equiv \mathcal{B}^L \setminus \bigcup_{i=1} \big(N(\mathsf{V}_i)\cup \mathsf{V}_i\big)\,.
    \end{equation}
    
    By definition, the syndromes on $N(\mathsf{V}_i)=0$ (in $X$ and $Y$), and therefore $X_{\partial \mathsf{B}}=0$. Similarly, if we remove $\mathsf{C} = N(\mathsf{V})\cup \mathsf{V}$ from $\mathsf{B}$, we have $Y_{\partial (\mathsf{B}\setminus \mathsf{C})}=0$. We can therefore apply the domain \textsf{SSM} property to the three regions $\mathsf{B}$, $\mathsf{C}\subset \mathsf{B}$, and $\mathsf{B}\setminus  \mathcal{B}^{r}_{\mathsf{V}}$.

    From \cref{lemma:ssm}, with high probability via $X', Y'$ agree on a separating half-surface $\mathsf{S}$ around the differing component $\mathsf{V}$. We have that $\mathsf{S}\cup (\partial\mathsf{B}\cap \partial(\mathsf{B}\setminus \mathsf{C}))$ then screens $\mathsf{C} = N(\mathsf{V})\cup \mathsf{V}$ from $\mathsf{B}\setminus  \mathcal{B}^{r}_{\mathsf{V}}$, and naturally partitions the block $\mathcal{B}^L$ into an ``inside" region around the component $\mathsf{V}$, of volume $\leq \mathcal{B}^r_\mathsf{V}$ and an ``outside" region. By the Markov property (\cref{lemma:markov}), there is a coupling $\omega$ where $X'$ and $Y'$ agree on the entirety of the outside region. We can now compute the change in the hamming weight between $X', Y'$, which upper bounds the change in cluster distance. 
    \begin{align}
         \mathbb{E}_\omega[d_{\mathsf{Cluster}}(X', Y')\big| \mathsf{V}\cap\partial \mathcal{B}^L \neq \emptyset] \leq |\mathcal{B}^r_\mathsf{V}| + |\mathcal{B}^L|\times |\mathcal{B}^1_\mathsf{V}|\cdot 20\cdot d
          e^{-(\beta-\beta_0(d))r/2}\,,
    \end{align}
    which implies the desired claim.
\end{proof}

\section{Efficient Block Updates for the 4D Toric Code}
\label{section:locality}

Naively, the runtime to implement a quasi-local block update may cost exponentially in the size of the update. Since the volume of each block in our algorithm is $\polylog n$, this results in quasi-polynomial runtime when implemented using $2$-qubit gates. Indeed, this is the situation encountered in other quantum Gibbs samplers based exclusively on decay of correlations, such that of \cite{Brando2016FiniteCL}. In this section, we reason that the individual $\polylog(n)$ local block updates for the 4D toric code can each be implemented in $\polylog \, n$ time, resulting in the following theorem:

\begin{theorem}\label{theorem:toric_low_depth_circuit}
    There exists a threshold temperature $\beta_0$, s.t. $\forall \beta>\beta_0$ the Gibbs state within a logical sector $\rho_{\beta, \psi}$ of the toric code on a 4D hypercubic torus  $(\mathbb{Z}/w\mathbb{Z})^4$, can be prepared using a local quantum circuit of $\polylog \,(w)$ depth acting on the ground state $\psi$, up to error $2^{-\polylog \,(w)}.$
\end{theorem}

The central building block is the following lemma, which shows how to implement the conditional block updates of \cref{definition:cbd}, on the 4D toric code, by appealing to a connection to the \textit{worm process} for the 2D Ising model \cite{Jerrum1990PolynomialTimeAA, PS01, CGTT16}. 

\begin{theorem}
    [Efficient Block Updates on the 4D toric Code]\label{theorem:block_updates_toric} Consider the 4D toric code on a 4D hypercubic torus $(\mathbb{Z}/w\mathbb{Z})^4$, and fix $L\leq w^{1/4}$. A conditional block update (\cref{definition:cbd}) on a block $\mathcal{B}^L$ of syndromes, corresponding to a sample from the conditional syndrome Gibbs distribution on said block, can be implemented up to TVD $\epsilon$ using $O(|\mathcal{B}^L|^7\cdot \log \frac{|\mathcal{B}^L|}{\epsilon}) = \poly(L)\log\frac{1}{\epsilon}$ 2-qubit gates and ancillas. 
\end{theorem}

By combining with \cref{theorem:cbd_rapid_mixing}, \cref{theorem:toric_low_depth_circuit} is an immediate corollary. 

Roughly speaking, the worm process is an efficient algorithm to sample from a distribution over even subgraphs (that is, collections of cycles, \cref{definition:even_subgraphs}) of any graph $G$, as weighted by the total size of the subgraph. We do not spell out all the details of the algorithm here, and instead refer to \cite{CGTT16} for intuition and analysis; nevertheless, we state its guarantees in \cref{section:even_subgraph}. In \cref{section:toric_is_even_subgraph}, we spell out the equivalence between the even subgraph model and the syndrome distribution of the 4D toric code, which proves \cref{theorem:block_updates_toric}.  

\subsection{The Even Subgraph Model}
\label{section:even_subgraph}

An even subgraph of a graph is a collection of cycles on the edges of the graph. 

\begin{definition}
    [Even Subgraphs]\label{definition:even_subgraphs} Given a finite graph $G = (V, E)$, a subset $A\subseteq E$ is said to be an even subgraph if all the vertices in $(V, A)$ have even degree. We let $\mathcal{C}_0 \subseteq \{0, 1\}^{|E|}$ denote the collection of even subgraphs of $G$.
\end{definition}

One can define a Gibbs measure on even subgraphs, as weighted by the total length of the cycles. 

\begin{definition}
    [The even subgraph distribution]
    Given a finite graph $G = (V, E)$ and an inverse-temperature $\beta$, we define the measure $\mu_\mathsf{even}$ on even subgraphs:

    \begin{equation}
       \forall A\in \mathcal{C}_0:\quad \mu_\mathsf{even}(A) = \frac{e^{-\beta |A|}}{\sum_{S\in  \mathcal{C}_0} e^{-\beta |S|} }\,.
    \end{equation}
\end{definition}

\begin{theorem}
    [The Worm Process is Rapid Mixing \cite{CGTT16}]\label{theorem:worm_mixing} Fix a graph $G = (V, E)$ of degree at most $d$. There exists an algorithm to sample from a distribution $\epsilon$ close to $\mu_\mathsf{even}$ in TVD, which runs in time $O(d|V|^5|E|^2\log \frac{|V|}{\epsilon})$. 
\end{theorem}

\subsection{The 4D Toric Code as an Even Subgraph}
\label{section:toric_is_even_subgraph}

We refer the reader to \cref{section:toric_background} for a background on the 4D toric code. The $X$ parity checks of the 4D toric code lie on edges of a 4D hypercubic lattice, where the valid syndrome configurations form cycles on the edges of the lattice. We claim that to sample from the conditional distribution over syndromes of the  4D toric code, it suffices to sample from an even subgraph distribution on a subgraph of the hypercubic lattice.  

\begin{claim}
    [The 4D toric code as an Even Subgraph]\label{claim:toric_is_even_subgraph}
    Consider the 4D toric code on the a 4D hypercubic torus  $(\mathbb{Z}/w\mathbb{Z})^4$. Fix $L\leq (w/2)^{1/4}$. 

    \begin{enumerate}
        \item The set of valid $X$ syndromes of the 4D toric code, when restricted to a subset $\mathsf{M}\subseteq \mathcal{B}^L$ of syndromes with boundary $\partial \mathsf{M}$ pinned to $0$, are the even subgraphs of a restriction of $(\mathbb{Z}/w\mathbb{Z})^4$ to the edges in $\mathsf{M}$.

        \item The distribution sampled from the conditional block dynamics on $\mathsf{M}$ (\cref{definition:cbd}) is the even subgraph distribution on said restriction of the hypercubic lattice. That is, \begin{equation}
        \pi_\beta(s_{\mathsf{M}}|s_\mathsf{\partial M} = 0) = \mu_\mathsf{even}(s_\mathsf{M}) \propto e^{-\beta |s_{\mathsf{M}}|}\,,
    \end{equation}
    if $s_{\mathsf{M}}\in \{0, 1\}^{\mathsf{M}}$ is an even subgraph of $\mathsf{M}$.
    \end{enumerate}
\end{claim}

In this manner, we can black-box leverage the worm process statement in \cref{theorem:worm_mixing}; \cref{theorem:block_updates_toric} follows as an immediately corollary.

\begin{proof}[Proof of \cref{claim:toric_is_even_subgraph}] By definition of the meta-checks in the 4D toric code (cf. \cref{section:toric_background}), the $X$ syndrome configurations in the 4D toric code correspond to cycles on the edges of the 4D hypercubic torus. To establish Conditions 1+2 in \cref{claim:toric_is_even_subgraph}, it then suffices to prove that any collection of cycles within the subset $\mathsf{M}$ (under pinned boundaries) is also valid syndrome configuration. 

To this end, note that the constraint that $\mathsf{M}\subseteq \mathcal{B}^L$ implies
\begin{equation}
    |\mathsf{M}|\leq |\mathcal{B}^L|\leq L^4\leq (w/2)\,.
\end{equation}

\noindent Further, since the boundary $\partial \mathsf{M}$ has syndromes pinned to $0$, this implies that all syndrome configurations within $\mathsf{M}$ are topologically non-trivial; and thereby individually erasable from \cref{definition:connectedness}. Thereby, any even subgraph of $\mathsf{M}$ defines a valid syndrome configuration.
\end{proof}

\printbibliography

\appendix

\section{The 4D Toric Code}
\label{section:toric_background}

The 4D toric code is a generalization of Kitaev's surface code model to higher dimensions. On a 4D hypercubic lattice, one can define points, edges, faces (squares), cubes and tesseracts.  

\begin{definition}
    [The 4D toric code] Fix an integer $w>2$ and a tesselation of a 4D hypercubic torus of side-length $w$, $(\mathbb{Z}/w\mathbb{Z})^4$. The 4D toric code is a $[[6w^4, 6, w^2]]$ CSS code given by placing qubits on all the faces, $X$-checks on edges (as products of single-qubit Pauli $X$ on the 6 adjacent faces to said edge) and $Z$-checks on cubes.
\end{definition}

The logical operators of the 4D toric code i.e. the non-trivial operators with 0 syndrome, are 2D surfaces/toruses in $(\mathbb{Z}/w\mathbb{Z})^4$. We remark that the action of the $Z$ checks on the lattice is identical to the action of the $X$ checks on the dual lattice.

The integral feature of the 4D toric code, as opposed to its lower-dimensional counterparts, is that its parity checks exhibit redundancies. The product of all 8 $X$ stabilizers (or alternatively, the XOR of the parity checks) associated to edges incident on any given vertex, must be $0$. What this implies is that the violated $X$ syndromes form a collection of loops (i.e. all vertices have even degree); they cannot form strings of violations which terminate/have end-points.

As a consequence, it admits the connectedness criterion as detailed in  \cref{definition:connectedness}. To be explicit, we spell out the parameters here.

\begin{fact}
    The 4D toric code admits the $(6, 14, 5, \frac{1}{4})-\mathsf{CC}$ connectedness criterion.
\end{fact}

\begin{proof}

    As noted above, both $X$ and $Z$ parity checks are of weight $\ell = 6$ -- the 6 faces adjacent to edges. The metachecks are of weight 8 -- 8 edges surrounding a vertex -- and the syndrome network graph has degree $d = 14$ -- each endpoint of an edge, induces linear dependencies with 7 other edges. 

    The energy barrier of the toric code scales with its side-length, corresponding to the formation of a topologically nontrivial loop. Thus, $\chi$ approaches $\frac{1}{4}$. Below such a scale, the syndrome patterns are topologically non-trivially and thereby erasable. 

    The fact that topologically non-trivial syndrome patterns confined to balls are locally erasable, is based on a iterative loop contraction algorithm, see e.g.~\cite[Fig.~1]{Alicki2008OnTS}. An analogous discussion for higher-dimensional toric and color codes can be found in~\cite[Section~7]{Bombin2009SelfcorrectingQC}.

    Finally, topologically non-trivial loops in groups of 4 can be contracted together; this is since the $1$st Homology group of the 4D hypercubic torus is of dimension 4 as a $\mathbb{Z}_2$ vector space. See the discussion on~\cite[Section~7, page~14]{Bombin2009SelfcorrectingQC}.
\end{proof}

\section{Open Quantum Systems}
\label{section:open_quantum_systems}

\subsection{Lindbladian Dynamics and Davies Generators}

We dedicate this subsection to background on the evolution and convergence of open quantum systems described by Lindbladian dynamics. Recall, a general \textit{purely irreversible} Lindbladian $ \mathcal{L}:\mathcal{B}(\mathcal{H})\rightarrow \mathcal{B}(\mathcal{H})$ is a continuous-time Markov chain acting on density operators\footnote{Where $\{A,B\}:=AB+BA$ is the anti-commutator.}

\begin{equation}
    \mathcal{L}[\rho] = \sum_j J_j \rho J_j^\dagger -\frac{1}{2}\big\{J_j^\dagger J_j, \rho\big\}, \quad \text{for some set of Lindblad operators } \{J_i\}
\end{equation}

\noindent which generates a family of completely positive and trace-preserving maps

\begin{equation}
    e^{\mathcal{L}t}[\rho] \text{ for each }t\geq 0.
\end{equation}

\noindent A density operator $\sigma$ is said to be a fixed point of the dynamics if $\mathcal{L}[\sigma]=0$, which implies that $e^{\mathcal{L}t}[\sigma] = \sigma$. 

\subsubsection{Thermal Lindbladians and Davies Generators}

Of particular interest will be a subclass of Lindbladians known as Davies Generators~\cite{Davies1974MarkovianME}, which describe the dynamics of a quantum system coupled to a thermal bath. Let $\mathcal{H}$ be a commuting Hamiltonian with integer spectra $\in [0, m]$, and let $\Pi_k$ be the orthogonal projection onto the eigenspace of energy $k$. The Davies Generator associated to $\mathcal{H}$, is specified by:

\begin{itemize}
    \item \textbf{A collection of jump operators} to generate the transitions, $\{A_a\}$. For instance, the set of single qubit Pauli operators
    \begin{equation}
        \{A_a\} = \big\{P_i\otimes \mathbb{I}_{[n]\setminus i}: i\in [n], P\in \{\mathbb{I}, X, Y, Z\}\big\}
    \end{equation}

    \noindent In contrast to classical Markov Chain transitions, these quantum jumps may change the energy of the system in superposition. Thereby, it will be convenient to decompose the jump operators into the energy basis:
\begin{align}
\forall \nu \in [-m, m]:\quad A^a_\nu := \sum_{k\in [m]} \Pi_{k+\nu} A^a \Pi_k \quad \text{such that}\quad \sum_{\nu\in [-m, m]} A^a_\nu = A^a.  
\end{align}

\item \textbf{The transition weights}. Here we always select the Glauber dynamics weight,
    \begin{equation}
    \gamma(\nu) = 1/(1+e^{\beta \nu}) \quad \text{for all }\nu\in [-m, m].
    \end{equation}    
Put together, the associated family of Davies generators $\mathcal{L}$ can be written down as
\begin{equation}\label{equation:Lindbladian}
    \mathcal{L}[\rho] = \sum_{a \in \mathcal{A}} \sum_{\nu}\gamma(\nu) \bigg(A^a_\nu \rho (A^a_\nu)^\dagger -\frac{1}{2} \bigg\{(A^a_\nu)^\dagger A^a_\nu, \rho \bigg\}\bigg).
\end{equation}
\end{itemize}

So long as the choice of jumps is made to ensure $\mathcal{L}$ is ergodic, the Gibbs state $\rho\propto e^{-\beta \mathcal{H}}$ of $\mathcal{H}$ is the unique fixed point of $\mathcal{L}$.

\subsection{Quantum Memories}

We now elaborate on the particular dynamics which define quantum memories. To do so, we consider the Davies generator associated to a code Hamiltonian $\mathcal{H}$ (\cref{definition:code_hamiltonian}) of a CSS code $\mathsf{CSS}(H_X, H_Z)$, of $X$ checks $\{X^{a_i}:a_i\in H_X\}$ and $Z$ checks $\{Z^{b_j}: b_j\in H_Z\}$. We consider the following $X$ or $Z$ noise process.

\paragraph{Single site dynamics.} Consider the set of single-qubit $X$ or $Z$ Paulis on $n$ qubits, $\mathcal{A} = \mathcal{A}_X\cup \mathcal{A}_Z$:
\begin{equation}
    \mathcal{A}_X = \big\{X_i\otimes \mathbb{I}_{[n]\setminus i}: i\in [n]\big\} \quad \mathcal{A}_Z = \big\{Z_i\otimes \mathbb{I}_{[n]\setminus i}: i\in [n]\big\}
\end{equation}

\noindent The jump operators in the frequency basis then fall into two types:
\begin{equation}
    A_{i, \nu}^X = X_{i} P_{i, \nu}, \quad A_{i, \nu}^Z = Z_{i} R_{i, \nu}, 
\end{equation}

\noindent where $P_{i, \nu}$ (resp. $R_{i, \nu}$) belongs to the algebra spanned by the $Z$ checks $Z^{b_j}$ (resp. $X$ checks $X^{a_j}$) whose support contains the $i$th spin. The dissipative generator then has the following form

\begin{equation}
    \mathcal{L}_{\mathsf{SS}} = \mathcal{L}_x+\mathcal{L}_z, \text{ where }\mathcal{L}_x = \sum_{i\in [n]} \mathcal{L}_x^i, \quad \mathcal{L}_z = \sum_{i\in [n]} \mathcal{L}_z^i
\end{equation}

\noindent This Davies generator describes an elementary local (single-qubit) ``$X$ or $Z$" noise process. As previously discussed in \cref{definition:classical_glauber}, this model naturally dequantizes into two classical Markov chains: one for the $X$ errors, and one for $Z$ errors.

\paragraph{Block dynamics.} Before proceeding, we remark that there is a natural way to ``coarse-grain" the single-site process into a block-dynamics. Suppose the $n$ spins are arranged into some graph $\mathcal{G} = ([n], E)$. Given a spin $i\in [n]$, we consider a box of spins of radius $R$ from $i$:

\begin{equation}
    \mathsf{B}_i^R = \big\{j\in [n]: d_{\mathcal{G}}(i, j) \leq R\big\}
\end{equation}

The Coarse-Grained/Block Dynamics will consist of quasi-local box updates, which ``resample" entire regions $\mathsf{B}_i^R$ conditioned on their boundaries. To formalize this process, consider the restriction of the single-site dynamics around a point $i$:
\begin{equation}
    \mathcal{L}^{ \mathsf{B}_i^R}_x = \sum_{j\in  \mathsf{B}_i^R} \mathcal{L}_x^j, \quad  \mathsf{E}_x^{^{ \mathsf{B}_i^R}}[\cdot] = \lim_{t\rightarrow \infty} \bigg(\exp\bigg[ \mathcal{L}^{ \mathsf{B}_i^R}_x t\bigg] \otimes \mathbb{I}_{[n]\setminus \mathsf{B}_i^{R+1}}\bigg)[\cdot],
\end{equation}
\noindent (and similarly for the $z$ terms). The channel $\mathsf{E}_x^{^{ \mathsf{B}_i^R}}[\cdot]$ is the conditional expectation map, and is analogous to the classical heat-bath block-dynamics generator. We remark that these block operators are quasi-local in the sense that they act on regions of radius $(R+1)$.

We can now re-phrase the continuous-time block-dynamics, in terms of a convex combination of the discrete conditional expectation maps:

\begin{equation}
    \mathcal{L}_{\mathsf{Block}}[\rho] = \sum_{i\in [n]} \bigg(\mathsf{E}_x^{^{ \mathsf{B}_i^R}}[\rho] - \rho\bigg)+\sum_{i\in [n]} \bigg(\mathsf{E}_z^{^{ \mathsf{B}_i^R}}[\rho] - \rho\bigg)
\end{equation}

See also \cref{fact:cont_disc_quantum} for a generalization of this discrete-to-continuous time conversion. 

\subsection{From discrete to continuous time}

As previously alluded to in \cref{section:preliminaries}, in our analysis we often interchange between continuous and discrete time (quantum and classical) Markov processes. In the following we exhibit a simple fact which allows us to convert mixing time bounds between the two.

\begin{fact}\label{fact:cont_disc_quantum}
    Given a completely positive, trace preserving quantum channel $\mathcal{N}:\mathcal{B}(\mathcal{H})\rightarrow \mathcal{B}(\mathcal{H})$, there exists a generator $\mathcal{L}:\mathcal{B}(\mathcal{H})\rightarrow \mathcal{B}(\mathcal{H})$ for a continuous time Lindbladian dynamics of the form 
    \begin{equation}
        \mathcal{L}[\sigma]  = \mathcal{N}[\sigma]-\sigma.
    \end{equation}
    Moreover, $\mathcal{L}$ and $\mathcal{N}$ share the same fixed points, and their mixing times can be related via \cref{fact:cont_disc}.
\end{fact}

\end{document}